\documentclass{article}
\usepackage{arxiv}
\usepackage{hyperref} 
\usepackage{amsmath} 
\usepackage{amssymb}
\usepackage{graphicx}
\usepackage{epstopdf}

\usepackage{booktabs}
\usepackage{threeparttable}
\usepackage{algorithm}
\usepackage{algorithmic}
\usepackage{xcolor} 
\usepackage{indentfirst}
\usepackage{hypernat}
\usepackage{appendix}

\usepackage{lipsum} 

\usepackage[sc]{mathpazo} 
\usepackage[T1]{fontenc} 
\usepackage{microtype} 

\usepackage{booktabs} 
\usepackage{float} 

\usepackage{lettrine} 

\usepackage{amsthm}
\usepackage{multirow,array}
\usepackage{caption}
\usepackage{subcaption}
\newtheorem{rem}{Remark}

\newtheorem{theorem}{Theorem}
\newtheorem{conjecture}{Conjecture}

\newcommand{\Ret}{\text{Ret}}
\newcommand{\DD}{\mathbb{D}}
\newcommand{\MM}{\mathbb{M}}

\title{Second Order Unconditional Positive Preserving Schemes for Non-equilibrium Reacting Flows with Mass and Mole Balance} 
\author{
Jianhua Pan\\
William. G. Lowrie Department \\
of Chemical and Biomolecular Engineering,\\
Koffolt Labs, The Ohio State University\\
 \texttt{pan.796@osu.edu}\\
\And
Yu-yen Chen\\
William. G. Lowrie Department \\
of Chemical and Biomolecular Engineering,\\
Koffolt Labs, The Ohio State University\\
\texttt{chen.5797@buckeyemail.osu.edu}\\
\And 
Liang-Shih Fan\\
William. G. Lowrie Department \\
of Chemical and Biomolecular Engineering,\\
Koffolt Labs, The Ohio State University\\
\texttt{fan.1@osu.edu}\\
}
\date{}
\begin{document}
\maketitle

\begin{abstract}
    In this study, a family of second order process based modified Patankar Runge-Kutta schemes is proposed with both the mass and mole maintained in balance while preserving the positivity of density and pressure with the time step determined by convection terms. The accuracy analysis is conducted to derive the sufficient and necessary conditions for the Runge-Kutta and Patankar coefficients. Coupling with the finite volume method, the proposed schemes are extended to Euler equations with non-equilibrium reacting source terms. Benchmark tests are given to prove the prior order of accuracy and validate the positive-preserving property for both density and pressure.
\end{abstract}


\section{Introduction}
Reactive non-equilibrium flows are ubiquitous in a broad range of disciplines of science and engineering. Typical examples of such phenomena are given by the reaction engineering problems, where the fluid flow transports with the occurring of chemical reactions , such as the multiphase fluidized bed reactors \cite{zhou2013syngas,zhou2014second} and combustion engines \cite{lieuwen2012unsteady}. Besides reaction engineering, further examples can be found in geochemistry \cite{burchard2005application}, meteorology \cite{kuhn1997one}, etc. The one-dimensional model for such non-equilibrium flows can be expressed by a system of reactive Euler equations, which consists a set of hyperbolic equations with source terms,
\begin{equation}
    U_t + F(U)_x = S(U)\,,
    \label{eq:euler1}
\end{equation}
where
\begin{equation}
    \begin{aligned}
        U & = \left(
        \rho\,, m\,, E\,, \rho Y_1\,,\rho Y_2,\cdots,\rho Y_N
        \right)^T\,,\\
        F(U) &= \left(m, \rho u^2 +p, (E+p) u, \rho u Y_1, \rho u Y_2,\cdots \rho u Y_N\right)^T\,,\\
        S(U) & = \left(0,0,0,\omega_1,\omega_2,\cdots,\omega_N\right)^T\,,
    \end{aligned}
    \label{eq:euler1detailed}
\end{equation}
and
\begin{equation}
    \begin{aligned}
    &\sum_{i=1}^N{Y_i} = 1,\,\,m = \rho u,\\
    &E = \frac{1}{2} \rho u^2 + {\left[\rho Y_i e_i(T) + \rho Y_i h^0_i\right]}\,,\\
    \end{aligned}
    \label{eq:previousE}
\end{equation}
where $\rho$ is the density for the mixture of the gas; $Y_i$ is the mass fraction for $i$-th specie; $p$ is the pressure, and $T$ is the temperature. Additionally, $h_i^0$ is the energy corresponding to the chemical bonds and is constant. The internal energy $e_i$ and the pressure $p$ are correlated by the state of equation,
\[
    \rho Y_i e_i = \rho e = \frac{p}{\gamma-1}\,,
\]
where $e$ and $\gamma$ are the specific internal energy and specific heat ratio denoted for the mixture, respectively. The source terms satisfy the property of conservation, i.e., $\sum_{i=1}^N{\omega_i} = 0$. Note that due to the conservation property in Eq. (\ref{eq:euler1}), one relation among the $N+1$ mass equations for both the mixture and the $N$ species is not independent, which is presented here just for completeness. Also note that the Einstein summation rule is applied to $\rho Y_ie_i$ in Eq. (\ref{eq:previousE}). In the remainder of this paper, unless otherwise noted, the Einstein notation is applied to all subscripts which appear more than once in the same term and such an index is noted as the dummy index. When a dummy index is enclosed by a bracket, it is not applied by the Einstein summation rule, e.g., $\rho Y_{(i)} e_{(i)}$ which simply refers to the internal energy per volume for specie $i$.

One of the main difficulties in solving systems of Eq. (\ref{eq:euler1}) is how to preserve the density, pressure and temperature in a physically bounded space. This difficult is especially significant with the presence of source terms, which greatly increase the stiffness of Eq. (\ref{eq:euler1}). For reactions that are constituted by a complex reaction network wherein multiple species and reaction steps are involved, the disparate timescale of the reactions could greatly increase the stiffness of the system.

A general framework which is with high order of accuracy and is able to keep the positivity of hyperbolic systems has been proposed by the series work of Shu, Zhang, Xing and their collaborators \cite{zhang2010positivity,zhang2010maximum,xing2010positivity,zhang2012positivity,zhang2012maximum,zhang2013maximum}. This technique has been successfully applied to shallow water equations \cite{xing2010positivity}, convection-diffusion equations \cite{zhang2013maximum}, Navier-Stokes equations \cite{zhang2012positivity}, etc., in the context of Discontinuous Galerkin (DG) method, finite volume method, and finite difference method. The basic idea of this framework is utilizing the convex combination in both spatial and temporal domains. Firstly, one need to prove the positive-preserving property of a basic forward Euler scheme combing with a well-designed slope limiter which is applied to keep both the accuracy and boundedness of the space distribution of variables. Secondly, the strong-stability-preserving (SSP) Runge-Kutta (RK) scheme which can be expressed by a convex combination of the forward Euler scheme is utilized to realize the high-order of accuracy in time. In \cite{zhang2010positivity}, Zhang and Shu extended the idea to hyperbolic systems with reactions under the framework of DG method, which is however suffered from the limited timestep when the source terms become extremely stiff.

Several schemes have already been proposed to relax the constraints on the size of time step and preserve the positivity of solutions when solving stiff equations, including the implicit-explicit (IMEX) \cite{chertock2015steady,chertock2015well,huang2017second,hu2018asymptotic, lee2020stability} RK method, the exponential integration method \cite{pope1963exponential,huang2018bound,Du2019}, and the Patankar method \cite{patankar2018numerical, burchard2003high,Bruggeman2007,Broekhuizen2008,Huang2019a,Huang2019}. In IMEX-RK method, the convection part and the stiff source term are discretized in an explicit and implicit manner, respectively. The implicit treatment of stiff source terms eliminates the constraints on time step and preserves the positivity of desired variables. Chertock, Cui, Kurganov and Wu \cite{chertock2015steady} developed a family of second order sign-preserving IMEX-RK methods for systems of ordinary differential equations (ODEs) with a stiff damping term. Following this approach, the IMEX-RK method was applied to shallow water equations \cite{chertock2015well}, the Kerr-Debye relaxation system \cite{huang2017second}, the stiff BGK equations \cite{hu2018asymptotic}, the Cahn-Hilliard equation \cite{lee2020stability}, etc. However, the main drawback of the IMEX-RK method is that it is difficult to be extended to systems with multiple destructions and productions while maintaining the positivity of the desired variables in a conserved and consistent way. The idea of exponential integration method can be dated back to 1960s \cite{pope1963exponential}. In this method, the stiff source term is integrated accurately with the aid of exponential function, which also retains the positivity of desired variables. Recent work regarding the hyperbolic systems comprises those of Huang and Shu \cite{huang2018bound} and Du and Yang \cite{Du2019}. In Huang and Shu \cite{huang2018bound}, a family of bound-preserving modified exponential RK-DG schemes was developed to solve scalar hyperbolic equations with stiff source terms. In Du and Yang \cite{Du2019}, a third-order conservative positive-preserving and steady-state-preserving exponential integration scheme was developed which was applied to multi-species and multi-reaction chemical reactive flows.

Similar to the IMEX method, the Patankar method introduces an implicit modification to the stiff source terms to achieve the unconditional positivity \cite{patankar2018numerical}; note that the "unconditional" implying that the stiff source terms do not impose constraints to the time step. Several modified Patankar methods have been further proposed to ensure the conservation \cite{burchard2003high,Bruggeman2007,Broekhuizen2008}, to improve the order of accuracy \cite{kopecz2018unconditionally,offner2020arbitrary}, and to integrate with SSP-RK schemes \cite{Huang2019,Huang2019a} which are popular in the simulation of hyperbolic systems. In Euler equations with multiple chemical reactions, both the mass conservation of the species and the mole balance of each involved element must be preserved to obtain physically reasonable solutions. The "mole balance of each involved element" denotes that the ratio of the destruction or production rate (with unit $\mathrm{mol/(m^3\cdot s)}$) of different species in a reaction must equal to the ratio of the corresponding stoichiometry coefficients. However, for some versions of the Patankar-type schemes \cite{Huang2019,Huang2019a}, only the mass conservation is kept while the mole balance of the involved elements is violated, since the systems are expressed in a pair-wise production-destruction (PD) way. For some versions of Patankar schemes, e.g., \cite{Bruggeman2007, Broekhuizen2008}, to keep both the mass and mole balance, a common Patankar modification is applied to all the reactions. However, with such a common Patankar coefficient, a very stiff source term would cause all other reactions being reduced to almost zero even they are totally irrelevant. Additionally, in the Euler system with non-equilibrium reactions, the positivity of pressure cannot be guaranteed \cite{Huang2019a,Huang2019} when the reaction is endothermic for all the available Patankar-type schemes.

In this study, we construct a second order Patankar RK scheme which is able to keep both the mass and mole balance, while preserving the positivity of pressure and is applicable to the hyperbolic systems. To achieve the mass and mole balance simultaneously, a process based modified Patankar RK scheme (PMPRK) is proposed herein. Different from the schemes proposed by Huang, Zhao and Shu \cite{Huang2019,Huang2019a} for PD systems, the Patankar modifications in this study are applied to reaction by reaction. Thus, the mass conservation and mole balance are maintained at the same time. An accuracy analysis is conducted to determine the constraints for the coefficients of the SSP-RK schemes and the detailed form of the Patankar coefficients. The positivity of the proposed scheme is validated through a large number of numerical samples. Several theoretical proofs are given under simplified conditions. To keep the positivity of the pressure in the Euler system, a simple but an efficient technique is introduced. At the end, several benchmark tests are given to validate the proposed Patankar scheme in Sec. \ref{sec:numerical_examples} and conclusions are provided in Sec. \ref{sec:conclusion}.

\section{Second Order Process Based Modified Patankar RK Scheme}
In order to illustrate the idea of PMPRK scheme, consider the following ordinary differential equations with only source terms,
\begin{equation}
    \begin{aligned}
    &\frac{\mathrm{d} c_i}{\mathrm{d} t} = \frac{\omega_{(i)}}{M_{(i)}} = \sum_{k=1}^K{\frac{\omega_{k(i)}}{M_{(i)}}} = R_k \lambda_{ki}\,,\\
    &\lambda_{ki}M_i = 0\,,
    \end{aligned}\label{eq:ode}
\end{equation}
where $c_i$ is the mole concentration for the $i$-th species and $c_i = \frac{\rho Y_{(i)}}{M_{(i)}}$, $\omega_{ki}$ is the mass variance (production or destruction) for the $i$-the specie due to reaction $k$, $R_k$ is the reaction rate for the $k$-th chemical reaction, $\lambda_{ki}$ is the stoichiometry coefficient for specie $i$ in reaction $k$, and $M_i$ is the mass per mole for specie $i$, with $i \in \{1,\cdots, N\}$ and $k \in \{1,2,\cdots, K\}$.  The mass conservation and mole balance can be expressed as
\begin{align}
&\sum_{i=1}^{N}{\omega_{ki}} = \lambda_{(k)i}M_i R_{(k)} = 0 \,,\label{eq:ode_mass_conservation}\\
&\frac{\omega_{(k)i}}{\omega_{(k)j}} = \frac{\lambda_{(k)(i)}M_{(i)}}{\lambda_{(k)(j)}M_{(j)}}\,,\forall \, i,j \in \left\{1,2,\cdots,N\right\}\,,\label{eq:ode_mole_balance}
\end{align}
$\forall\, k \in \left\{1,2,\cdots,K\right\}$. 

The explicit second order SSP-RK scheme in the Shu-Osher form for Eq. (\ref{eq:ode}) is
    \begin{align}
        c_i^{(0)} & = c_i^n\,,\\
        c_i^{(1)} & = \alpha_{10}c_i^{(0)} + \Delta t \beta_{10}R_k^{(0)} \lambda_{ki}\,,\\
        c_i^{(2)} & = \alpha_{20}c_i^{(0)} +\alpha_{21}c_i^{(1)} + \Delta t \left(\beta_{20}R_k^{(0)} +\beta_{21}R_k^{(1)}\right)\lambda_{ki}\,,\\
        c_i^{n+1} & = c_i^{(2)}\,,
    \end{align}
where $\alpha_{10} = 1$ should be firstly satisfied, see e.g. Hairer et al. \cite{hairer1991solving}. In the following discussion, to make the equations more compact, $\alpha_{10} = 1$ is applied implicitly. To apply the positive-preserving technique proposed by Gottlieb and Shu \cite{gottlieb2011strong}, all th e coefficients of $\alpha$ and $\beta$ should be positive such that $c_i^{(k)}\,,k=1,2$ can be written as a convex combination of forward Euler schemes.

In order to keep the positivity for concentrations $c_i^{(k)},\,k=0,1,2$, the Patankar modification is applied to the source terms. Moreover, to keep the mole balance, the Patankar modification to the reaction rate should be process based, i.e., each $R_k$ adopts a same Patankar weight in all specie equations as follows,
    \begin{align}
        c_i^{(0)} & = c_i^n\,,\\
        c_i^{(1)} & = c_i^{(0)} + \Delta t \beta_{10}R_k^{(0)} \lambda_{ki}\chi_k^{(1)}\,,\label{eq:MPRK_s1}\\
        c_i^{(2)} & = \alpha_{20}c_i^{(0)} +\alpha_{21}c_i^{(1)} + \Delta t \left(\beta_{20}R_k^{(0)} +\beta_{21}R_k^{(1)}\right)\lambda_{ki}\chi_k^{(2)}\,,\label{eq:MPRK_s2}\\
        c_i^{n+1} & = c_i^{(2)}\,,
    \end{align}
where $\chi_k^{(1)}$ and  $\chi_k^{(2)}$ are functions defined by
\[
    \begin{aligned}
    \chi_k^{(1)}& = \chi_k^{(1)}\left(c_i^{(0)}, c_i^{(1)}| i = 1,2,\cdots N\right)\,,\\
\chi_k^{(2)} & = \chi_k^{(2)}\left(c_i^{(0)}, c_i^{(1)}, c_i^{(2)}| i = 1,2,\cdots N\right)\,.
    \end{aligned}
\] 
The introduction of $\chi_k^{(j)}\,,j=1,2$ makes Eqs. (\ref{eq:MPRK_s1}-\ref{eq:MPRK_s2}) implicit, and in turn makes it possible for the positive preserving. Additionally, since the Patankar coefficients are process based, for the $k$-th reaction, the numerical source terms satisfy
\begin{align} 
    &\sum_{i=1}^{N}{\omega_{ki}} = \lambda_{(k)i}M_i R_{(k)}\chi_{(k)}^{(s)} = \left(\lambda_{(k)i}M_i\right) R_{(k)}\chi_{(k)}^{(s)}  = 0\,,s=1,2\,,\\
    &\frac{\omega_{(k)i}}{\omega_{(k)j}} = \frac{\lambda_{(k)(i)}M_{(i)}R_{(k)}\chi_{(k)}^{(s)}}{\lambda_{(k)(j)}M_{(j)}R_{(k)}\chi_{(k)}^{(s)}} = \frac{\lambda_{(k)(i)}M_{(i)}}{\lambda_{(k)(j)}M_{(j)}}\,,s=1,2\,,\forall \,i,j\in\left\{1,2,\cdots,N\right\}\,,
\end{align}
which is in accordance with Eqs. (\ref{eq:ode_mass_conservation}-\ref{eq:ode_mole_balance}), i.e., the mass conservation and mole balance are kept. 

Let us first discuss the sufficient and necessary conditions for the PMPRK scheme to achieve the prior order of accuracy, and the detailed expressions for the Patankar coefficients $\chi_k^{(j)}\,,j=1,2$ will be provided later. Since the proposed PMPRK is just a minor adjustment of the original scheme, the truncation error of the modification coefficient $\chi_k^{(1)}$ and $\chi_k^{(2)}$ is assumed to have the following form,
\begin{eqnarray}
\chi_k^{(1)}  & = 1 + X_k \Delta t + O({\Delta t}^2)\,,\label{eq:chi1}\\
\chi_k^{(2)}  & = 1 + Y_k \Delta t + O({\Delta t}^2)\,.\label{eq:chi2}
\end{eqnarray}
To fulfill the requirements for the order of accuracy, there should exist certain constraints on $X_k$ and $Y_k$.

Substituting Eq. (\ref{eq:chi1}) into Eq. (\ref{eq:MPRK_s1}), we obtain
\begin{equation}
        c_i^{(1)} = c_i^{(0)} + \Delta t \beta_{10}R_k^{(0)} \lambda_{ki}(1 + X_k \Delta t )+ O({\Delta t}^3)\,.
        \label{eq:expand_c1}
\end{equation}
Meanwhile, expand the expression of $c_i^{(2)}$ near time $t^{(0)}$ and substitute Eqs. (\ref{eq:chi2}) into it,
\begin{equation}
    \begin{aligned}
        c_i^{(2)} 
        & = \,\alpha_{20}c_i^{(0)} +\alpha_{21}c_i^{(1)} \\
        & + \Delta t \left\{\beta_{20}R_k^{(0)}
         +\beta_{21}\left[R_k^{(0)}
        + R_{k,j}^{(0)}(c_j^{(1)} - c_j^{(0)})
        \right]\right\}\lambda_{ki}
        \cdot(1 + Y_k \Delta t ) + O({\Delta t}^3)\,,
    \end{aligned}
        \label{eq:expand_c2}
\end{equation}
where the subscript after the comma in $R_{k,j}$ denotes the index of the specie to which the operation of partial derivative is applied to, i.e.,
\begin{equation}
\begin{aligned}
R_{k,j} = \frac{\partial R_k}{\partial c_j}\,.
\end{aligned}
\end{equation}
Substitute Eq. (\ref{eq:expand_c1}) into Eq. (\ref{eq:expand_c2}), the final expression of $c_i^{(2)}$ can be obtained.

On the other hand, the Taylor expansion of exact solution $\widetilde{c}_i$ near $c_i^{(0)}$ is
\begin{equation}
    \begin{aligned}
    \widetilde{c}_i & = c_i^{(0)} +R_k \lambda _{ki}\Delta t + \frac{1}{2} R_{k,j}R_m\lambda _{mj}\lambda _{ki}\Delta t^2 
    + O(\Delta t^3)\,.
    \end{aligned}
\end{equation}
Subtracting $\widetilde{c}_i$ from $c_i^{(2)}$, the following equation is obtained,
\begin{equation}
    \begin{aligned}
    c_i^{(2)} - \widetilde{c}_i & = \theta_1 c_i^{(0)} + \theta_2 \Delta t + \theta_3 \Delta t^2  + O(\Delta t^3)\,,
    \end{aligned}
\end{equation}
where
\begin{equation}
    \begin{aligned}
    \theta_1 & = \alpha _{20}+\alpha _{21} - 1\,,\\
\theta_2 & = \left[ (\alpha _{21}  \beta _{10}+\beta _{20}+\beta _{21})-1\right]R_k \lambda_{ki}\,,\\
    \theta_3 &  = \left[\alpha _{21} \beta _{10} X_k+\left(\beta _{20}+\beta _{21}\right) Y_k\right]R_k \lambda _{ki}+ \left(\beta _{10} \beta _{21}-\frac{1}{2}\right)R_m \lambda _{m,j} R_{k,j} \lambda _{k,i}\,.
    \end{aligned}
\end{equation}

In order to achieve a prior second order of accuracy, $\theta_k,\,k=1,2,3$ should all be $0$. Thus, the sufficient and necessary conditions for the second order PMPRK scheme is
\begin{equation}
    \begin{aligned}
        & \alpha_{10} = 1,\, \alpha_{20} + \alpha_{21} = 1,\,\\
        & \alpha _{21}  \beta _{10}+\beta _{20}+\beta _{21} = 1,\,\beta_{20}\beta_{21} = \frac{1}{2}\,,\\
    \end{aligned}
    \label{eq:rk_order_requirements0}
\end{equation}
\begin{align}
    &\alpha _{21} \beta _{10} X_k+\left(\beta _{20}+\beta _{21}\right) Y_k =0\,.
    \label{eq:rk_order_requirements1}
\end{align}

Equation (\ref{eq:rk_order_requirements0}) is the order requirements of the original second order RK scheme. Equation (\ref{eq:rk_order_requirements1}) is for the Patankar modification and is consistent with the form derived in Huang and Shu \cite{Huang2019a} for the PD systems.

\section{Construction of $\chi_k^{(1)}$ and $\chi_k^{(2)}$}
The key novelty of this work is that the Patankar modification in the proposed PMPRK scheme is process based and applied to the source terms reaction by reaction, and thus should be the same for all the species involved in the same reaction.

An intuitive construction of $\chi_k^{(1)}$ and $\chi_k^{(2)}$ is
    \begin{align}
        \chi_k^{(1)} & = \left(\prod_{\varepsilon\in \Ret(k)}\left.\frac{c_\varepsilon^{(1)}}{\pi_\varepsilon}\right.\right)^{q_k}\,,\label{eq:chi1_detailed}\\
        \chi_k^{(2)} & = \left(\prod_{\varepsilon\in \Ret(k)}\left.\frac{c_\varepsilon^{(2)}}{\tau_\varepsilon}\right.\right)^{q_k}\,,\label{eq:chi2_detailed}
    \end{align}
where $\Ret(k)$ denotes the set of reactants (the species which are destructed) in the $k$-th reaction; $q_k$ is a real number locates in $(0,+\infty)$. In this study, it is assumed that a specie cannot be both a reactant and a production. If a specie does show up in both side of the reaction, the side with smaller stoichiometry coefficient is to be cancelled. $\pi_\varepsilon$ is a function of $c_\varepsilon^{(0)}$; $\tau_\varepsilon$ is a function of $c_\varepsilon^{(0)}$, $ \pi_\varepsilon$ and $c_\varepsilon^{(1)}$. 

The form of $\chi_k^{(j)},j=1,2$ only involves the reactants of reaction $k$; therefore, if the destructed reactant approaches $0$, $\chi_k^{(j)}, j=1,2$ approaches zero, thereby playing as a scaling factor for reaction $k$. As will be shown later, numerical investigation of both randomly generated samples and the benchmark tests both demonstrate the effectiveness of the proposed form of $\chi_k^{(j)},j=1,2$.

Before discussing the positivity of the scheme, let us first determine the expression of $\pi_\varepsilon$ and $\tau_\varepsilon$ which can preserve the prior order of accuracy. The most straightforward design of $\pi_\varepsilon$ is like that of the work in Huang and Shu \cite{Huang2019a} and Huang et al.\cite{Huang2019},
\begin{equation}
    \pi_\varepsilon = c_\varepsilon^{(0)}\,.
\end{equation}
Thus, we have
\begin{equation}
    \chi_k^{(1)} = 1 + O(\Delta t)\,.\label{eq:zero_order_truncation}
\end{equation}
The substitution of Eq. (\ref{eq:zero_order_truncation}) into Eq. (\ref{eq:MPRK_s1}) yields
\begin{equation}
    c_i^{(1)} = c_i^{(0)} + \Delta t \beta_{10} R_k \lambda_{ki} + O(\Delta t^2)\,.
    \label{eq:ci1_dt}
\end{equation}
Then, substituting Eq. (\ref{eq:ci1_dt}) into Eq. (\ref{eq:chi1_detailed}), the first order expansion of $\chi_k^{(1)}$ becomes
\begin{equation}
    \chi_k^{(1)} = \left(\prod_{\varepsilon \in \Ret(k)}\frac{c_\varepsilon^{(1)}}{c_\varepsilon^{(0)}}\right)^{q_k} = 1 + \Delta t q_k \beta_{10} \sum_{\varepsilon \in \Ret(k)}{\frac{R_m \lambda_{m\varepsilon}}{c_\varepsilon^{(0)}}} + O(\Delta t^2) \,.
    \label{eq:xi1_dt}
\end{equation}
Thus, $X_k$ is expressed as
\begin{equation}
    \begin{aligned}
    X_k & = q_k \beta_{10} \sum_{\varepsilon \in \Ret(k)}{\frac{R_m \lambda_{m\varepsilon}}{c_\varepsilon^{(0)}}}\,.
    \end{aligned}
    \label{eq:xkuk}
\end{equation}
From Eq. (\ref{eq:rk_order_requirements1}), $Y_k$ is obtained as
\begin{equation}
    Y_k = -\frac{\alpha_{21} \beta_{10}}{\beta_{20} + \beta_{21}} q_k \beta_{10} \sum_{\varepsilon \in \Ret(k)}{\frac{R_m \lambda_{m\varepsilon}}{c_\varepsilon^{(0)}}}\,.
\end{equation}
Since $\Ret(k)$ is an arbitrary set containing the indices of the reactants, we have
\begin{align}
    & \frac{c_i^{(2)}}{\tau_i} = 1-\frac{\alpha_{21} \beta_{10}}{\beta_{20} + \beta_{21}} \beta_{10} {\frac{R_k \lambda_{ki}}{c_i^{(0)}}}\Delta t + O(\Delta t^2)\,.
\end{align}
Thus, the first order expansion of $\frac{\tau_i}{c_i^{(0)}}$ is
\begin{align}
     \frac{\tau_i}{c_i^{(0)}} & = \frac{c_i^{(2)}/c_i^{0}}{c_i^{(2)}/\tau_i} = \frac{1+ \frac{R_k \lambda_{ki}}{c_i^{(0)}} \Delta t +O(\Delta t^2)}{1-\frac{\alpha_{21} \beta_{10}}{\beta_{20} + \beta_{21}} \beta_{10} {\frac{R_k \lambda_{ki}}{c_i^{(0)}}}\Delta t + O(\Delta t^2)}\nonumber\\
     & = 1 + \left(1+\frac{\alpha_{21} \beta_{10}}{\beta_{20} + \beta_{21}} \beta_{10}\right) {\frac{R_k \lambda_{ki}}{c_i^{(0)}}}\Delta t + O(\Delta t^2)\,.
\end{align}
Like the work in Huang and Shu \cite{Huang2019a}, we can construct the $\tau_i$ in the form as the following equation,
\begin{equation}
    \frac{\tau_i}{c_i^{(0)}} = \left(\frac{c_i^{(1)}}{c_i^{(0)}}\right)^s\,.
\end{equation}
Then the parameter $s$ is determined as
\begin{equation}
    s = \left(\frac{1}{\beta_{10}}+\frac{\alpha_{21} \beta_{10}}{\beta_{20} + \beta_{21}} \right)\,.\label{eq:s}
\end{equation}
In fact, $\tau_i$ can be constructed in any form which can recover its first order expansion and can be expressed as a convex combination of several positive terms. Equation (\ref{eq:s}) is consistent with the form in Huang and Shu \cite{Huang2019}. But the analysis shown in this study reveals that Eq. (\ref{eq:s}) does not only apply for the linear Patankar coefficients as in the PD systems, but also suitable for the nonlinear Patankar coefficients as proposed in Eqs. (\ref{eq:chi1_detailed}-\ref{eq:chi2_detailed}).


\begin{rem}
    The analysis shows that the real number $q_k$ in Eqs. (\ref{eq:chi1_detailed}) and (\ref{eq:chi2_detailed}) does not affect the prior order of accuracy and is thus a free parameter. However, according to the Taylor expansion of $\chi_k^{(1)}$ and $\chi_k^{(2)}$, a larger $q_k$ usually corresponds to a larger truncation error. Additionally, through numerical investigations in Sec. \ref{sec:general_positivity} with randomly generated samples, with the form of constant $q$ for all $\chi_k$, e.g., $q=1$, if the number of reactants increases, the cut-off error due to the limited bytes of floating numbers makes the calculations of Patankar coefficients difficult and makes it hard to converge for the Newton iterations, which is utilized to solve the implicit function of Eqs. (\ref{eq:MPRK_s1}-\ref{eq:MPRK_s2}). Thus, in the current work, we choose $q_k = \frac{1}{\#\Ret(k)}$, where $\#\Ret(k)$ denotes the number of reactants for the $k$-th reaction.
\end{rem}

\section{Positivity of the PMPRK Scheme}
\label{sec:positivity_discussion}
For the PMPRK scheme, a system of implicit equations needs to be solved. Given only the destructed reactants are involved in the expression of $\chi_k^{(j)},j=1,2$; $c_i^{\left(j\right)},j=0,1$ and $\Delta tR_{\left(k\right)}\lambda_{\left(k\right)i}$ in Eqs. (\ref{eq:MPRK_s1}-\ref{eq:MPRK_s2}), $\pi_i$ and $\tau_i,\,i=1,2,\cdots,N$ in Eqs. (\ref{eq:chi1_detailed}-\ref{eq:chi2_detailed}) are all non-negative numbers, the solving of implicit Eqs. (\ref{eq:MPRK_s1}-\ref{eq:MPRK_s2}) is equivalent to solving Eq. (\ref{eq:conjecture}).
\begin{conjecture}
    \label{conjecture}
    Given $a_i$, $a_i^*$, $b_{ki}$ and $q_k$ as constants satisfying $a_i >0\,, a_i^* > 0\,, b_{ki} \geq 0$,  $q_k \in (0, +\infty)$, $c_i$ as unknowns and $i=1,2,\cdots,N;k=1,2,\cdots,K$, additionally,
\begin{equation}
    \sum_{i \in \Ret(k)}{b_{ki}} \geq \sum_{i \not \in \Ret(k)}{b_{ki}}\,.
    \label{eq:dominant}
\end{equation} 
    If a system of equations can be expressed in the form
    \begin{equation}
        c_i - a_i + \sum_{k, \Ret(k) \ni i} \left[b_{ki} \left(\prod_{j\in \Ret(k)}{\frac{c_j}{a_j^*}} \right)^{q_k}\right]- \sum_{k, \Ret(k) \not\ni i} \left[b_{ki} \left(\prod_{j\in \Ret(k)}{\frac{c_j}{a_j^*}}\right)^{q_k}\right] = 0\,,
        \label{eq:conjecture}
    \end{equation}
it has at least one solution. And among these solutions, there exists a unique solution satisfying $c_i > 0, i = 1,2,\cdots, N$.
\end{conjecture}
\begin{rem}
    Conjecture \ref{conjecture} is an extension of the linear system in \cite{Huang2019a}, and the constraint of Eq. (\ref{eq:dominant}) is introduced to be consistent with the linear system in \cite{Huang2019a} where a strict diagonal dominant m-matrix can be obtained. However, through the numerical tests, the constraint of Eq. (\ref{eq:dominant}) can be relaxed a little for the nonlinear system when the pressure is introduced as shown in Secs. \ref{sec:positivity_euler} and \ref{sec:1dOxygen}.
\end{rem}
\begin{rem}
    The density of the reactants can be zero, i.e., $a_i$ and $a_i^*$ are zero, which leads the corresponding reaction rate to be also zero and thus be decoupled from Eq. (\ref{eq:conjecture}). Therefore, only the system with strict positive densities for all species are discussed in this section.
\end{rem}

The correctness of Conjecture \ref{conjecture} is unknown. However, if we take a glance at Eq. (\ref{eq:conjecture}), it can be found that Eq. (\ref{eq:conjecture}) can be treated as a first order backward Euler discretization with monotone source terms; thus Conjecture \ref{conjecture} is assumed to be correct in this study. In this section, we give two conclusions under simplified conditions. For more general cases, the conjecture is validated with exhausted numerical samples in the parameter space of $a_i$, $a_i^*$ and $b_{ki}$. In the remainder of this paper, if all the elements of a vector or a matrix is greater than zero, or greater or equal than zero, it is denoted as $ \vec{c} > 0,\, \mathbb{M} > 0$ and $\vec{c} \geq 0,\, \mathbb{M} \geq 0$, respectively.
\subsection{Positivity under Simplified Conditions}
\begin{theorem}
    If each of the reactions has only one reactant, then the reactions can be expressed by the indices of the species and denoted as $q_k = q_i,\,\forall\,\, \Ret(k) \ni i$. Furthermore, Eq. (\ref{eq:conjecture}) has at least one solution, and there exists an one and only solution satisfying $\vec{c} > 0$ among all the solutions.
    \label{theorem:1}
\end{theorem}
\begin{proof}
    After a variable substitution of $v_i = c_i^{q_i}$, Eq. (\ref{eq:conjecture}) can be rewritten in the following form,
    \begin{equation}
        \vec{v}^{\frac{1}{q}} + \MM \, \vec{v} = \vec{a},\,
        \label{eq:fixed_proof_1}
    \end{equation}
    where $\vec{v}^{\frac{1}{q}} = \left\{v_1^{\frac{1}{q_1}}, v_2^{\frac{1}{q_2}},\cdots, v_N^{\frac{1}{q_N}}\right\}^T$ and
    \begin{equation}
        \begin{aligned}
        & \MM_{ii}  = \sum_{k, \Ret(k) \ni i}{\frac{b_{ki}}{{a_i^*}^{q_i}}}\,,\quad\quad
         \MM_{ij}  = -\sum_{k, \Ret(k) \ni i}{\frac{b_{kj}}{{a_i^*}^{q_i}}}\,, i\neq j,\,\\
        & \vec{v}  = \left\{v_1, v_2, \cdots, v_N\right\}^T\,,\quad\quad
        \vec{a} = \left\{a_1, a_2, \cdots, a_N\right\}^T\,.
        \end{aligned}
        \label{eq:M_Matrix}
    \end{equation}
    Note that the Einstein summation rule does not apply to the subscript of $\MM_{ii}$ in Eq. (\ref{eq:M_Matrix}) and all the other expressions in the remainder of this section. According to Eq. (\ref{eq:dominant}) and considering $i$ is the only reactant in $\Ret(k)$, ${b_{ki}} \geq \sum_{j \neq i}{b_{kj}}$. Thus $\MM$ is a column diagonally dominant $m$-matrix. Then we introduce a small number $\epsilon > 0$ and rearrange Eq. (\ref{eq:fixed_proof_1}) as
    \begin{equation}
        \begin{aligned}
            \left(\epsilon \mathbb{I+D}\right)^{-1}\left(\vec{v}^{\frac{1}{q}}+\mathbb{D}\vec{v}\right) = \left(\epsilon \mathbb{I+D}\right)^{-1}\left(\mathbb{R} \vec{v}+\vec{a}\right)\,,
        \end{aligned}
    \end{equation}
    where $\mathbb{I}$ is an identity matrix, $\mathbb{D}$ is the diagonal part of $\MM$, $\MM = \mathbb{D-R}$ and $\mathbb{D}\geq 0,\,\mathbb{R}\geq 0$. $\forall\,\, \vec{v} > 0$, we have $\left(\epsilon \mathbb{I+D}\right)^{-1}\left(\mathbb{R} \vec{v}+\vec{a}\right) > 0$.

    Denote
    \begin{equation}
\vec{F}(\vec{v}) = \left(\epsilon \mathbb{I+D}\right)^{-1}\left(\vec{v}^{\frac{1}{q}}+\mathbb{D}\vec{v}\right)\,,
    \end{equation}
    and since $\epsilon \mathbb{I+D}$ is a diagonal matrix, it is easy to verify that, in the first phase, $\forall\,\,\vec{f} > 0$, $\vec{F}(\vec{v}) = \vec{f}$ has one and only one solution. Thus, we denote the projection $\vec{F}^{-1}: \vec{f} \rightarrow \vec{v},\,\vec{v} > 0$ and define the following fixed-point iteration scheme
    \begin{equation}
        \begin{aligned}
             \vec{f} & = \left(\epsilon \mathbb{I+D}\right)^{-1}\left(\mathbb{R} \vec{v}^n+\vec{a}\right)\,,\\
             \vec{v}^{n+1} & = \vec{F}^{-1}(\vec{f})\,.
        \end{aligned}
        \label{eq:fixed_iteration}
    \end{equation}
    Since $\epsilon > 0$ and $\MM$ is a column diagonally dominant $m$-matrix, $\epsilon \mathbb{I} +\MM$ is a strict column diagonally dominant $m$-matrix, the spectral radius of $\left(\epsilon \mathbb{I+D}\right)^{-1}\mathbb{R} $ is smaller than 1, and there must exist a $p$-norm which makes $\left\|\left(\epsilon \mathbb{I+D}\right)^{-1}\mathbb{R}\right\|_p < 1$. Hence, for any given $\vec{v}_1^n > 0$ and $\vec{v}_2^n > 0$, we have
    \[
        \vec{f}_1 > 0,\,\,\vec{f}_2 > 0\,,  
    \]
    and
    \[
        \left\|\vec{f}_1 - \vec{f}_2\right\|_p \leq \left\|\left(\epsilon \mathbb{I+D}\right)^{-1}\mathbb{R}\right\|_p\left\|\vec{v}^n_1 - \vec{v}^n_2\right\|_p < \left\|\vec{v}^n_1 - \vec{v}^n_2\right\|_p\,.
    \]
    On the other hand, if we can find a $\epsilon$ which makes $\left\|\vec{v}_1^{n+1} - \vec{v}_2^{n+1}\right\|_p \leq \left\|\vec{f}_1 - \vec{f}_2\right\|_p$, the iteration scheme is convergent.

    The solution of the second step in Eq. (\ref{eq:fixed_iteration}) can be written element-wisely as
    \begin{equation}
         v_i^{\frac{1}{q_i}} = - \mathbb{D}_{ii} v_i + \left(\epsilon + \mathbb{D}_{ii}\right) f_i\,,
         \label{eq:step2}
    \end{equation}
    where the superscript $n+1$ of $v_i$ is omitted without confusion. Since $\mathbb{R} \geq 0$, $f_i \geq \frac{a_i}{\epsilon + \DD_{ii}}$, the x-intercept of the linear function on the right-hand side of Eq. (\ref{eq:step2}) satisfies
    \[
        \frac{\left(\epsilon + \DD_{ii}\right) f_i}{\DD_{ii}} \geq  \frac{a_i}{\epsilon + \DD_{ii}} \frac{\left.\epsilon + \DD_{ii}\right.}{\DD_{ii}} = \frac{a_i}{\DD_{ii}}\,,
    \]
    and the y-intercept satisfies
    \[
      \left(\epsilon + \DD_{ii}\right) f_i \geq  \frac{a_i}{\epsilon + \DD_{ii}} \left(\epsilon + \DD_{ii}\right) = a_i\,.
    \]
    Denoting the first-phase solution of $v_i^{\frac{1}{q_i}} = -\DD_{ii} v_i + a_i$ as $\widetilde{v}_i$ which is greater than 0 and its vector form as $\vec{\widetilde{v}}$, it can be concluded that the solution of Eq. (\ref{eq:step2}) is a bounded value locates in $[\widetilde{v}_i,+\infty)$.

    Meanwhile,
    \begin{equation}
        \frac{\partial \vec{F}}{\partial \vec{v}} = \left(\epsilon \mathbb{I+D}\right)^{-1}\left(\text{Diag}(\frac{1}{q}\vec{v}^{\frac{1}{q}-1})+\mathbb{D}\right) \,,
    \end{equation}
    where $\text{Diag}(\cdot)$ means a diagonal matrix and 
    \[\frac{1}{q}\vec{v}^{\frac{1}{q}-1} = \left\{ \frac{1}{q_1}v_1^{\frac{1}{q_1}-1}, \frac{1}{q_2}v_2^{\frac{1}{q_2}-1}\cdots, \frac{1}{q_N}v_N^{\frac{1}{q_N}-1}\right\}^T\,.
    \]
    If $\epsilon  = \min(\frac{1}{q_i}\widetilde{v}_{i}^{\frac{1}{q_i}-1}),\,\,i = 1,2,\cdots,N$,  each component of $\vec{f}_{1}-\vec{f}_{2}$ and $\vec{v}_{1}^{n+1}-\vec{v}_{2}^{n+1}$ satisfies
    \[
        \frac{|\vec{f}_{1i}-\vec{f}_{2i}|^p}{|\vec{v}_{1i}^{n+1} - \vec{v}_{2i}^{n+1}|^p} = \left(\left.\frac{\partial \vec{F}}{\partial \vec{v}} \right|_{\vec{v}*}\right)_{ii}^p \geq 1,\,\,\vec{v}* \in [\vec{v}_1^{n+1},\vec{v}_2^{n+1}]  \subset [\vec{\widetilde{v}}, +\infty)\,.
    \]
    Thus, we have
    \[
        \left\|\vec{v}_1^{n+1} - \vec{v}_2^{n+1}\right\|_p \leq \left\|\vec{f}_1 - \vec{f}_2\right\|_p \leq \left\|\left(\epsilon \mathbb{I+D}\right)^{-1}\mathbb{R}\right\|_p\left\|\vec{v}^n_1 - \vec{v}^n_2\right\|_p < \left\|\vec{v}_1^{n} - \vec{v}_2^{n}\right\|_p\,.
    \]
    In conclusion, with the selected $\epsilon$, for arbitrary given first phase initial values, the iteration scheme defined as in Eq. (\ref{eq:fixed_iteration}) converges to the only one solution. Given $\vec{v}^n > 0$, $\vec{v}^{n+1}$ is also greater than 0 and the converged solution locates in the first phase. The proof is completed.
\end{proof}

If $q_i = 1, \forall\,\, i = 1,2\cdots,N$, Theorem \ref{theorem:1} reduces to the case of the PD system as in \cite{Huang2019,Huang2019a}. From Theorem $\ref{theorem:1}$, it is concluded that the power of the Patankar coefficients can be chosen as any positive value.
\begin{theorem}
    Among all the reactions, there is only one reaction involving with more than one reactant but without production. Without losing generality, assume $\#\Ret(1) > 1$ and $\Ret(1) = \left\{1,2,\cdots, k\right\}$. $q_1 = p$. For the reaction with only one reactant, $q_k = 1,\,\,k=2,3,\cdots,K$. Under such conditions, Eq. (\ref{eq:conjecture}) has at least one solution, and there is one and the only one satisfying $\vec{c} > 0$.
\end{theorem}
\begin{proof}\label{proof2}
    Equation (\ref{eq:conjecture}) can be rewritten as
    \begin{equation}
        \MM \vec{c}  = \vec{a} + \vec{B} t\,,
        \label{eq:theorem2_1}
    \end{equation}
    where $\MM$ has the same definition as in Eq. (\ref{eq:M_Matrix}) except that the summation only applies to reactions from 2 to $K$. $\vec{B} = \left\{-b_{11},-b_{12},\cdots,-b_{1k}, 0,0,\cdots,0\right\}^T$ 
    and $t = \left(\frac{c_1 c_2 \cdots c_k}{a_1^* a_2^* \cdots a_k^*}\right)^q$.
    
    Thus
    \begin{equation}
        \vec{c} = \MM^{-1} \left(\vec{a} +\vec{B} t\right).
        \label{eq:theorem2_c}
    \end{equation}
    Since $\MM$ is a column diagonally dominant $m$-matrix, all entries of $\MM^{-1}$ is non-negative, and the diagonal elements are positive and are greater than other elements in the same row, denote the right-hand side of Eq. (\ref{eq:theorem2_c}) as
    \[
    \MM^{-1} \left(\vec{a} +\vec{B} t\right) = \left\{\xi_1 - \zeta_1 t,\xi_2 - \zeta_2 t,\cdots,\xi_k - \zeta_k t, 0,0,\cdots, 0\right\}^T\,,   
    \]
    where $\xi_i$, $\zeta_i$, $i=1,2,\cdots,k$ are positive values. Then, we have
    \begin{equation}
        t = \left(\frac{\left(\xi_1 - \zeta_1 t\right)\left(\xi_2 - \zeta_2 t\right)\cdots\left(\xi_k - \zeta_k t\right)}{a_1^* a_2^* \cdots a_k^*}\right)^q\,.
        \label{eq:theorem2_3}
    \end{equation}
    Since the first phase solution should satisfy $\vec{c} > 0$, we have
    \begin{equation}
        t < \min(\frac{\xi_i}{\zeta_i}),\,i=1,2,\cdots,k\,.    
        \label{eq:theorem2_constraint}
    \end{equation}
    With constraint of Eq. (\ref{eq:theorem2_constraint}), in the first phase, Eq. (\ref{eq:theorem2_3}) has one and only one solution. Thus, in the first phase, Eq. (\ref{eq:theorem2_1}) has one and only one solution. The proof is completed.
\end{proof}
\subsection{Positivity for General Conditions}\label{sec:general_positivity}
For general cases, we have tried to prove the correctness of Conjecture \ref{conjecture} but have not yet found the approach. Thus, this problem is remained to be solved. Instead of providing a rigorous and mathematically elegant proof, in this study, we try to validate this conjecture through large randomly generated samples.

A system with 100 species ($N = 100$) and 100 reactions ($K = 100$) are considered. Each reaction involves randomly sampled but with a maximum number of 10 reactants and randomly sampled productions. The maximum number of 10 reactants is considered here to avoid the cut-off error due to too many floating number multiplications. For real systems of chemical reactions, reaction with more than 10 species is also rarely seen. All the constants $a_i$, $a_i^*$ and $b_{ki}$ are drawn out from a uniform distribution in $[0,1]$ and satisfy Eq. (\ref{eq:dominant}). For a real reactive system, the value of $a_i$, $a_i^*$ and $b_{ki}$ of course can be of any value. However, it should be noted that, if $c_i$, $a_i$, $a_i^*$ and $b_{ki}$ hold for Eq. (\ref{eq:conjecture}), after a scaling, $h c_i$, $h a_i$, $h a_i^*$ and $h b_{ki}$ still hold for Eq. (\ref{eq:conjecture}), where $h$ is an arbitrary real number. This means, for any sets of coefficients, that they can be scaled and locate in the region of $[0,1]$. Thus, it is reasonable to draw the samples of $a_i$, $a_i^*$, $b_{ki}$ out from $[0,1]$. For the power of the Patankar coefficient, case with $q_k = \frac{1}{\#\Ret(k)}$ is considered, which is also a consideration from the viewpoint of truncation error and cut-off error.

Since obtaining all the solutions for such a big system is very time-consuming, Newton iteration is used to solve the equations in the real space. 1000000 samples are drawn out and tested. Denotes $\chi_k = \left(\prod_{\varepsilon \in \Ret(k)}\frac{c_\varepsilon}{a_\varepsilon^{*}}\right)^{q_k}$, $v_i = c_i^{q_m^{i}}$, $f_i = c_i - a_i + \sum_{k, \Ret(k) \ni i} \left[b_{ki} \chi_k\right]- \sum_{k, \Ret(k) \not\ni i} \left[b_{ki} \chi_k\right]$, $J_{ij} = \frac{\partial f_i}{\partial v_j}$ and $\mathbf{J} = \left\{J_{ij}\right\}$, where $q_m^i = \min_{\Ret(k) \ni i}\left\{q_k\right\}$. To keep the positivity of $c_i$ during the iteration, Newton iterative algorithm with a re-initialization strategy is briefly listed in Algorithm \ref{algorithm_newton}, which is shown
to be effective in all the testing cases and numerical examples in Sec. \ref{sec:numerical_examples}.

For each sample, 10 different randomly generated initial values of $c_i$ are tested and converge to the same solution. Since the data for the tested coefficients of Eq. (\ref{eq:conjecture}) are too large for the presentation, the source code for generating and testing samples are directly provided in Pan \cite{Pan2020} as a supporting information. Interested readers can download and test it. Although being not a rigorous proof, the result of this section can provide as a solid supporting to Conjecture \ref{conjecture}, which is the basis of the proposed scheme.

\begin{algorithm}
    \caption{Newton iteration for the Patankar scheme.\label{algorithm_newton}}
\begin{algorithmic}
    \STATE $c_i \leftarrow c_i^0$
    \STATE $v_i \leftarrow c_i^{q_m^i}$
    \WHILE{$f_i$ > Threshold}
        \STATE $f_i \leftarrow c_i - a_i + \sum_{k, \Ret(k) \ni i} \left[b_{ki} \chi_k\right]- \sum_{k, \Ret(k) \not\ni i} \left[b_{ki} \chi_k\right]$
        \FOR { $k = 1$ to $K$}
        \STATE $\chi_k \leftarrow \left(\prod_{\varepsilon \in \Ret(k)}\frac{c_\varepsilon}{a_\varepsilon^{*}}\right)^{q_k}\,.$
        \ENDFOR
        \STATE $\Delta \mathbf{v} =  - \mathbf{J}^{-1} \mathbf{v}.$
        \FOR{$i = 1$ to $N$}
            \IF {$v_i + \Delta v_i < 0$}
              \STATE $c_i \leftarrow a_i + \sum_{k, \Ret(k) \not\ni i} \left[b_{ki} \chi_k\right]$
              \STATE $c_i \leftarrow \min\left(c_i, \sum_{i=1}^N{a_i}\right)$
              \STATE $v_i \leftarrow c_i^{1/q_m^i}$
            \ELSE
              \STATE $v_i \leftarrow v_i + \Delta v_i$
              \STATE $c_i \leftarrow v_i^{q_m^i}$
            \ENDIF
        \ENDFOR
    \ENDWHILE
\end{algorithmic}
\end{algorithm}

\section{Positive-Preserving for the Euler Equations}\label{sec:positivity_euler}
For the Euler equations coupled with convection terms, the positive-preserving technology proposed in Zhang and Shu \cite{Zhang2010} is firstly reviewed and utilized to guarantee the positivity of the convection part. Then the PMPRK scheme is used to keep the positivity concerning the source terms. As will be shown later, the PMPRK scheme of the form Eq. (\ref{eq:chi1_detailed}) can be easily coupled with the energy equation and help keep the positivity of the pressure.

If we define the admissible set $G$ as in Huang and Shu \cite{Huang2019a} and Huang et al. \cite{Huang2019}, that is
\begin{equation}
    G = \left\{U = \left(\rho, Y_1, Y_2, \cdots, Y_N, m, E\right)^T|\rho > 0, Y_i > 0,\,\,i = 1,2,\cdots,N, p > 0\right\}\,,
    \label{eq:G}
\end{equation}
it was found that if the reaction is endothermic, $G$ may not be a convex set due to the existence of $h_i^0$, see Section 3.3 and 3.4 in Huang and Shu \cite{Huang2019a} and thus the pressure could not be guaranteed to be positive as stated. However, if we make a little rearrangement of Eq. (\ref{eq:euler1}), this problem is readily solved. Rewrite Eq. (\ref{eq:euler1detailed}) as
\begin{equation}
    \begin{aligned}
        U & = \left( \rho\,, m\,, E\,, \rho Y_1\,,\rho Y_2,\cdots,\rho Y_N \right)^T\,,\\
        F(U) &= \left(m, \rho u^2 +p, (E+p) u, \rho u Y_1, \rho u Y_2,\cdots \rho Y_N\right)^T\,,\\
        S(U) & = \left(0,0,\omega_E,\omega_1,\omega_2,\cdots,\omega_N\right)^T\,.
    \end{aligned}
    \label{eq:euler2}
\end{equation}
and
 \begin{equation}
     \begin{aligned}
    \omega_E & = -\frac{\partial \rho Y_i h_i^0}{\partial t} - \frac{\partial \rho Y_i h_i^0 u}{\partial x}\,,\\
     E & = \frac{1}{2} \rho u^2 + \rho Y_i e_i(T)\,.
     \end{aligned}
     \label{eq:redefineE}
 \end{equation}
 All other definitions are the same as in Eq. (\ref{eq:euler1detailed}). Further, according to the mass conservation equations and considering that $h_i^0$ is constant for each specie, 
 \begin{equation}
    \omega_E = - \sum_{k=1}^K{R_k \lambda_{ki}M_i h_i^0}\,.
    \label{eq:we}
 \end{equation}
Equations (\ref{eq:euler2}-\ref{eq:we}) are the final forms that we solved for the Euler equations in this study. Then, we define the admissible set of $G$ in Eq. (\ref{eq:G}) with $E$ defined as in Eq. (\ref{eq:redefineE}). After rewriting, $G$ becomes a convex set as in Zhang and Shu \cite{Zhang2010,Zhang2011} and can be preserved to be positive if only the source term of energy equation is treated appropriately.

Assume that the computational domain constitutes $M$ uniformly distributed cells and cell $j$ is defined in $[x_{j-\frac{1}{2}}, x_{j+\frac{1}{2}}]$. Consider the following finite volume scheme with forward Euler time integration,
\begin{equation}
    \overline{U}_j^{n+1} = \overline{U}_j^n - \frac{\Delta t}{\Delta x} \left[\hat{F}\left(U_{j+\frac{1}{2}}^-, U_{j+\frac{1}{2}}^+\right) - \hat{F}\left(U_{j-\frac{1}{2}}^-, U_{j-\frac{1}{2}}^+\right)\right] + \Delta t\overline{S}^n\,,
\end{equation}
where $\Delta t$ and $\Delta x$ are the time and space step, respectively. $\overline{U}_j$ means the averaged value in cell $j$. The superscripts $n$ and $n+1$ mean that the variable is defined in time step $n$ and $n+1$, respectively. Denote the piece-wise high order reconstructed polynomial based on the cell averaged value $\overline{U}^n$ for each cell as $Q^n_j(x)$. $U_{j+\frac{1}{2}}^- = Q^n_j(x_{j+\frac{1}{2}})$ and $U_{j-\frac{1}{2}}^- = Q^n_j(x_{j-\frac{1}{2}})$. $\hat{F}$ is the numerical flux. $\overline{S}^n$ is defined as 
\begin{equation}
\overline{S}^n = \frac{1}{\Delta x}\int^{x_{j+\frac{1}{2}}}_{x_{j-\frac{1}{2}}}{S(Q_j(x))}\,.
\end{equation}

Let $x_j^\eta$ be the $W$-point Legendre Gauss-Lobatto quadrature points for the interval $[x_{j-\frac{1}{2}}, x_{j+\frac{1}{2}}]$ and $\hat{w}_\eta$ be the corresponding weights int the standard cell $[-1/2,1/2]$ such that $\sum_{\eta = 1}^{W}{\hat{w}_\eta} = 1$ and $2 W -3 $ is greater than the order of $Q^n_j(x)$. 
\begin{theorem}\cite{Zhang2010}
    For a finite volume scheme and a positive preserving flux $\hat{F}$, if $Q^n_j(x_j^\alpha) \in G$ for all $j$ and all $\alpha$, then $\overline{U}_j^n - \frac{\Delta t}{\Delta x} \left[\hat{F}\left(U_{j+\frac{1}{2}}^-, U_{j+\frac{1}{2}}^+\right) - \hat{F}\left(U_{j-\frac{1}{2}}^-, U_{j-\frac{1}{2}}^+\right)\right] \in G$ under the CFL condition
    \begin{equation}
        \left\||u|+acou\right\|_\infty \frac{\Delta t}{\Delta x} \leq \hat{w}_1 \delta_0\,.
    \end{equation}
    \label{theorem_positive}
\end{theorem}
$acou$ is the acoustic wave speed defined as $acou = \sqrt{\gamma p/\rho}$, $\delta_0$ is a constant number depending on the positive preserving flux, e.g., the Godunov flux \cite{Einfeldt1991}, the Lax-Friedrichs flux \cite{Perthame1996}, the Boltzmann type flux \cite{Mathematics2019} , the Harten-Lax-van Leer flux \cite{Harten1983} and the HLLC \cite{batten1997choice,cheng2014positivity} flux. For the Lax-Friedrichs flux, $\delta_0$ can reach $1$ \cite{Zhang2010} and for four-point Legendre Gauss-Lobatto quadrature points, $\hat{w}_1 = \frac{1}{6}$.

Then incorporating the PMPRK scheme, the final single-step forward Euler time integration scheme of the Euler equations can be written as
\begin{equation}
    \begin{aligned}
        \overline{\rho}_j^{n+1} & = \left[\overline{\rho}_j^n - \frac{\Delta t}{\Delta x} (\hat{F}_{j+\frac{1}{2},1} - \hat{F}_{j-\frac{1}{2},1})\right]\,,\\
        \overline{m}_j^{n+1} & = \left[\overline{m}_j^n - \frac{\Delta t}{\Delta x} (\hat{F}_{j+\frac{1}{2},2} - \hat{F}_{j-\frac{1}{2},2})\right]\,,\\
        \overline{\rho e}_j^{n+1} & = \left[\overline{E}_j^n - \frac{\Delta t}{\Delta x} (\hat{F}_{j+\frac{1}{2},3} - \hat{F}_{j-\frac{1}{2},3}) - \frac{1}{2}\frac{{\overline{m}_j^{n+1}}^2}{\overline{\rho}_j^{n+1}}\right] - \Delta t \sum_{k=1}^K{\overline{R}_k^n \lambda_{ki} M_i h_i^0 \chi_k} \,,\\
        \overline{\rho Y_{i,j}}^{n+1} & = \left[\overline{\rho Y_{i,j}}^{n} - \frac{\Delta t}{\Delta x}(\hat{F}_{j+\frac{1}{2},3+i} - \hat{F}_{j-\frac{1}{2},3+i})\right] + \Delta t\overline{R}_k^n \lambda_{k(i)} M_{(i)} \chi_k\,,
        i=1,2,\cdots,N\,,
    \end{aligned}
    \label{eq:euler_update}
\end{equation}
where $\overline{\rho e} = \overline{E} - \frac{1}{2}\frac{{\overline{m}}^2}{\overline{\rho}}$. According to Theorem \ref{theorem_positive}, the values inside the square bracket in Eq. (\ref{eq:euler_update}) are all positive. In Eq. (\ref{eq:euler_update}), $\rho e$ can be treated as a special kind of specie. If $ \lambda_{ki}M_i h_i^0 < 0$ (exothermic), $\rho e$ is treated as a production; if $\lambda_{ki}M_i h_i^0 > 0$ (endothermic), $\rho e$ is treated as a reactant. However, the introduction of $\rho e$ equation breaks down the constraints of Eq. (\ref{eq:dominant}), i.e., if $\rho e$ is treated as a specie and if the reaction is exothermic, $\sum_{i \in \Ret(k)}{b_{ki}} < \sum_{i \not \in \Ret(k)}{b_{ki}}$. . Luckily, through the numerical investigations in a same way as described in Sec. \ref{sec:general_positivity}, even for the general system with 100 species and 100 reactions with each of which being either exothermic or endothermic, there exists a unique solution for the PMPRK scheme, which is different from the purely linear case in Huang and Shu \cite{Huang2019a} and Huang et al. \cite{Huang2019}, where a diagonal dominant requirements of Eq. (\ref{eq:dominant}) must be fulfilled.

To illustrate this, let us consider a very specific case with two species and two reactions. Each reaction involves only one reactant, and one reaction is endothermic, the other is exothermic. The system of PMPRK scheme, when written in the form of Eq. (\ref{eq:conjecture}), can be expressed as
\begin{equation}
    \begin{aligned}
    c_1 & = a_1 - b_{11} \frac{c_1}{c_1^*} + b_{21}\sqrt{\frac{c_2 c_3}{c_2^* c_3^*}}\,,\\
    c_2 & = a_2 + b_{12} \frac{c_1}{c_1^*} - b_{22}\sqrt{\frac{c_2 c_3}{c_2^* c_3^*}}\,,\\
    c_3 & = a_3 + b_{13} \frac{c_1}{c_1^*} - b_{23}\sqrt{\frac{c_2 c_3}{c_2^* c_3^*}}\,,\\
    \end{aligned}
    \label{eq:2specie2reaction}
\end{equation}
where $c_3$ represents the equation for the internale energy and $b_{11} \geq b_{12}$, $b_{21} \leq b_{22}$. Since $b_{ki}$ and $c_i^*$, $i=1,2,3; k=1,2$ are all constants, denoting $b_{1i}^* = b_{1i}/c_1^*, i=1,2,3$ and $b_{2i}^* = b_{2i}/\sqrt{c_2^* c_3^*}$, Eq. (\ref{eq:2specie2reaction}) can be simplified as
\begin{equation}
    \begin{aligned}
    c_1 & = a_1 - b_{11}^* {c_1} + b_{21}^*\sqrt{c_2 c_3}\,,\\
    c_2 & = a_2 + b_{12}^* {c_1} - b_{22}^*\sqrt{c_2 c_3}\,,\\
    c_3 & = a_3 + b_{13}^* {c_1} - b_{23}^*\sqrt{c_2 c_3}\,,\\
    \end{aligned}
    \label{eq:2specie2reaction2}
\end{equation}
and $b_{11}^* \geq b_{12}^*$, $b_{21}^* \leq b_{22}^*$. Denoting $t = \sqrt{c_2 c_3}$, $c_1$ can be expressed as $c_1 = \frac{a_1+b_{21}^*t}{1+b_{11}^*}$, the substitution of which into $c_2$ and $c_3$ yields
\begin{equation}
    \begin{aligned}
        c_2 & = a_2 + \frac{a_1 b_{12}^*}{1+b_{11}^*} + (\frac{b_{12}^* b_{21}^*}{1+b_{11}^*}-b_{22}^*) t\,,\\
        c_3 & = a_3 + \frac{a_1 b_{13}^*}{1+b_{11}^*} + (\frac{b_{13}^* b_{21}^*}{1+b_{11}^*}-b_{23}^*) t\,.
    \end{aligned}
    \label{eq:c2c3t}
\end{equation}
Since $b_{12}^* \leq b_{11}^*$, $\frac{b_{12}^* b_{21}^*}{1+b_{11}^*}-b_{22}^* < b_{21}^* - b_{22}^* \leq 0$. Considering that $c_2 c_3 = t^2$ and $c_2$ and $c_3$ in Eq. (\ref{eq:c2c3t}) should all be greater than 0, it can be concluded that there exists a unique $t$ satisfying $t > 0,\,c_1 >0,\,c_2 >0$ and $c_3 > 0$.

Thus, the updating scheme of Eqs. ($\ref{eq:euler_update}$) is able to keep the positivity of $\overline{\rho e}^{n+1},\,\overline{\rho Y_i}^{n+1}$ given that $\overline{\rho e}^{n} > 0,\,\overline{\rho Y_i}^{n} > 0,\,i=1,2,\cdots,N$, 
that is, the positivity of both density and pressure are preserved.

The extension to high order SSP Runge-Kutta (SSP-RK) scheme is straightforward and as described in the literature work \cite{zhang2010positivity,zhang2010maximum,xing2010positivity,zhang2012positivity,zhang2012maximum,zhang2013maximum}, i.e., for each step of SSP-RK, the update scheme can be written as a convex combination of the forward Euler scheme. 

To preserve the positivity, the restriction of the CFL condition is
\begin{equation}
        \left\||u|+acou\right\|_\infty \frac{\Delta t}{\Delta x} \leq \hat{w}_1 \delta_0 \min{\frac{\alpha_{ij}}{\beta_{ij}}}\,.
        \label{eq:cfl_upper_bond}
\end{equation}
To maximize $\min{\frac{\alpha_{ij}}{\beta_{ij}}}$, for current second order PMPRK scheme, the following RK coefficients are selected,
\begin{equation}
    \alpha_{10} = \beta_{10} = 1, \quad \alpha_{20} = \alpha_{21} = \beta_{21} = \frac{1}{2} \quad \beta_{20} = 0\,,
\end{equation}
with  $\min{\frac{\alpha_{ij}}{\beta_{ij}}} = 1$.

Note that, in this study, the 2-D examples are solved on the structured grids, where the reconstructed polynomials are obtained through a dimension-by-dimension way. The source terms are integrated with a first order of accuracy. To capture both the shock and contact discontinuities of multi-species with high resolutions, a minimum dispersion and controllable dissipation (MDCD) scheme proposed by Sun, Ren, Wang and their collaborators \cite{sun2011class, wang2013low, wang2020consistent} is utilized in the characteristic space. The flux for the Euler equations is HLLC \cite{wang2020consistent} for multi-species. Besides the finite volume method, the proposed PMPRK scheme can of course be applied in any other framework of high order of accuracy, such as the DG method and the finite difference method.
\section{Numerical Examples\label{sec:numerical_examples}}
\subsection{Accuracy Test}
Firstly, the accuracy of the proposed PMPRK scheme is verified through the following ordinary differential equations,
\begin{equation}
\begin{aligned}
    \frac{\mathrm{d} c_1}{\mathrm{d} t} & = -\lambda_0 c_1^3 c_2^3 + 4 \lambda_1 c_2^3 c_3^3 c_4\,,\\
    \frac{\mathrm{d} c_2}{\mathrm{d} t} & = -2 \lambda_0 c_1^3 c_2^3 - 2 \lambda_1 c_2^3 c_3^3 c_4\,,\\
    \frac{\mathrm{d} c_3}{\mathrm{d} t} & = 3\lambda_0 c_1^3 c_2^3 - \lambda_1 c_2^3 c_3^3 c_4\,,\\
    \frac{\mathrm{d} c_4}{\mathrm{d} t} & =  - \lambda_1 c_2^3 c_3^3 c_4\,,
\end{aligned}
\end{equation}
where two reactions are involved, $\lambda_0$ and $\lambda_1$ are constants. Initial values are $\left\{c_1,c_2,c_3,c_4\right\} = \left\{0.1, 0.4, 0, 1.0\right\}$. Time step $\Delta t = {t_{end}}/{N_t}$. The results of $N_t = 10240000$ are taken as the reference values. Two different conditions are considered, i.e., $\left\{\lambda_0, \lambda_1\right\} = \left\{1000, 2000\right\}$ and $\left\{\lambda_0, \lambda_1\right\} = \left\{1\times10^8, 2\times10^6\right\}$, the latter of which represents a system with stiff source terms. The final time is $t_{end} = 0.5$ and $0.02$ for each condition, respectively. Errors and numerical orders are listed in Tab. \ref{tab:accuracy_test}. For the non-stiff system with $\left\{\lambda_0, \lambda_1\right\} = \left\{1000, 2000\right\}$, the second prior order of accuracy is observed. For the stiff system with $\left\{\lambda_0, \lambda_1\right\} = \left\{1\times10^8, 2\times10^6\right\}$, super-convergence is observed when $N_t \leq 320$. However, when $N_t = 640$, the order decreases to 0.64 but recovers to the prior order of convergence when $N_t$ increases further. For both cases, the positivity of $c_i$ is preserved unconditionally.

\begin{table}[htbp]
    \centering
\caption{Accuracy test with initial values $\left\{c_1,c_2,c_3,c_4\right\} = \left\{0.1, 0.4, 0, 1.0\right\}$.\label{tab:accuracy_test}}
\begin{tabular}{ccccccccc}
    \toprule
\multirow{2}{*}{$N_t$} &
 \multicolumn{4}{c}{$\left\{\lambda_0, \lambda_1\right\} = \left\{1000, 2000\right\}$} & \multicolumn{4}{c}{$\left\{\lambda_0, \lambda_1\right\} = \left\{1\times10^8, 2\times 10^6\right\}$} \\
 \cline{2-9}
    & $Error_2$ & Order & $Error_{\infty}$ & Order & $Error_2$ & Order & $Error_{\infty}$ & Order \\
    \cline{1-9}
40  &1.80E-6  &     & 2.84E-6 &      & 4.00E-3 &    & 6.19E-3 &    \\ 
80  &4.45E-7  &2.01 &7.02E07  & 2.02 & 7.67E-4 &2.38& 1.17E-3 &2.41\\
160 &1.11E-7  &2.00 & 1.75E-7 & 2.00 & 6.79E-5 &3.50& 9.78E-5 &3.58\\
320 &2.77E-8  &2.00 & 4.36E-8 & 2.00 & 6.89E-6 &3.30& 9.84E-6 &3.31\\
640 &6.92E-9  &2.00 &1.09E-8  & 2.00 & 3.57E-6 &0.95& 6.00E-6 &0.71\\
1280&1.73E-9  &2.00 &2.74E-9  & 1.99 & 1.27E-6 &1.49& 2.17E-6 &1.47\\
2560&4.32E-10 &2.00 &6.96E-10 & 1.98 & 3.14E-7 &2.02& 5.33E-7 &2.03\\
\bottomrule
\end{tabular}
\end{table}

\subsection{1D Detonation of $CH_4$\label{section1ddetonationCH4}}
In this example, the following 1D detonation problem \cite{bao2002random} is considered,
\[
    CH_4 + 2O_2 \rightarrow CO_2 +2 H_2O\,.
\]
The reaction rate is defined as
\begin{equation}
    R_1 = B_1 T^{\alpha_1} H(T-T_1) \frac{\rho Y_{CH_4}}{M_{CH_4}}\left(\frac{\rho Y_{O_2}}{M_{O_2}}\right)^2\,,
\end{equation}
where $B_1 = 1\times 10^6$, $\alpha_1 = 0$, $T_1 = 2.0$ and $H(x)$ is the Heaviside function. $M_{CH_4} = 16$, $M_{O_2} = 32$, $M_{CO_2} = 44$ and $M_{H_2O} = 18$. $h^0_{CH_4} = 500$ and $h^0_{O_2} =h^0_{CO_2} = h^0_{H_2O} = 0$. The specific heat ratio is set to be constant as $\gamma = 1.4$. Different from the idea gas assumption, in this example, the temperature is defined as $T = p/\rho$. The computational domain is $[-25, 25]$. Initial data are given by
\begin{equation}
    \left( \rho, u, p, Y_{CH_4}, Y_{O_2}, Y_{CO_2}, Y_{H_2O} \right) = 
        \left\{
            \begin{aligned}
                &\left( 2, 10, 40, 0, 0.2, 0.475, 0.325\right)\,,&x\leq -22.5\,,\\
                &\left( 1, 0, 1, 0.1, 0.6, 0.2, 0.1\right)\,,&x > -22.5\,.
            \end{aligned}
            \right.
\end{equation}
To accurately capture the detonation wave while avoiding the technique of subcell resolution \cite{wang2015high}, the computational domain is divided into 600 uniformly distributed cells. CFL number is chosen to be 0.2. It should be noted the used CFL number 0.2 is higher than the theoretical upper bond which guarantees the positivity in Eq. (\ref{eq:cfl_upper_bond}) when the four-point Legendre-Gauss-Lobatto quadrature rule is applied, i.e., $CFL \leq \frac{1}{6}$. From time $t^n$ to $t^{n} + \Delta t$, if negative pressure or density is found, the solved variables will be reinitialized by the solution at $t^{n}$ and the time step will be reduced to $\frac{\Delta t}{2}$. Such a time advancing technique keeps the CFL number as large as possible while guarantees the positivity when needed.

 The pressure, density and mass fraction of ${CH_4}$ and ${O_2}$ at $t = 3.0$ are plotted in Fig. \ref{fig:1dDetationCH4}, where the solid black lines are reference values obtained with grid cells $N = 4000$. As shown in Fig. \ref{fig:1dDetationCH4}, $N = 600$ is enough to resolve the location of right forward detonation wave.
 \begin{figure}
    \centering
    \begin{subfigure}[b]{0.4\textwidth}
        \includegraphics[width=\textwidth]{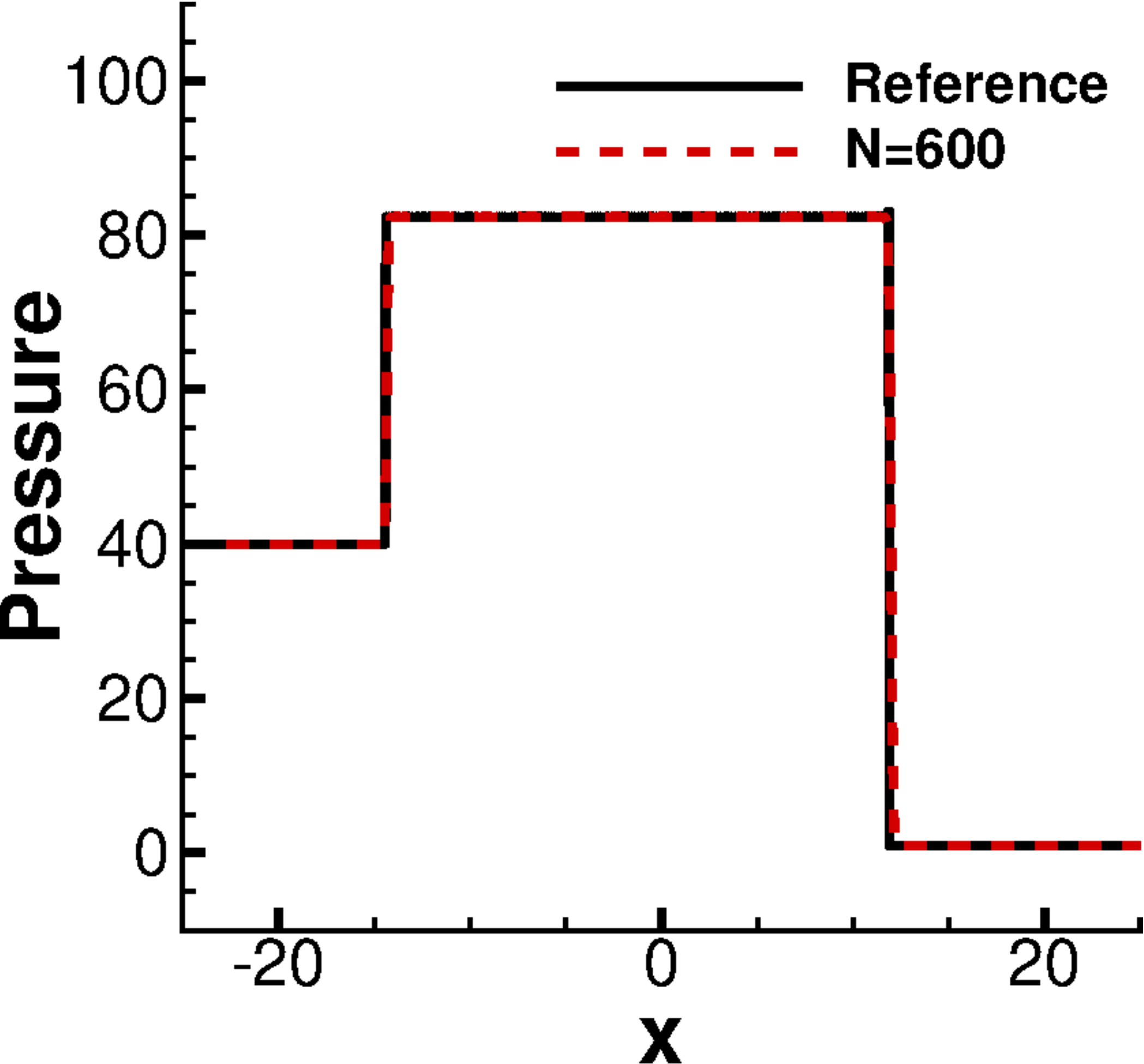}
    \end{subfigure}
    \begin{subfigure}[b]{0.4\textwidth}
        \includegraphics[width=\textwidth]{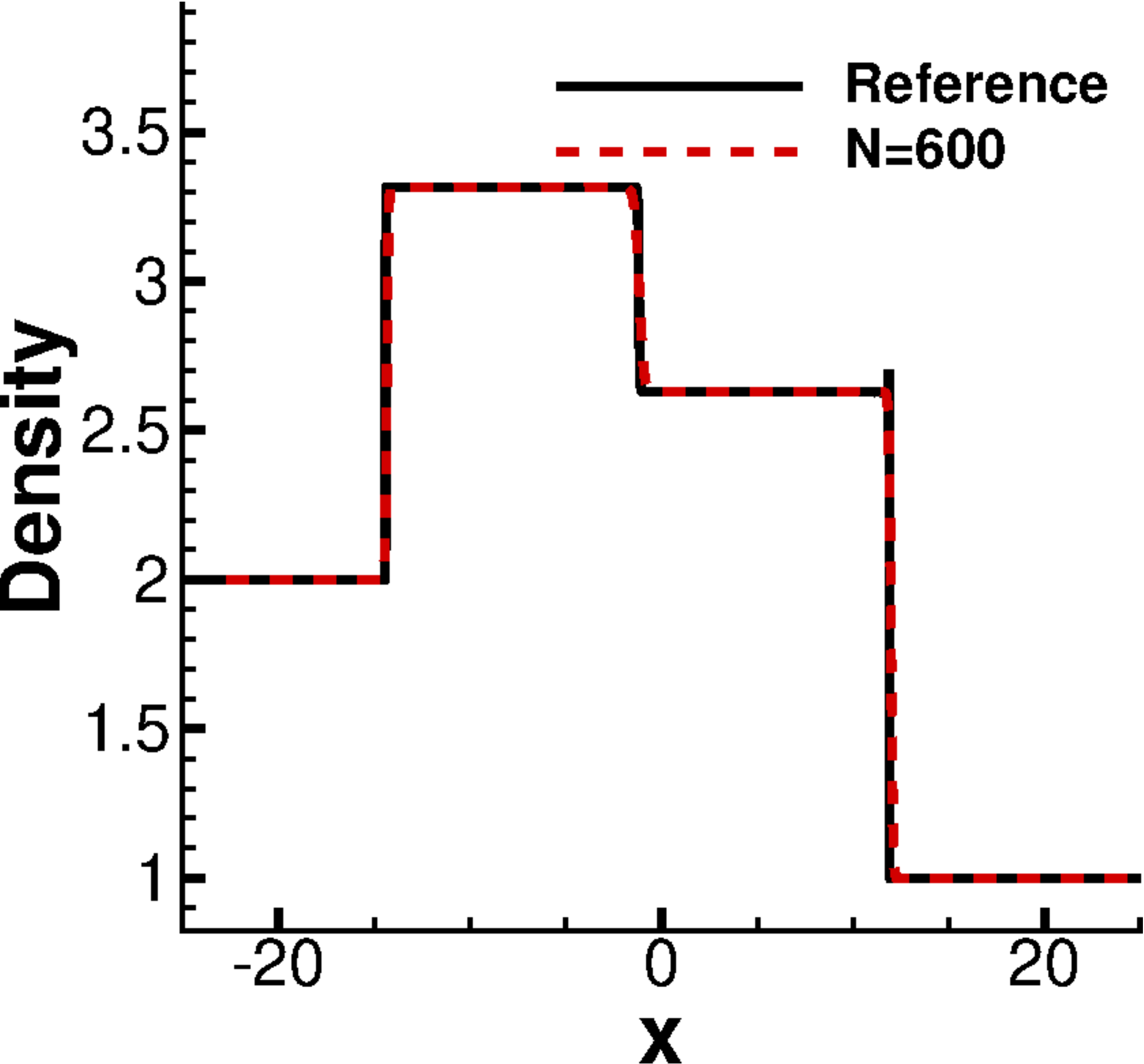}
    \end{subfigure}
    \begin{subfigure}[b]{0.4\textwidth}
        \includegraphics[width=\textwidth]{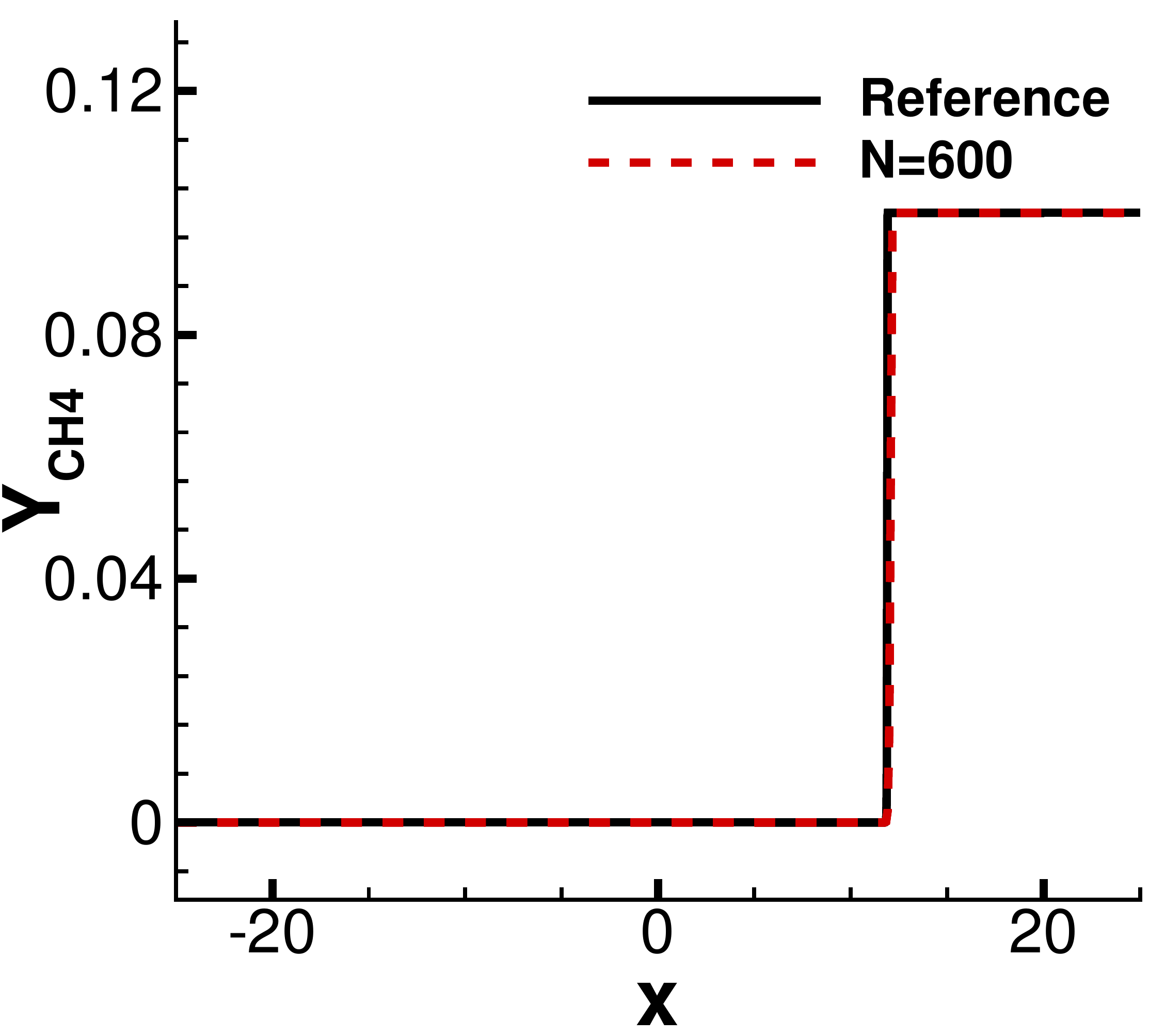}
    \end{subfigure}
    \begin{subfigure}[b]{0.4\textwidth}
        \includegraphics[width=\textwidth]{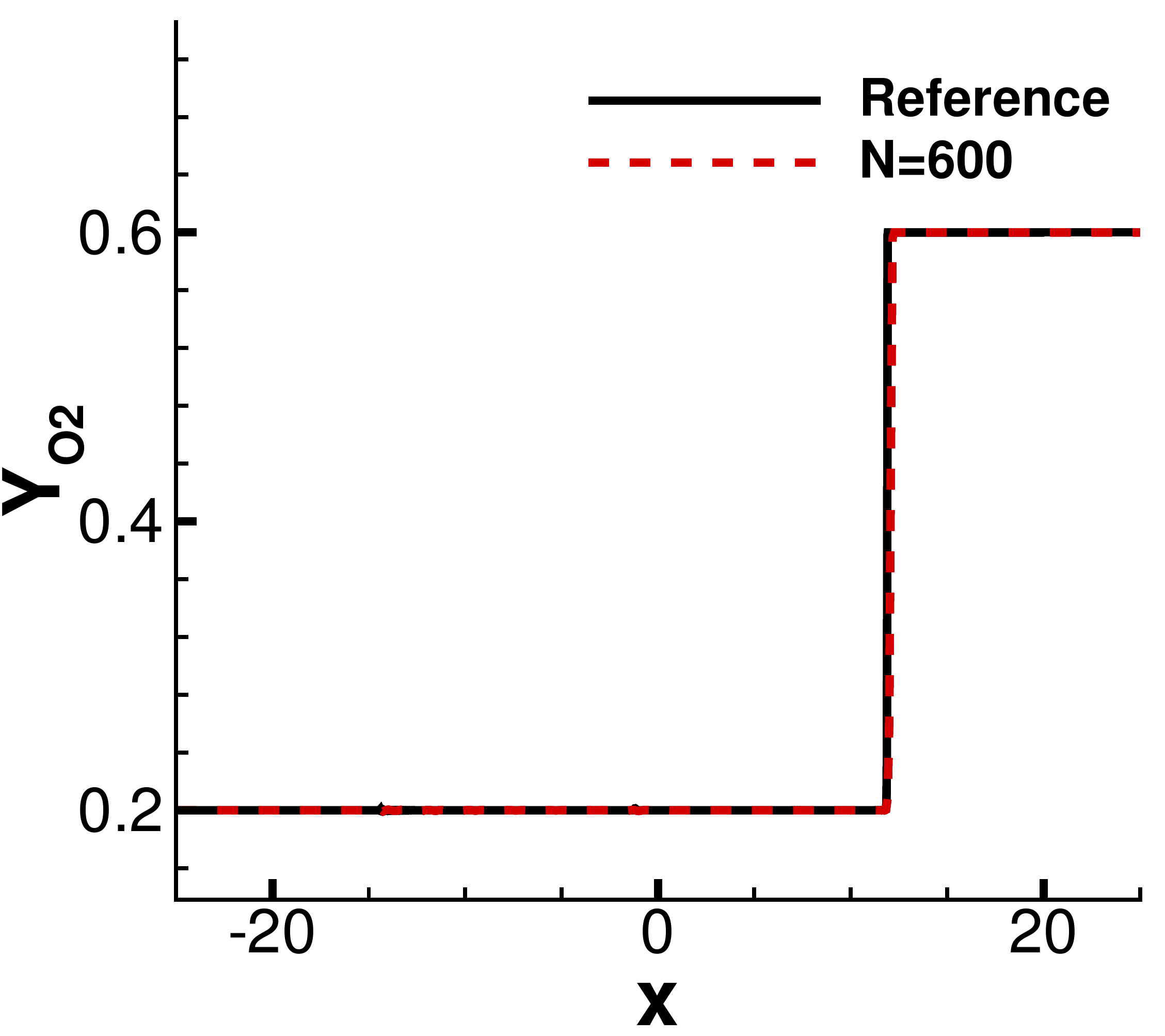}
    \end{subfigure}
\caption{Solutions of 1D detonation problem of Sec. \ref{section1ddetonationCH4} at time $t = 3.0$. CFL = 0.2 and cell number $N = 600$.}
\label{fig:1dDetationCH4}
 \end{figure}
 \subsection{2D Detonation of $CH_4$}
 In this example, a radial symmetry 2D Detonation problem \cite{wang2015high} is studied. The reaction model is the same as in Section \ref{section1ddetonationCH4} except that $h^0_{CH_4} = 200$. The computational domain is $[0,50]\times[0,50]$, where the lower and left boundaries are inviscid walls, and the upper and right boundaries are outflow conditions. The initial values consist of totally burnt gas inside a circle with radius of 10 and unburnt gas outside,
\begin{equation}
    \left( \rho, u, v, p, Y_{CH_4}, Y_{O_2}, Y_{CO_2}, Y_{H_2O} \right) = 
        \left\{
            \begin{aligned}
                &\left( 2, 10 \frac{x}{r}, 10 \frac{y}{r}, 40, 0, 0.2, 0.475, 0.325\right)\,,&r\leq 10\,,\\
                &\left( 1, 0, 0, 1, 0.1, 0.6, 0.2, 0.1\right)\,,&r > 10\,,
            \end{aligned}
            \right.
\end{equation}
where $r = \sqrt{x^2+y^2}$.

The computational domain consists of $600\times600$ cells. Contours of pressure, density and radial velocity are shown in Fig. \ref{fig:2dDetationCH4Contour}. Solutions along the line $x=y$ are shown in Fig. \ref{fig:2dDetationCH4Line}. A detonation wave, a contact discontinuity and a shock wave can be observed with the decreasing of $r$. Discontinuities are captured without oscillations and densities are advanced without negative values. Actually, for this problem, only one reaction involving two reactants is considered, the proposed Patankar scheme of Eq. (\ref{eq:conjecture}) can be solved analytically without Newton iterations. The results are consistent with those of the work in \cite{Du2019}.
 \begin{figure}[htbp]
    \centering
    \begin{subfigure}[b]{0.45\textwidth}
        \includegraphics[width=\textwidth]{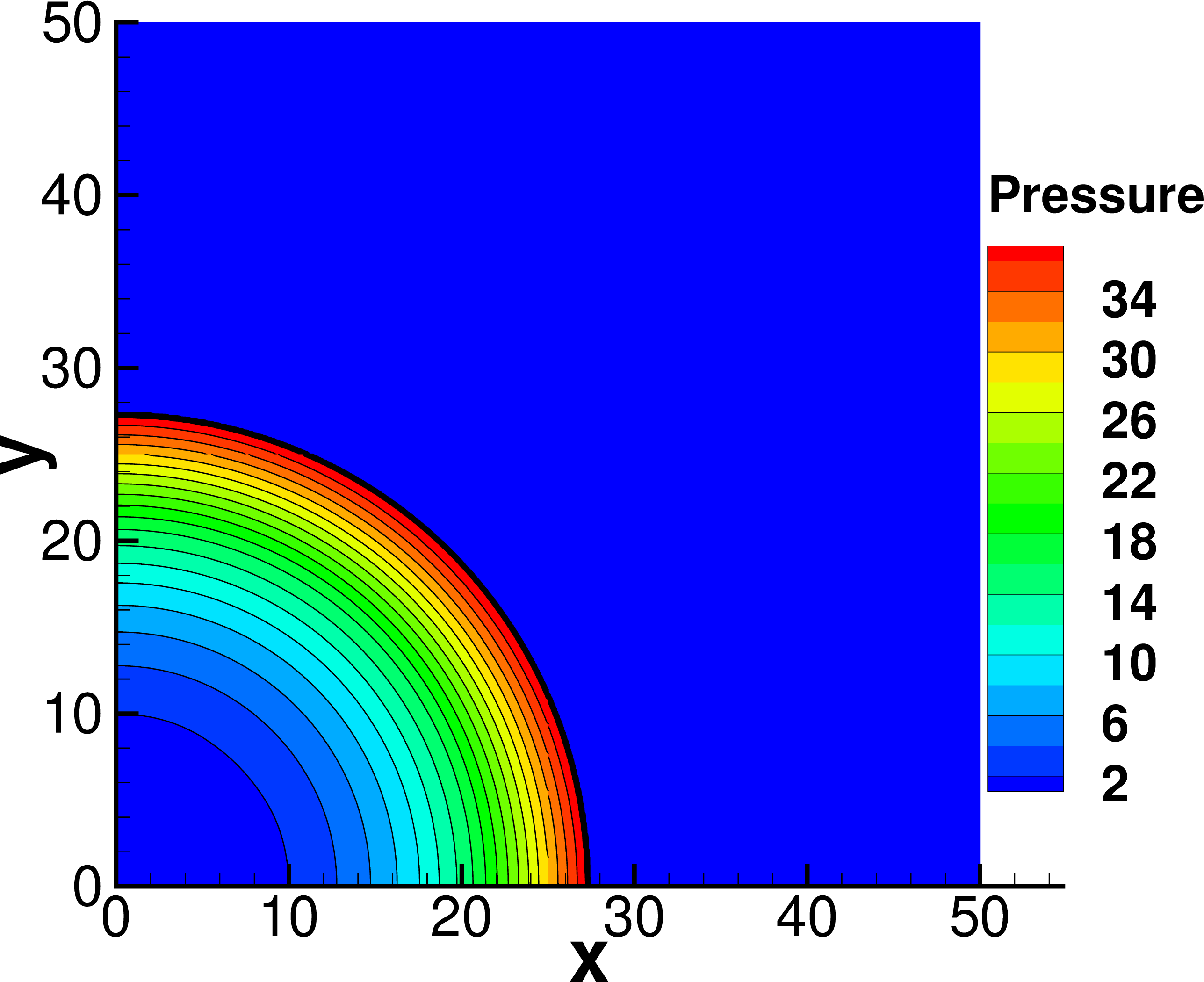}
        \caption{Pressure.}
    \end{subfigure}
    \begin{subfigure}[b]{0.45\textwidth}
        \includegraphics[width=\textwidth]{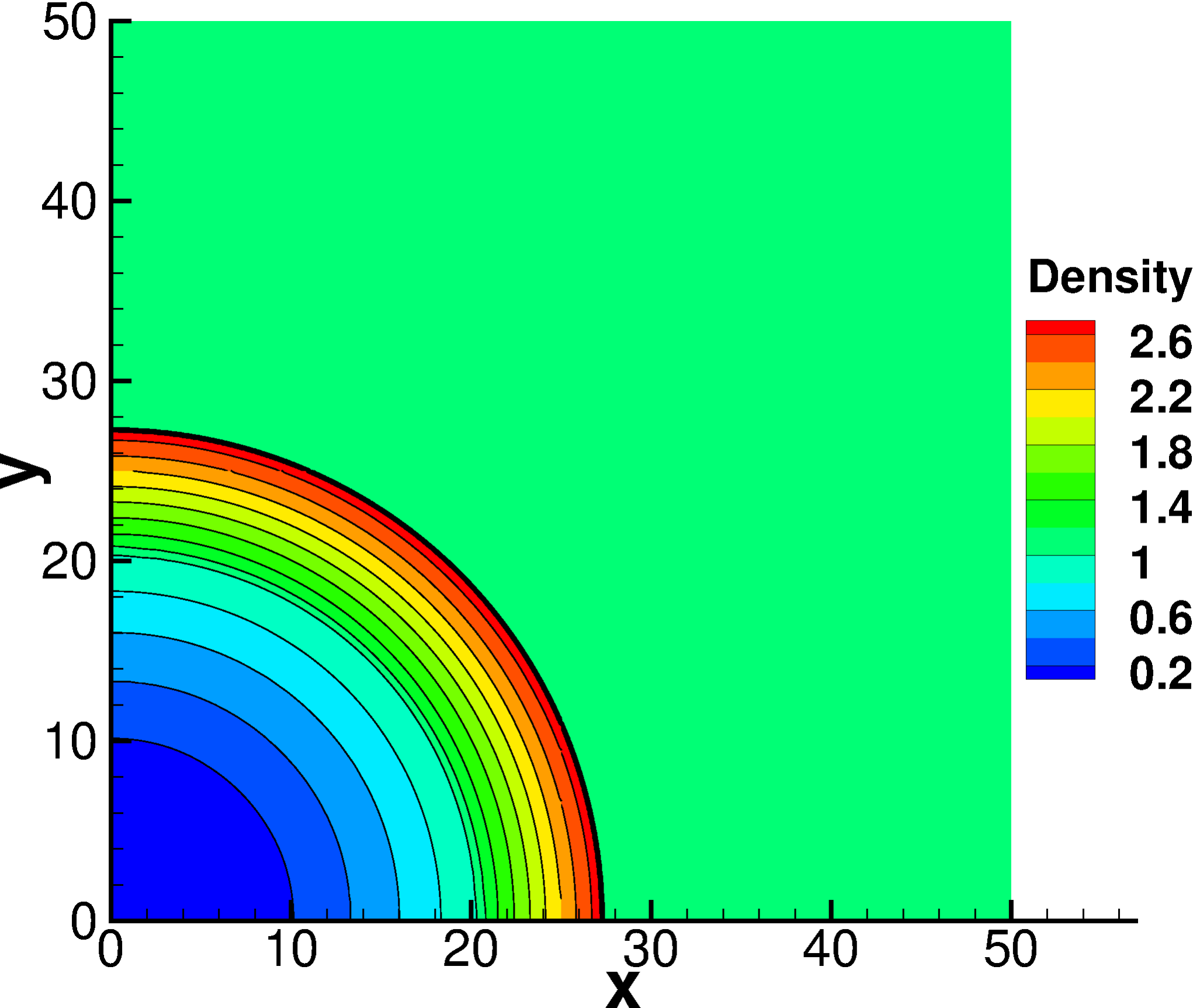}
        \caption{Density.}
    \end{subfigure}
    \begin{subfigure}[b]{0.45\textwidth}
        \includegraphics[width=\textwidth]{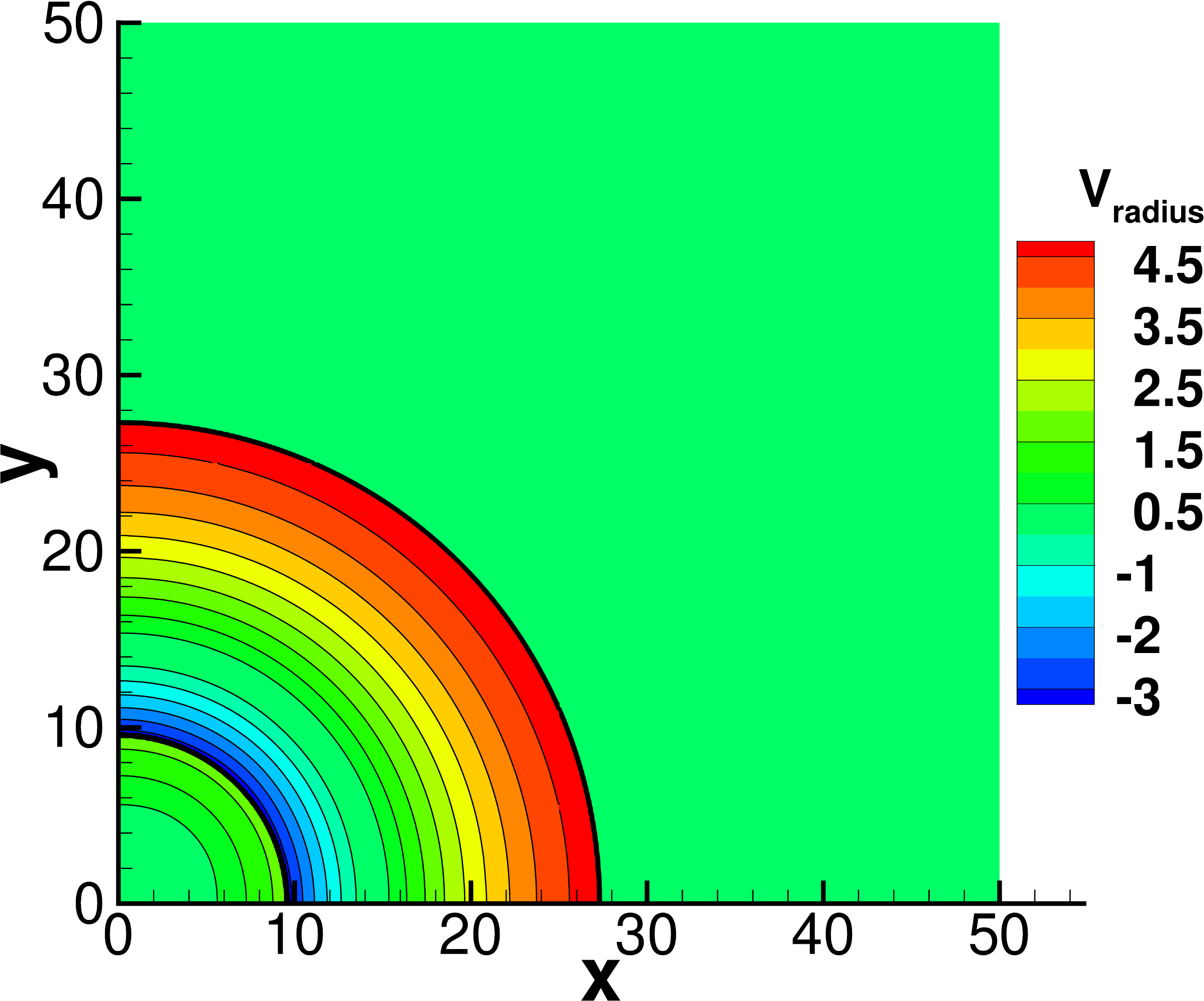}
        \caption{Radial velocity.}
    \end{subfigure}
\caption{Contours of 2D detonation problem of $CH_4$ at time $t = 2.0$. CFL = 0.2 and Cell size is $\Delta x = 1/12$.}
\label{fig:2dDetationCH4Contour}
 \end{figure}

 \begin{figure}[htbp]
    \centering
    \begin{subfigure}[b]{0.3\textwidth}
        \includegraphics[width=\textwidth]{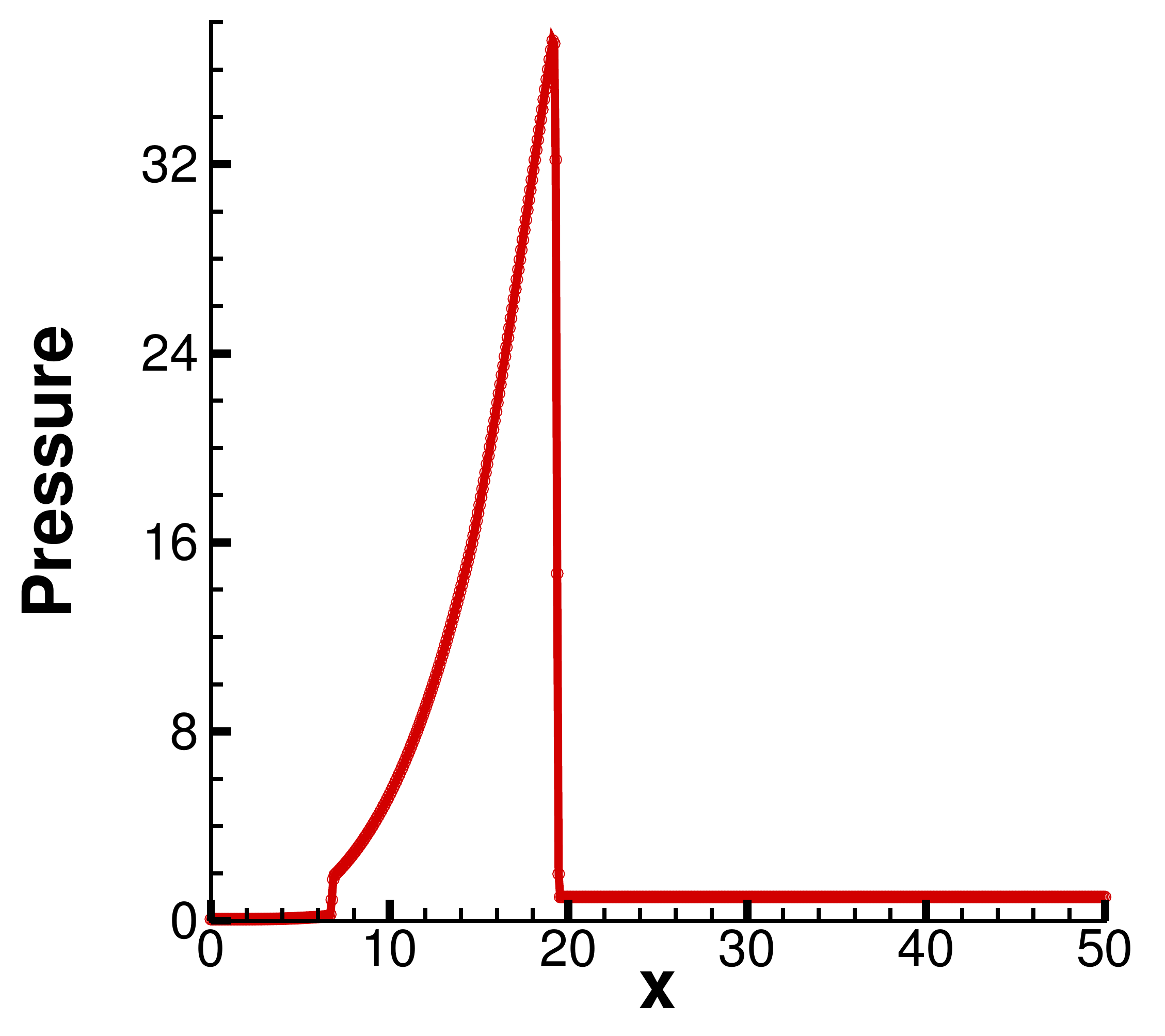}
        \caption{Pressure}
    \end{subfigure}
    \begin{subfigure}[b]{0.3\textwidth}
        \includegraphics[width=\textwidth]{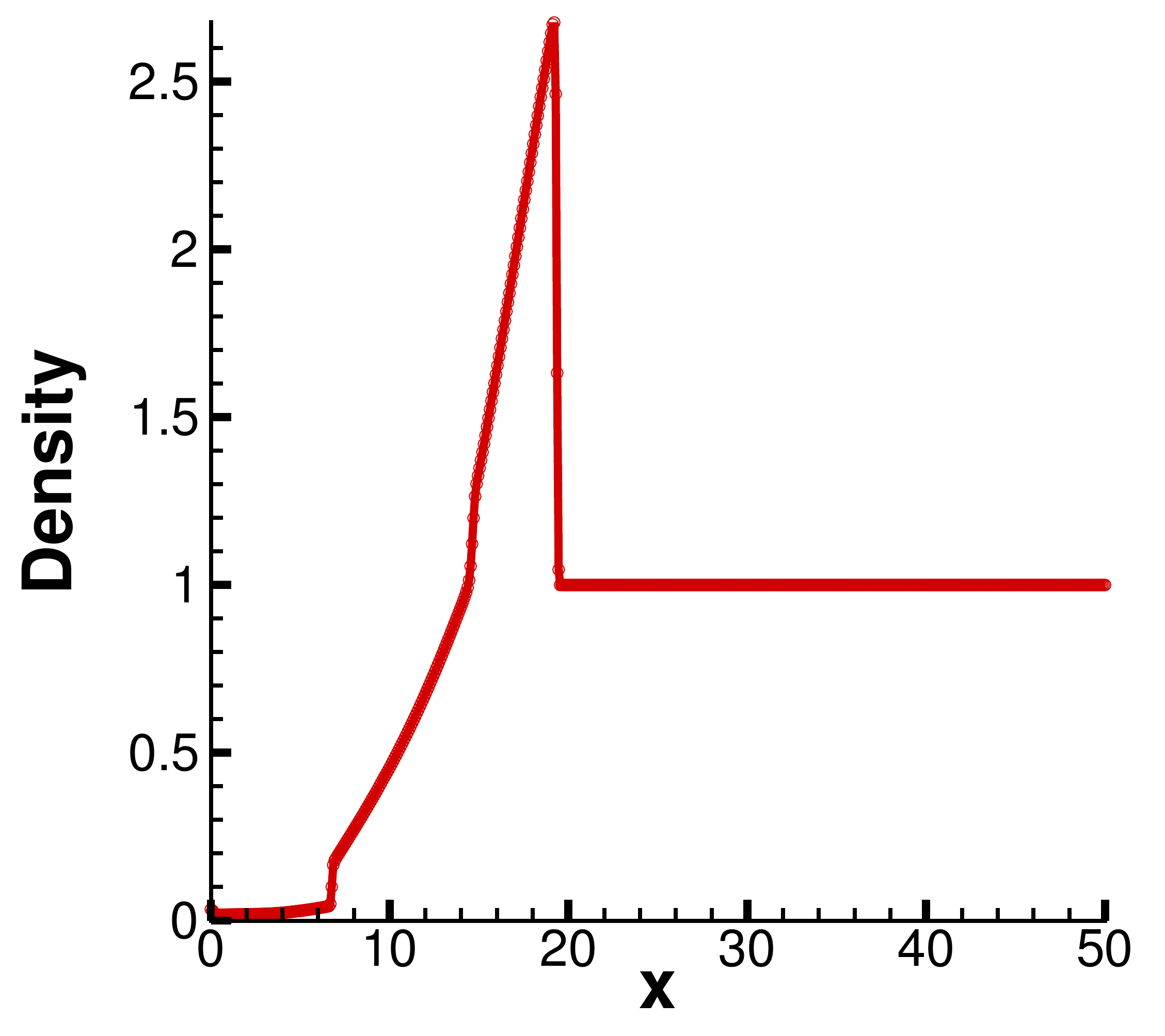}
        \caption{Density}
    \end{subfigure}
    \begin{subfigure}[b]{0.3\textwidth}
        \includegraphics[width=\textwidth]{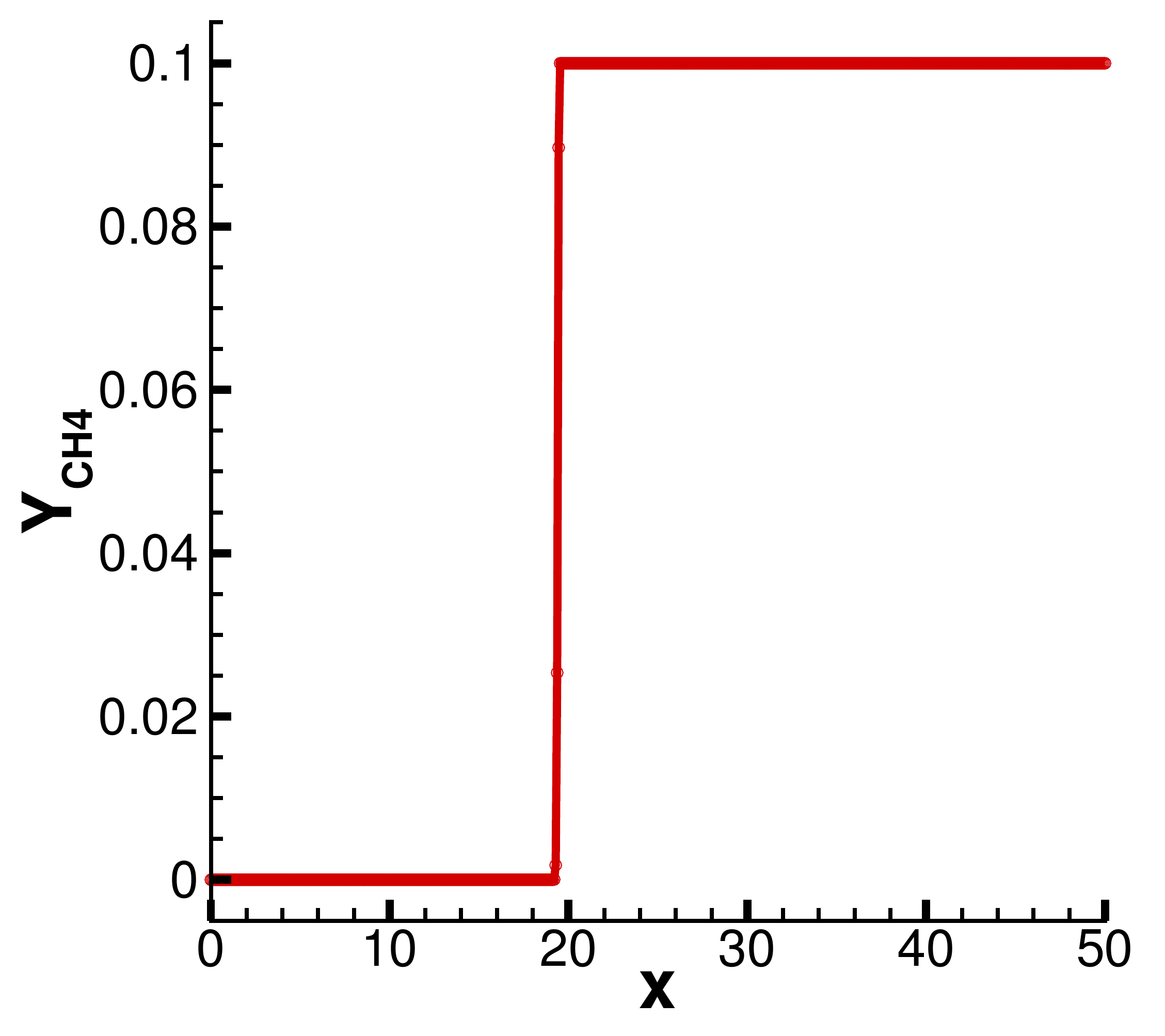}
        \caption{Mass fraction of $CH_4$}
    \end{subfigure}
    \begin{subfigure}[b]{0.3\textwidth}
        \includegraphics[width=\textwidth]{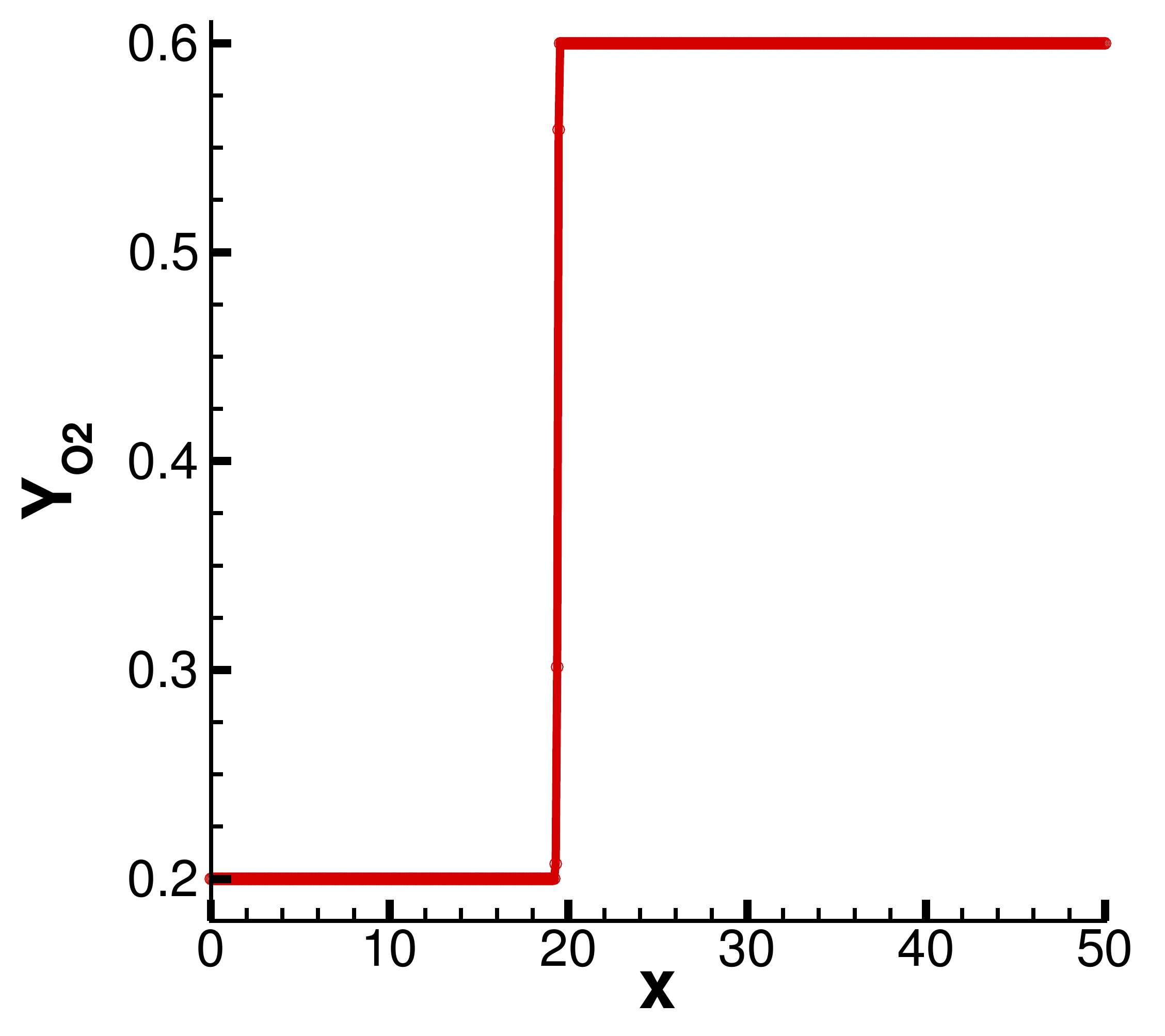}
        \caption{Mass fraction of $O_2$}
    \end{subfigure}
    \begin{subfigure}[b]{0.3\textwidth}
        \includegraphics[width=\textwidth]{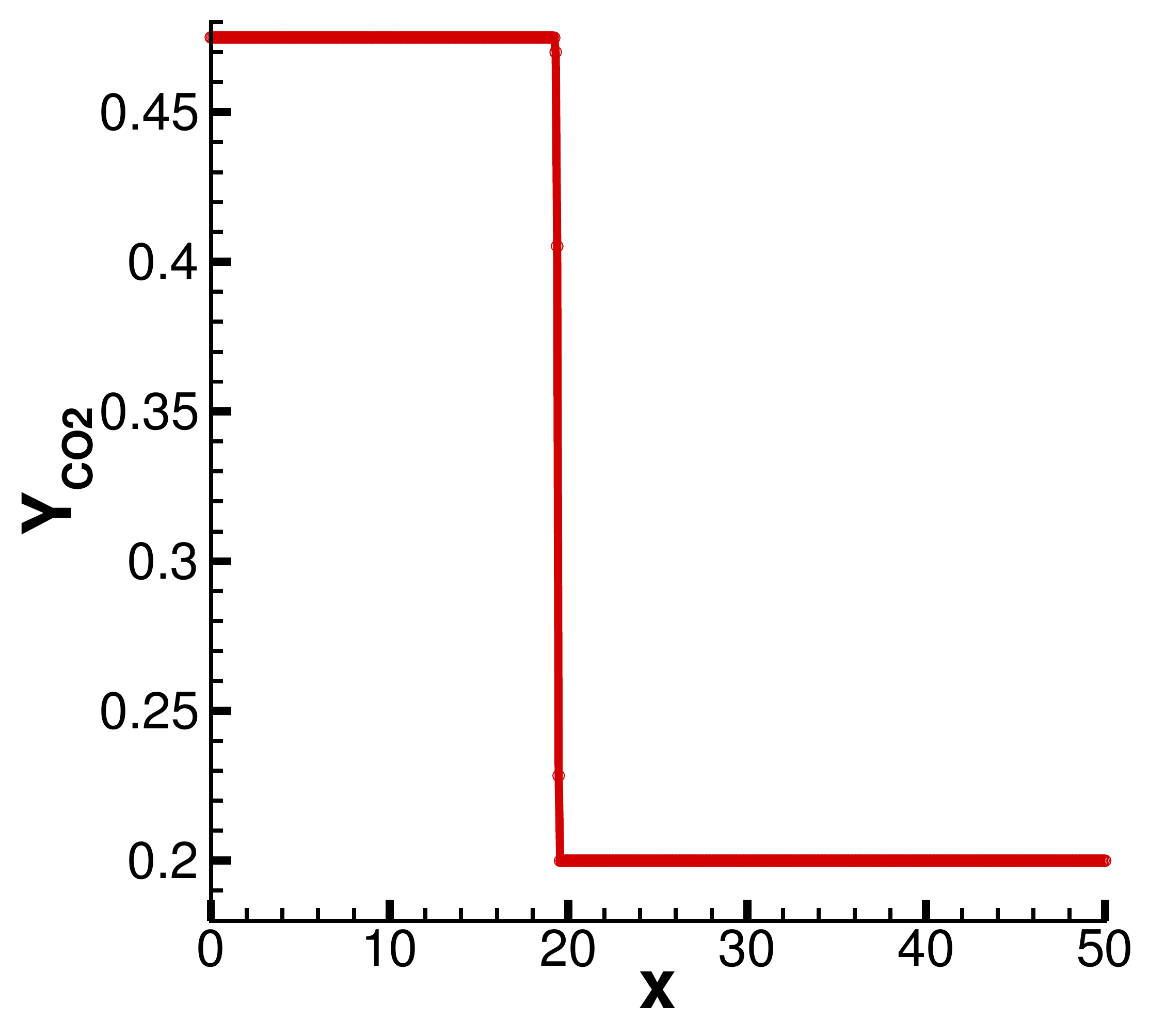}
        \caption{Mass fraction of $CO_2$}
    \end{subfigure}
    \begin{subfigure}[b]{0.3\textwidth}
        \includegraphics[width=\textwidth]{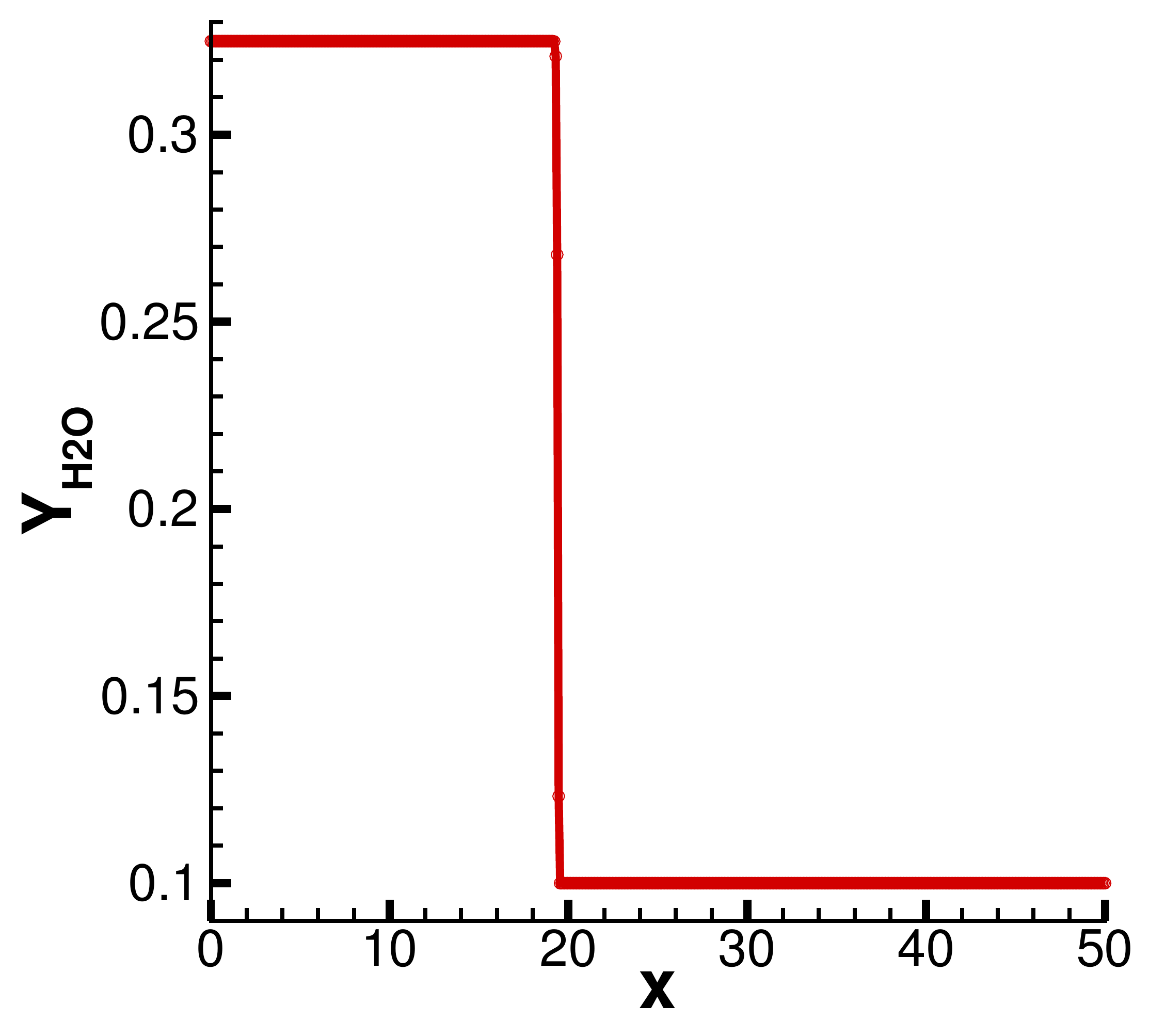}
        \caption{Mass fraction of $H_2O$}
    \end{subfigure}
\caption{Numerical solution of 2D detonation problem of $CH_4$ at time $t = 2.0$ across line $x=y$. CFL = 0.2. Cell size is $\Delta x = 1/12$.}
\label{fig:2dDetationCH4Line}
 \end{figure}
 \subsection{1D Detonation of $H_2$\label{sec:1d5species}}
 This problem involves 5 species and 2 reactions \cite{wang2015high} which can be modeled by
 \begin{equation}
    \begin{aligned}
        H_2 + O_2 &\rightarrow 2OH\,,\quad 2OH + H_2 & \rightarrow 2H_2O\,.
    \end{aligned}
 \end{equation}
$N_2$ plays the role of a catalyst. The corresponding reaction rate is
\begin{equation}
    \begin{aligned}
    R_1 & = B_1 T^{\alpha_1} H(T-T_1) \frac{\rho Y_{H_2}}{M_{H_2}}\frac{\rho Y_{O_2}}{M_{O_2}}\,,\\
    R_2 & = B_2 T^{\alpha_2} H(T-T_2) \frac{\rho Y_{H_2}}{M_{H_2}}\left(\frac{\rho Y_{OH}}{M_{OH}}\right)^2\,,
    \end{aligned}
\end{equation}
where $T_1 = 2.0$, $T_2 = 10.0$, $B_1 = B_2 = 1\times10^6$, $\alpha_1 = \alpha_2 = 0$, $h^0_{H_2} = h^0_{O_2} = h^0_{N_2} = 0$, $h^0_{OH} = -20$, $h^0_{H_2O} = -100$, $M_{H_2} = 2$, $M_{O_2} = 32$, $M_{OH} = 17$, $M_{H_2O} = 18$ and $M_{N_2} = 28$. $\gamma = 1.4$ and $T = p/\rho$. The computational domain is $[-25,25]$ and the initial values are
\begin{equation}
    \left( \rho, u, p, Y_{H_2}, Y_{O_2}, Y_{OH}, Y_{H_2O}, Y_{N_2} \right) = 
        \left\{
            \begin{aligned}
                &\left( 2, 10, 40, 0, 0, 0.17, 0.63, 0.2\right)\,,&x\leq -22.5\,,\\
                &\left( 1, 0, 1, 0.08, 0.72, 0, 0, 0.2\right)\,,&x > -22.5\,.
            \end{aligned}
            \right.
\end{equation}
Solutions with grid cell number $N = 600$ are shown in Fig. (\ref{fig:1dDetationH2O2}) where the solid lines are reference values obtained with $N = 4000$. Similarly, the exaction solution consists of a right-forward detonation wave, a contact discontinuity and a shock wave. The current Patankar scheme is able to capture the location of the detonation wave properly while retaining the positivity of density unconditionally.
 \begin{figure}
    \centering
    \begin{subfigure}[b]{0.4\textwidth}
        \includegraphics[width=\textwidth]{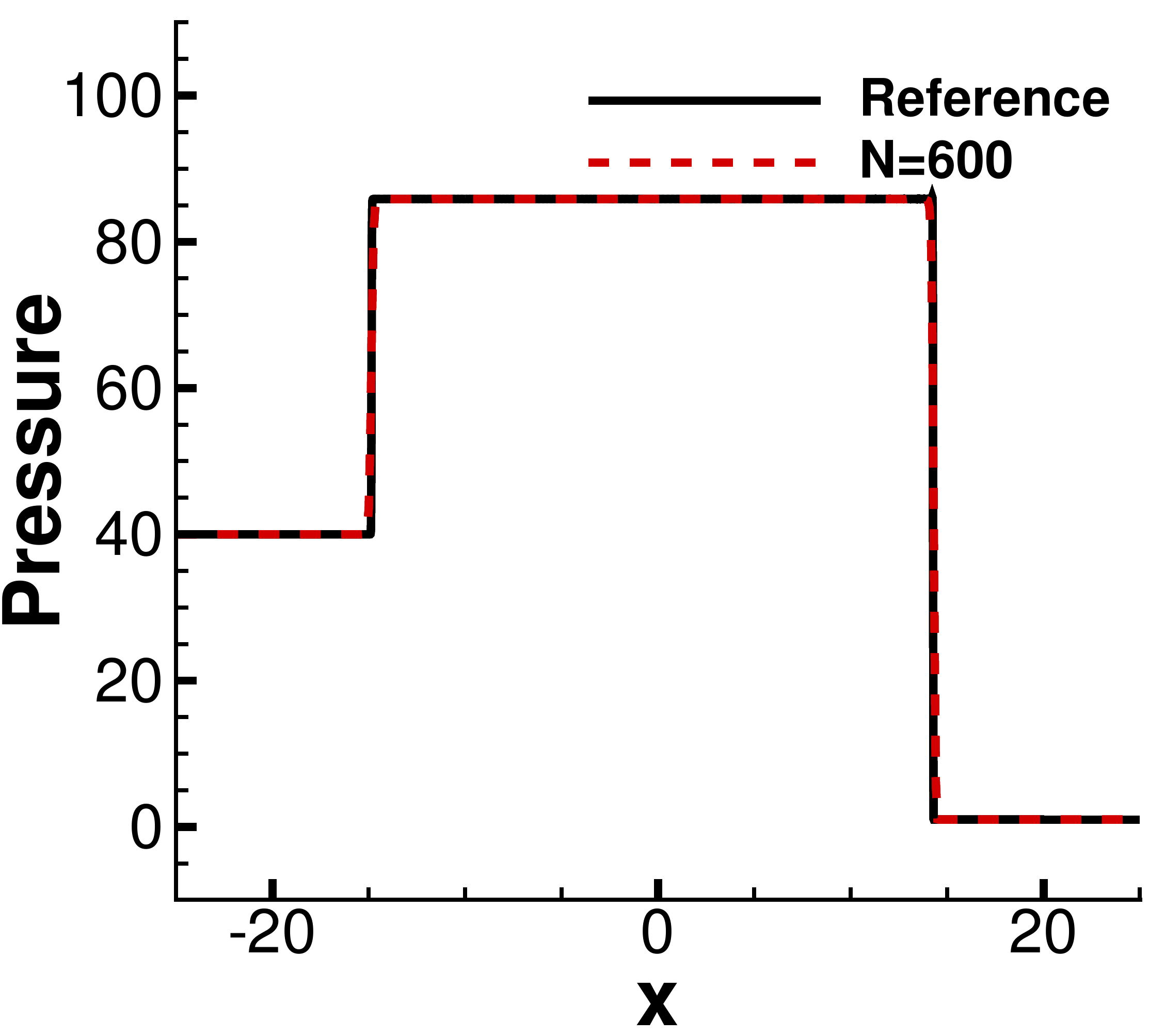}
    \end{subfigure}
    \begin{subfigure}[b]{0.4\textwidth}
        \includegraphics[width=\textwidth]{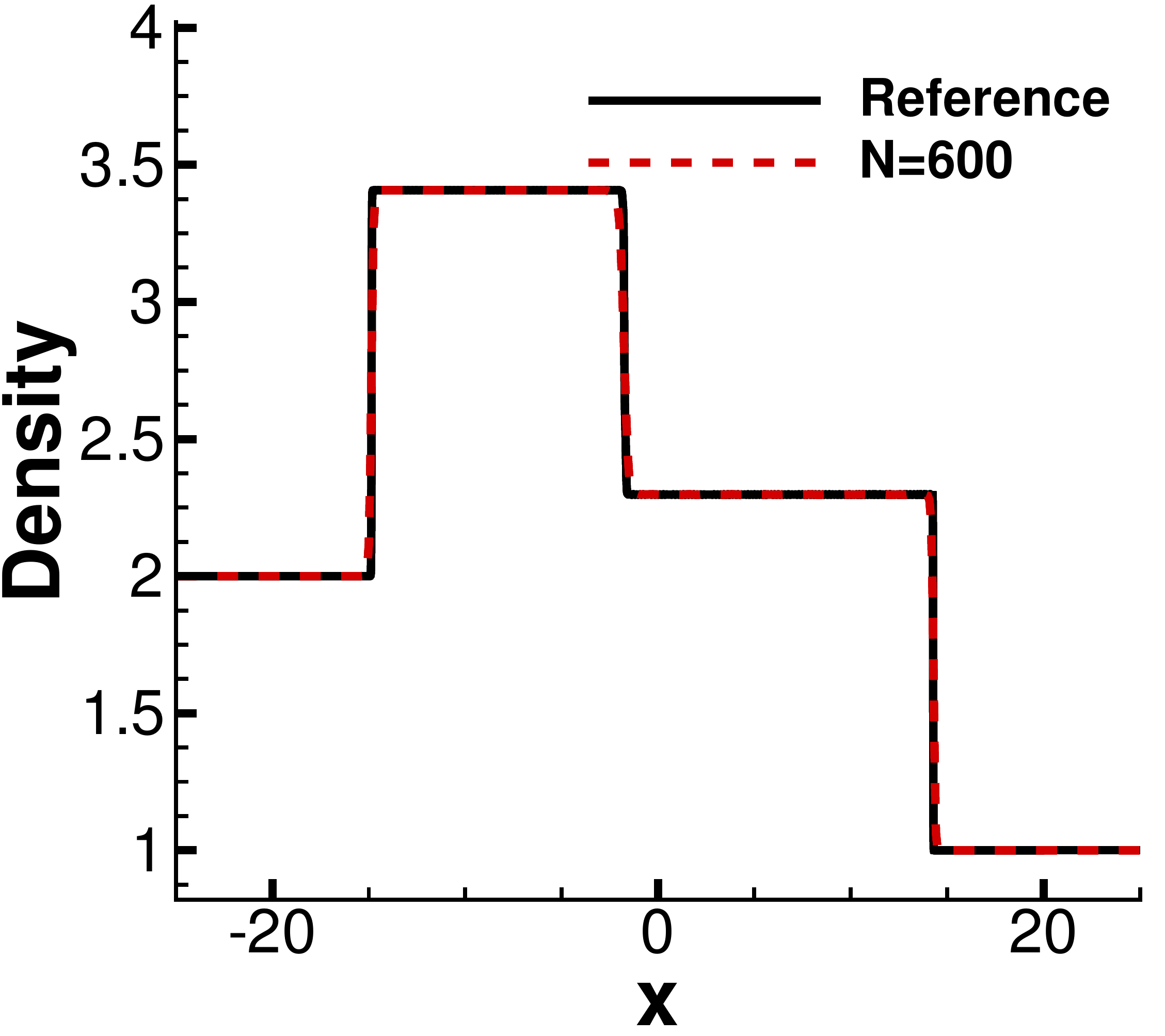}
    \end{subfigure}
    \begin{subfigure}[b]{0.4\textwidth}
        \includegraphics[width=\textwidth]{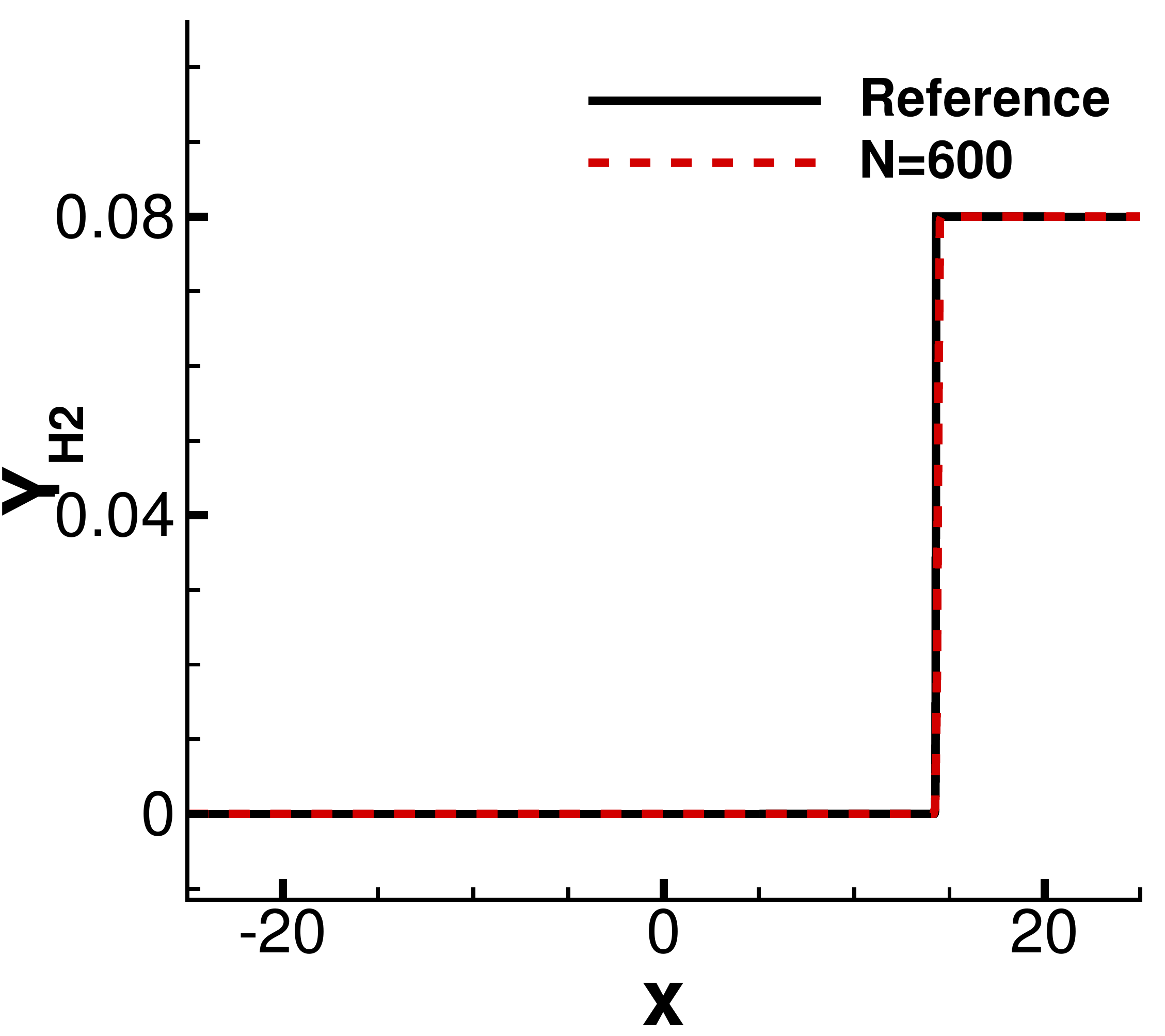}
    \end{subfigure}
    \begin{subfigure}[b]{0.4\textwidth}
        \includegraphics[width=\textwidth]{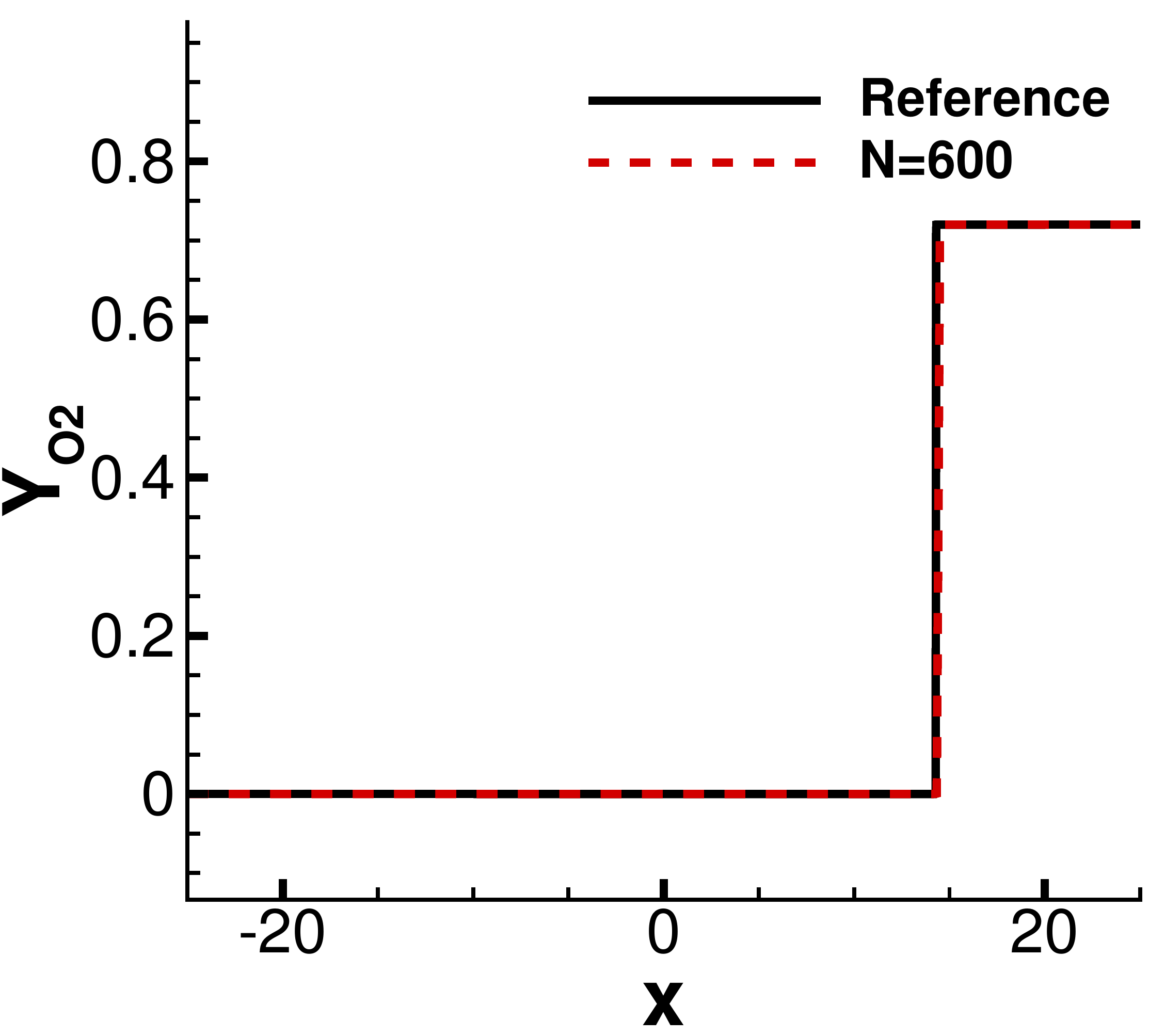}
    \end{subfigure}
\caption{Solutions of 1D detonation problem of $H_2$ at time $t = 3.0$. CFL = 0.2 and cell number $N = 600$.}
\label{fig:1dDetationH2O2}
 \end{figure}
 \subsection{2D Detonation of $H_2$\label{sec:2d5species}}
 A 2D problem \cite{wang2015high} with the same model set up as in Sec. \ref{sec:1d5species} is presented in this section. The computational domain is $[0,100]\times[0,25]$. Initial values are
\begin{equation}
    \left( \rho, u, p, Y_{H_2}, Y_{O_2}, Y_{OH}, Y_{H_2O}, Y_{N_2} \right) = 
        \left\{
            \begin{aligned}
                &\left( 2, 10, 40, 0, 0, 0.17, 0.63, 0.2\right)\,,&x\leq \xi(y)\,,\\
                &\left( 1, 0, 1, 0.08, 0.72, 0, 0, 0.2\right)\,,&x > \xi(y)\,,
            \end{aligned}
            \right.
\end{equation}
where
\begin{equation}
    \xi(y) = \left\{
        \begin{aligned}
            &5,& |y-12.5| > 7.5\,,\\
            &12.5-|y-12.5|,& |y-12.5| \leq 7.5\,.\\
        \end{aligned}
        \right.
\end{equation}
Computational domain of grid number $N_x \times N_y = 6000 \times 1500$ is used to fully resolve the detonation wave without spurious pressure oscillations ahead of the detonation wave. Inflow and outflow boundary conditions are applied on the left and right boundaries, respectively. The top and bottom boundaries are inviscid walls. 

Firstly, the solutions along the cross-section $y=12.5$ at $t = 2.0$ are presented in Fig. \ref{fig:2dDetationH2O2Line}. The detonation wave is captured sharply. At such a resolution, a weak shock wave can be observed after the detonation wave which is not found in the results of Du and Yang \cite{Du2019} and Wang et al. \cite{wang2015high}.

Secondly, the evolution of density contours from $t = 1$ to $t = 4$ is shown in Fig. \ref{fig:2dDetationH2O2}. An important feature of this problem is the existence of triple points, which travel along the detonation front in the transverse direction and reflect from the upper and lower boundaries, forming a cellular pattern. Behind the detonation front, there is a contact discontinuity presented as a jet and a strong shock, all of which moving right-forward.

 During the simulation, the proposed PMPRK scheme is able to preserve the positivity of all species unconditionally as expected.
 \begin{figure}
    \centering
    \begin{subfigure}[b]{0.4\textwidth}
        \includegraphics[width=\textwidth]{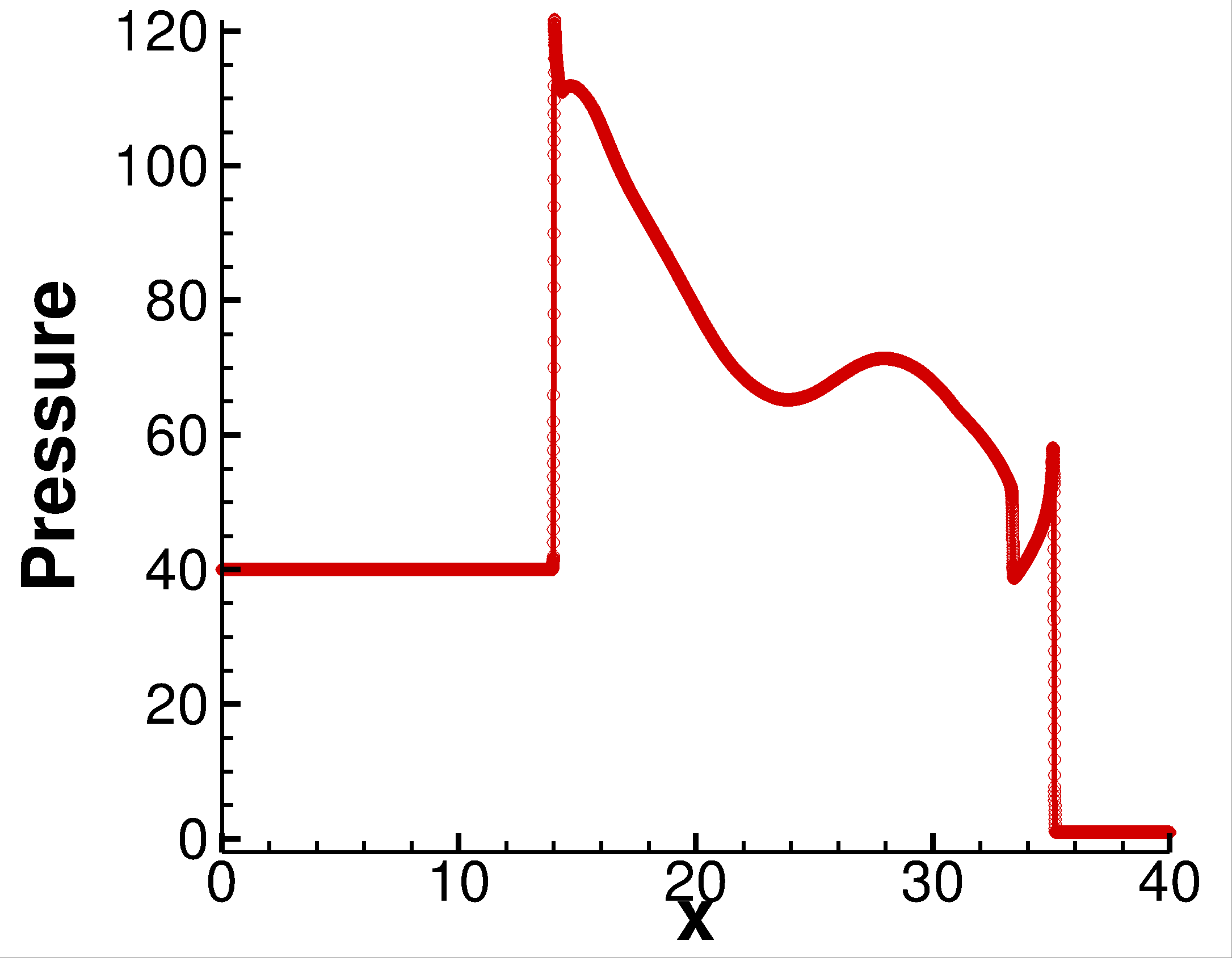}
    \end{subfigure}
    \begin{subfigure}[b]{0.4\textwidth}
        \includegraphics[width=\textwidth]{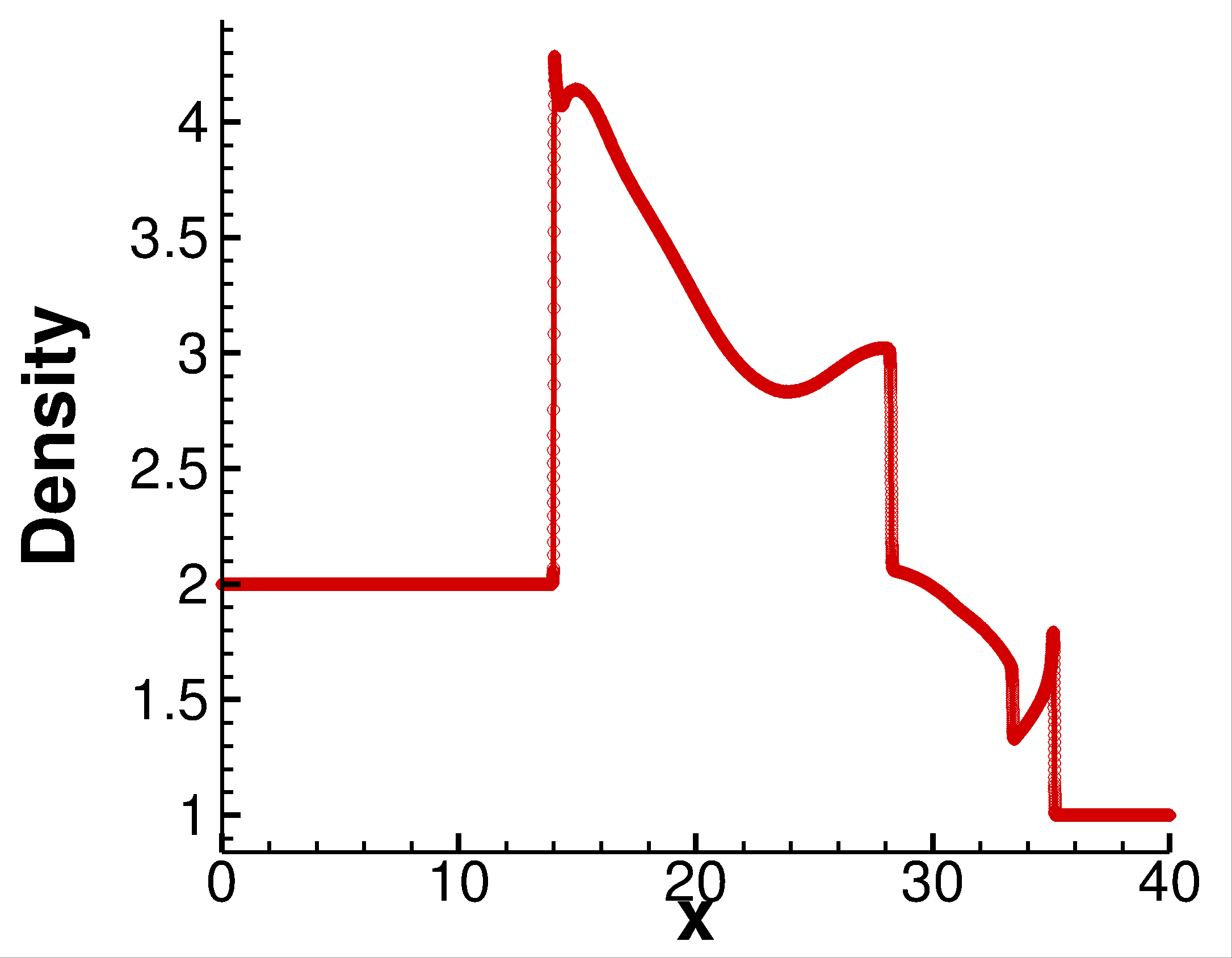}
    \end{subfigure}
    \begin{subfigure}[b]{0.4\textwidth}
        \includegraphics[width=\textwidth]{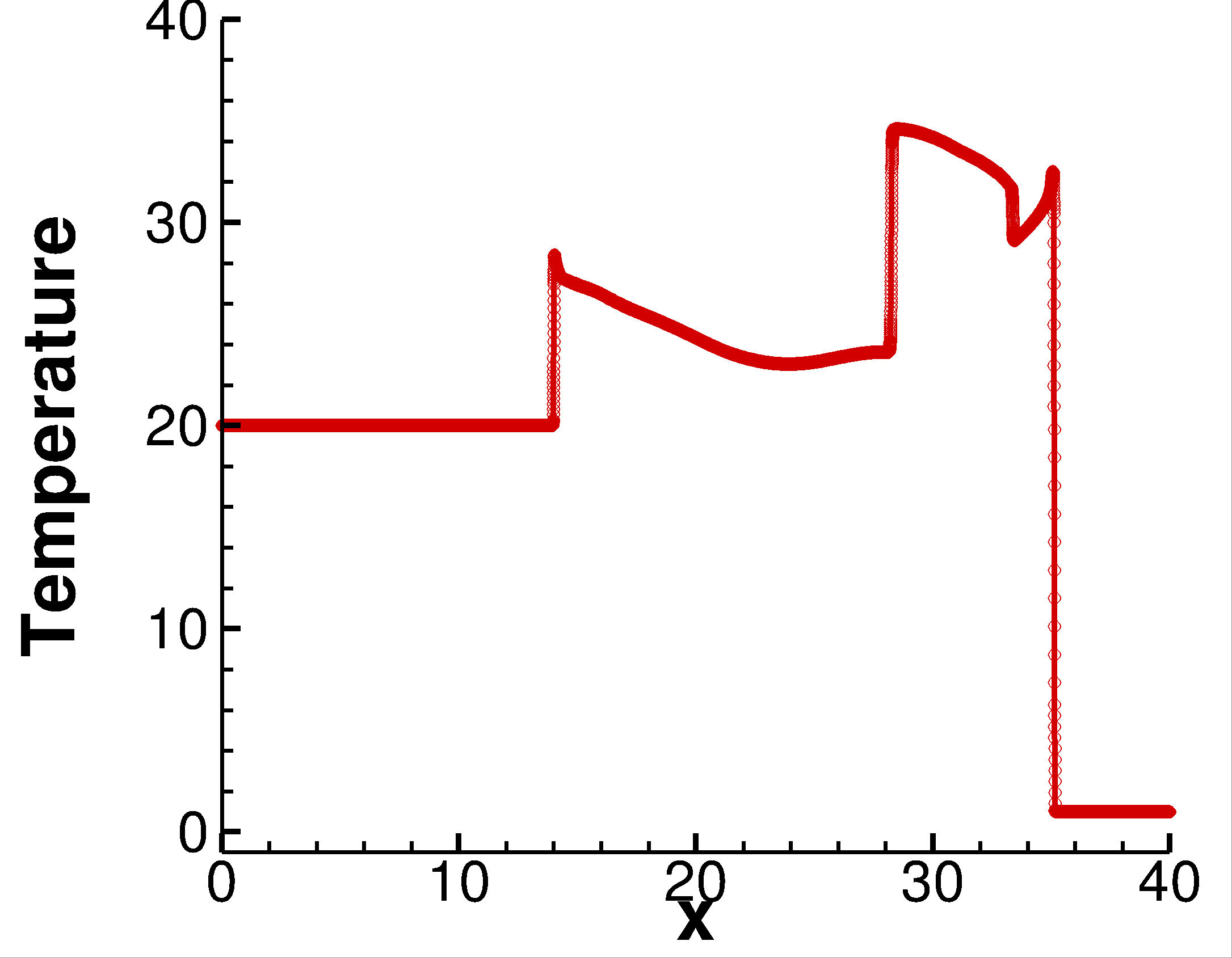}
    \end{subfigure}
    \begin{subfigure}[b]{0.4\textwidth}
        \includegraphics[width=\textwidth]{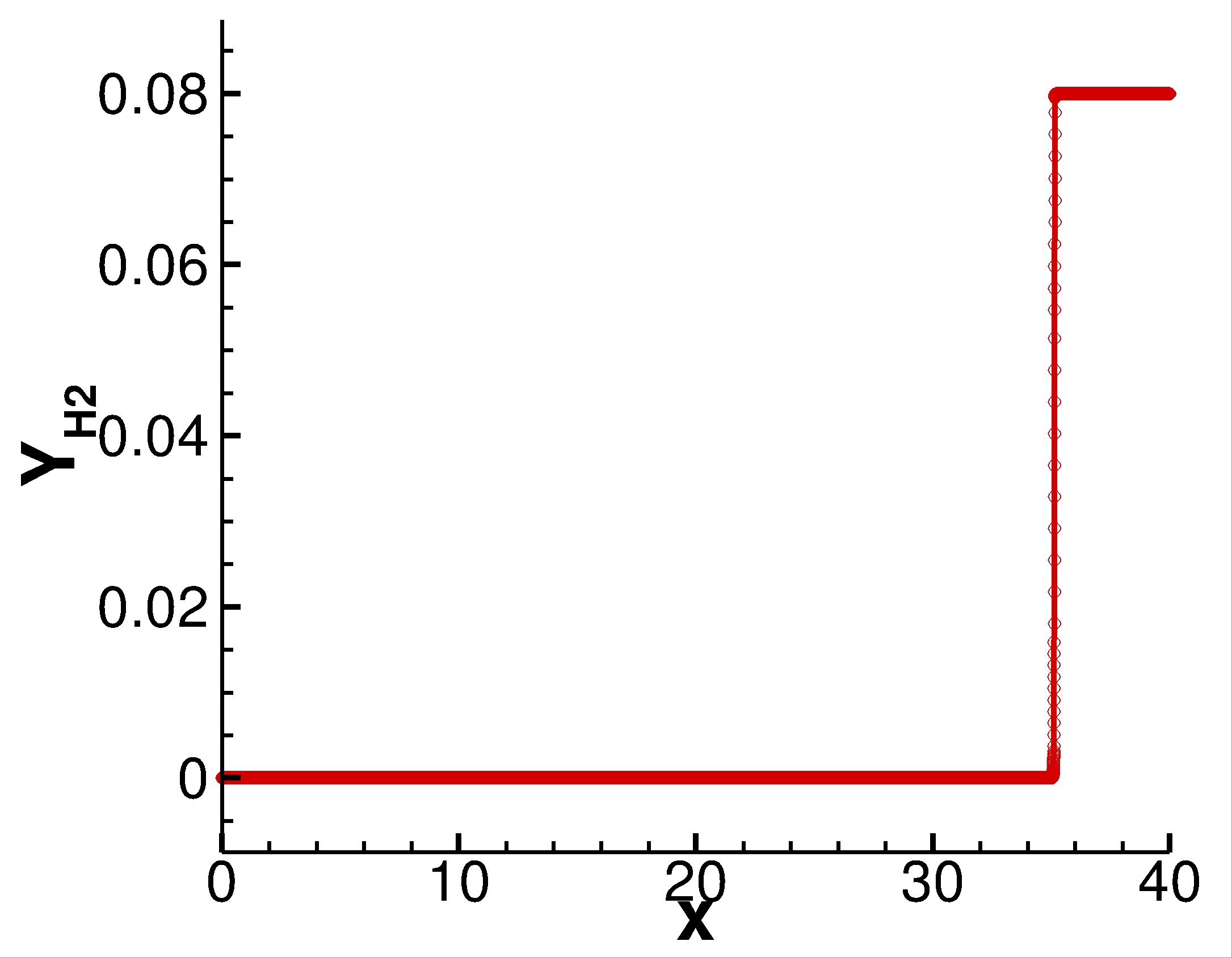}
    \end{subfigure}
    \begin{subfigure}[b]{0.4\textwidth}
        \includegraphics[width=\textwidth]{./2dDetonationH2O2Y1Line.png}
    \end{subfigure}
    \begin{subfigure}[b]{0.4\textwidth}
        \includegraphics[width=\textwidth]{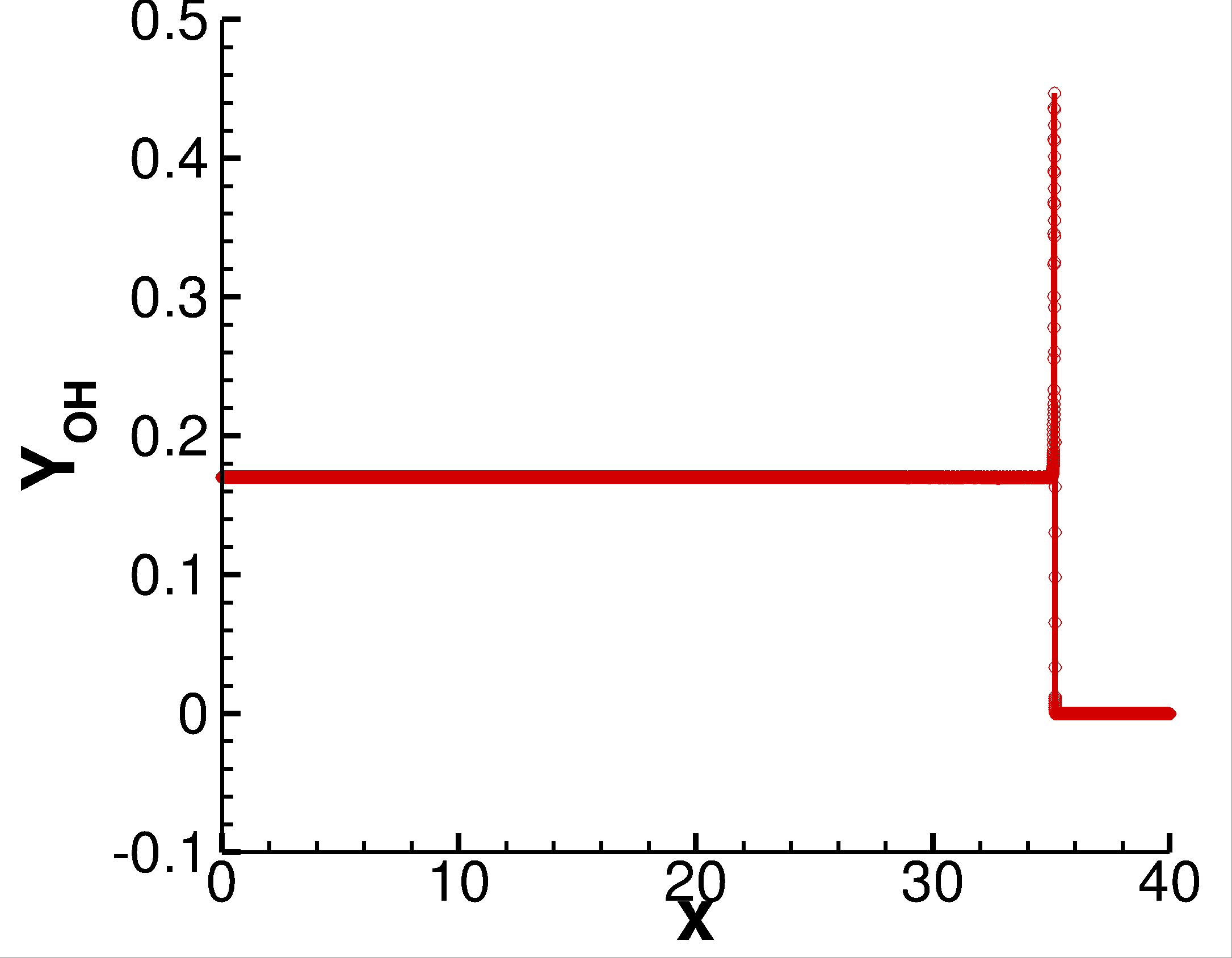}
    \end{subfigure}
    \begin{subfigure}[b]{0.4\textwidth}
        \includegraphics[width=\textwidth]{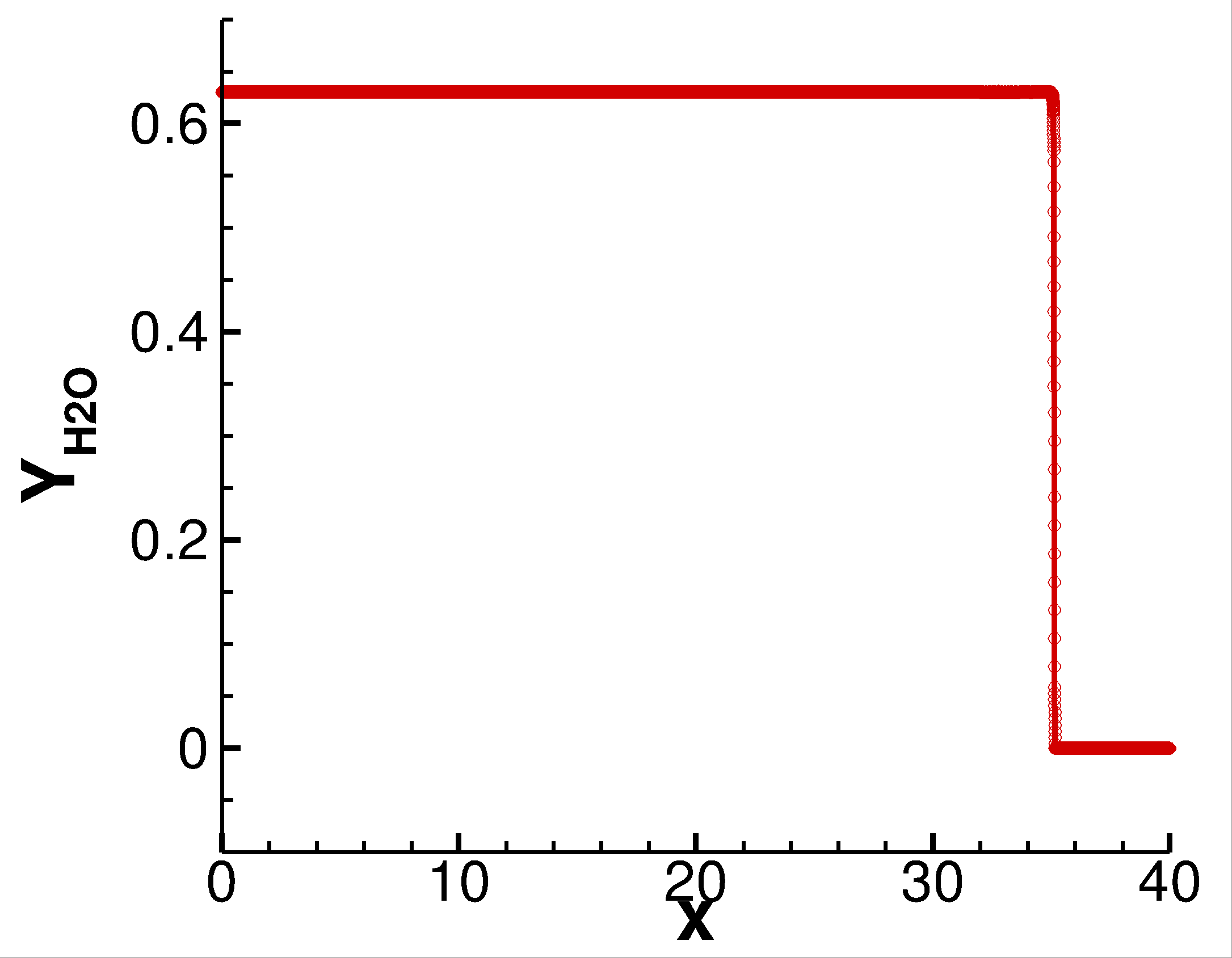}
    \end{subfigure}
    \begin{subfigure}[b]{0.4\textwidth}
        \includegraphics[width=\textwidth]{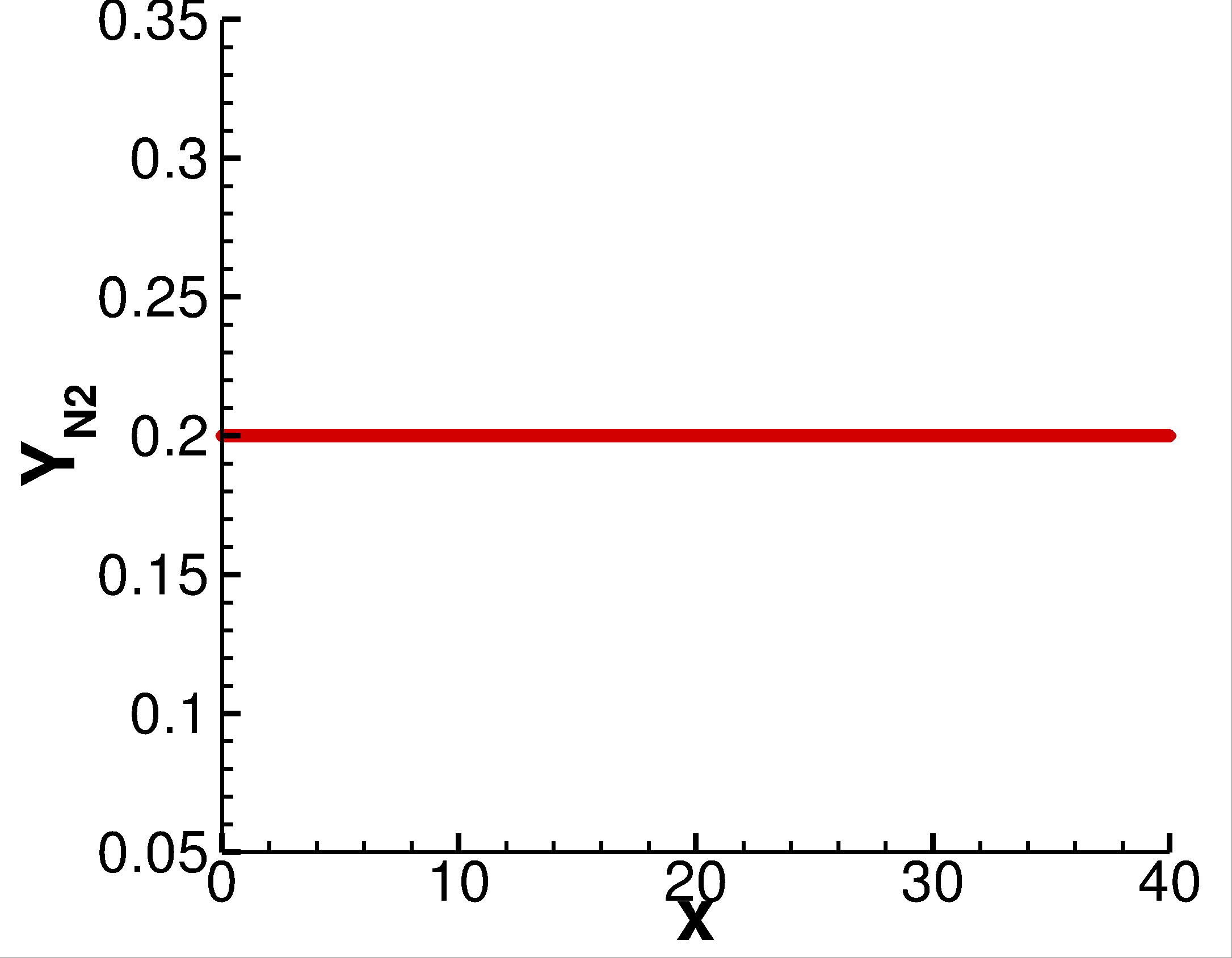}
    \end{subfigure}
\caption{Solutions along cross-section of $y=12.5$ of 2D detonation problem of $H_2$ at time $t = 2$. CFL = 0.4 and grid number is $N_x\times N_y = 6000\times1500$.}
\label{fig:2dDetationH2O2Line}
\end{figure}
 \begin{figure}
    \centering
    \begin{subfigure}[b]{0.4\textwidth}
        \includegraphics[width=\textwidth]{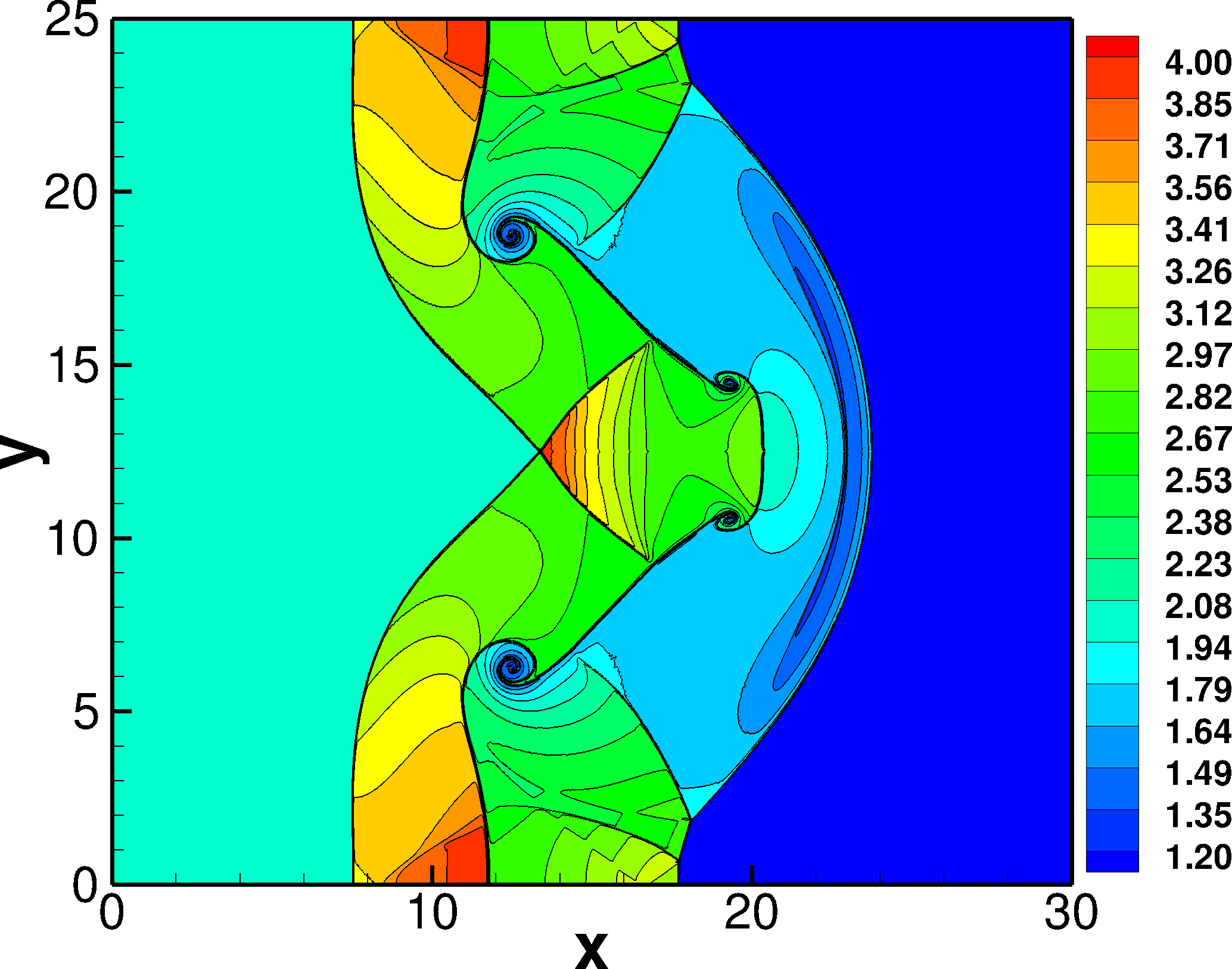}
        \caption{$t = 1$.}
    \end{subfigure}
    \begin{subfigure}[b]{0.4\textwidth}
        \includegraphics[width=\textwidth]{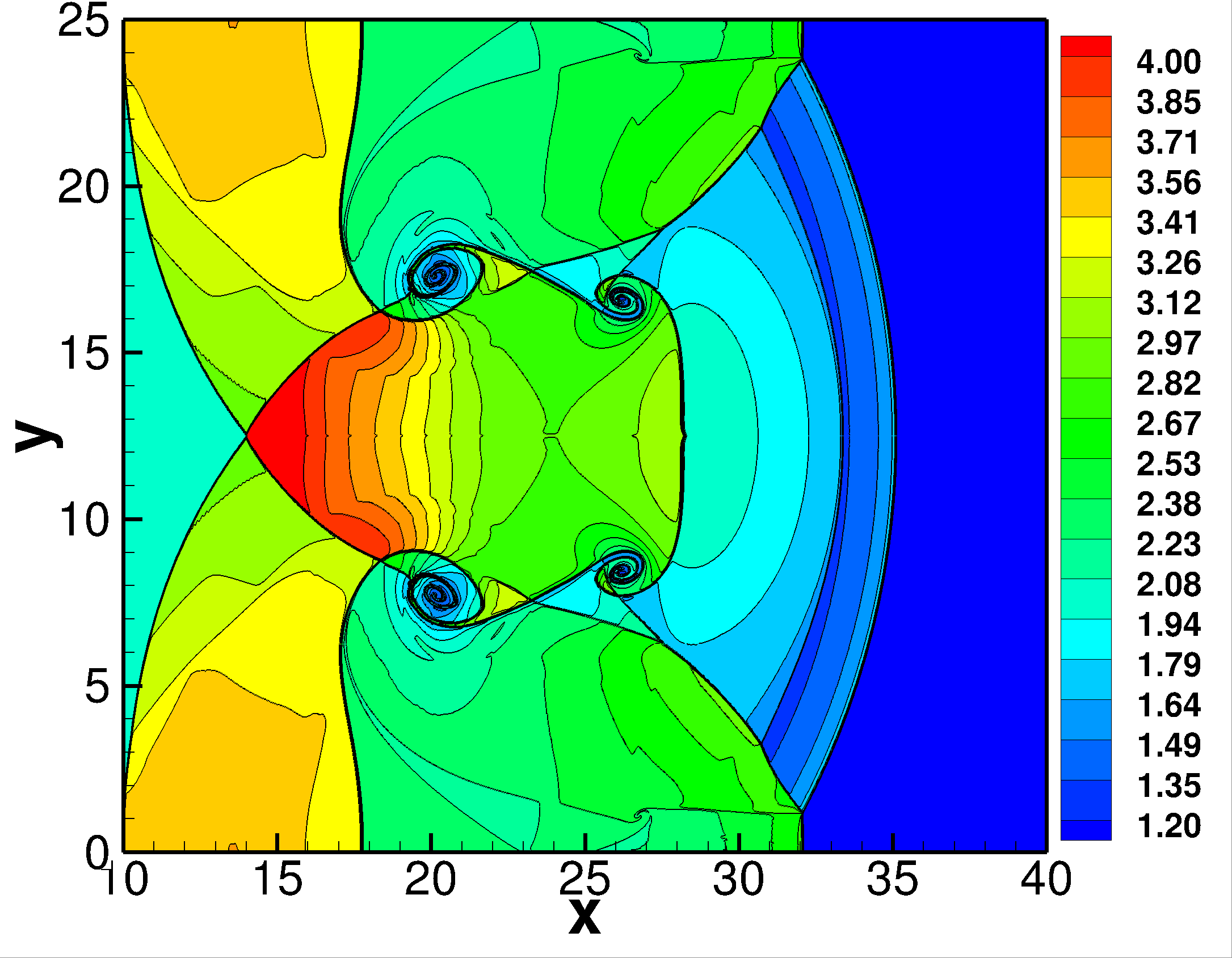}
        \caption{$t = 2$.}
    \end{subfigure}
    \begin{subfigure}[b]{0.4\textwidth}
        \includegraphics[width=\textwidth]{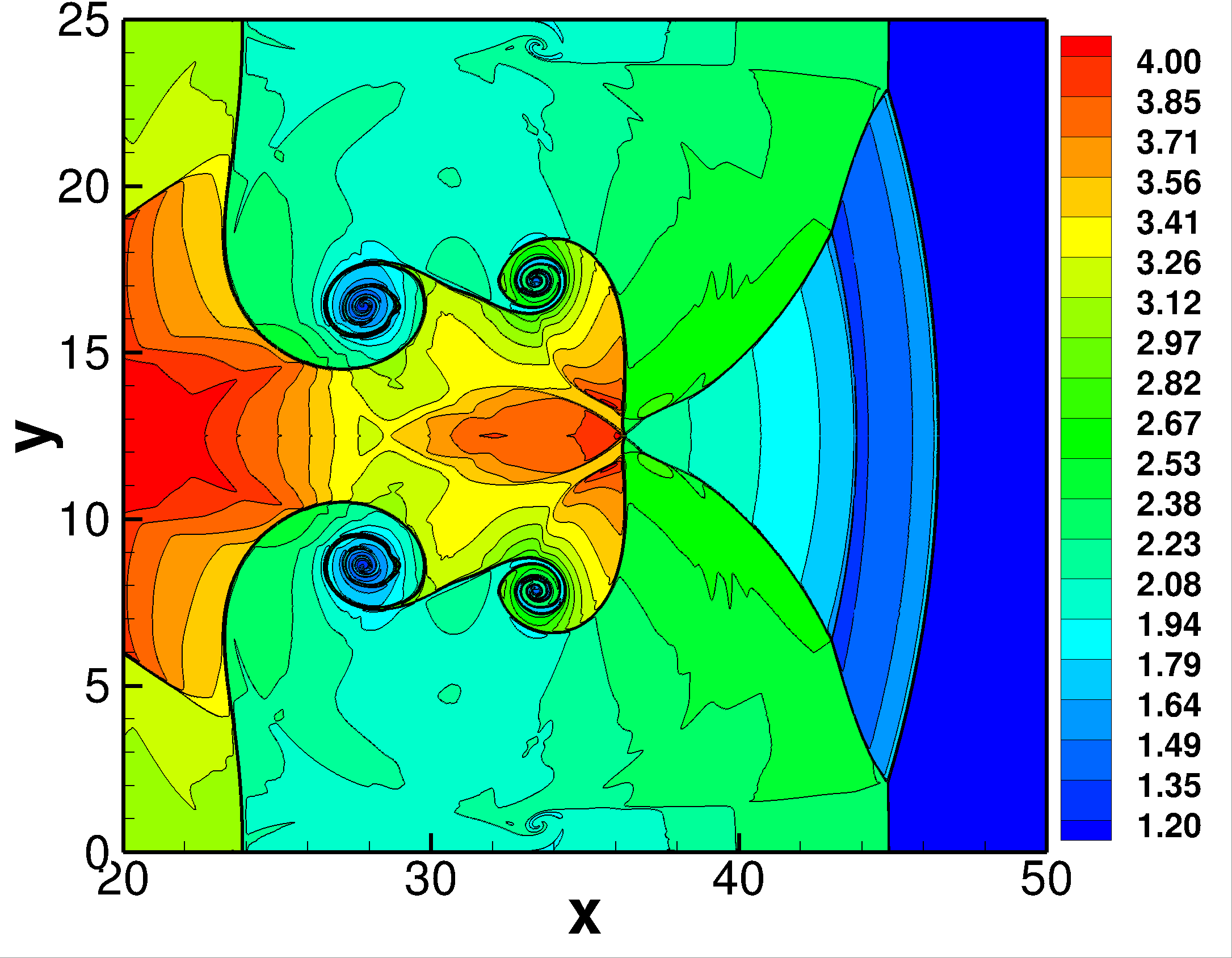}
        \caption{$t = 3$.}
    \end{subfigure}
    \begin{subfigure}[b]{0.4\textwidth}
        \includegraphics[width=\textwidth]{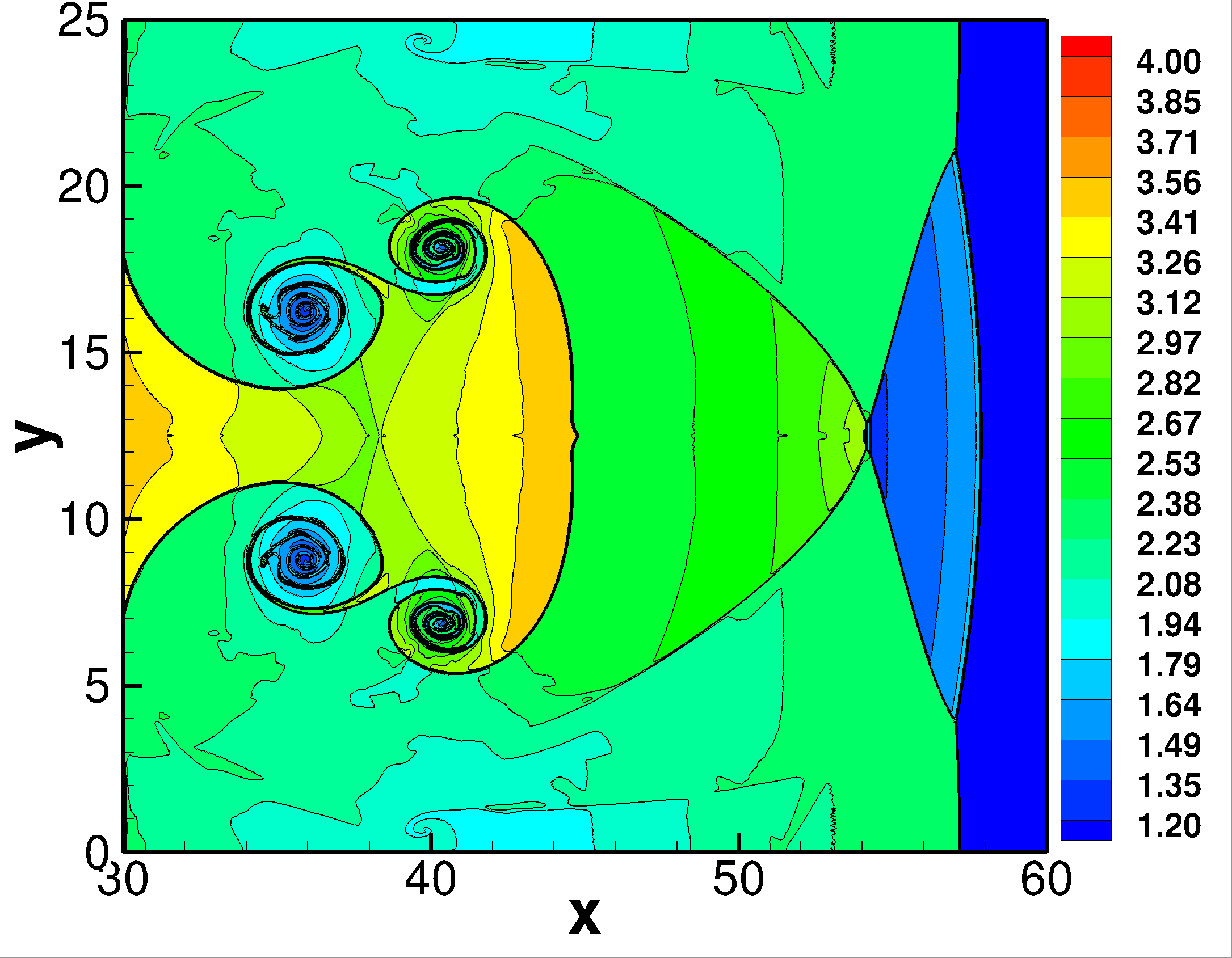}
        \caption{$t = 4$.}
    \end{subfigure}
\caption{Density contours of 2D detonation problem of $H_2$ from time $t = 1$ to $4$. CFL = 0.4 and grid number are $N_x\times N_y = 6000\times1500$.}
\label{fig:2dDetationH2O2}
\end{figure}

\subsection{1D Reversible Reactions of Oxygen\label{sec:1dOxygen}}
To illustrate that the proposed scheme is able to preserve the positivity of pressure, a problem with a couple of reversible reactions \cite{zhang2012positivity} is simulated. And one of two reactions is endothermic. Considering
\begin{equation}
    O_2\leftrightarrow 2O\,,
\end{equation}
with $N_2$ appears as a catalyst. The reaction rate is
\begin{equation}
    \begin{aligned}
    R_f & = k_f \left(\frac{\rho Y_{O_2}}{M_{O_2}}\right)\left(\frac{\rho Y_O}{M_O} +\frac{\rho Y_{O_2}}{M_{O_2}} +\frac{\rho Y_{N_2}}{M_{N_2}} \right)\,,\\
    R_b & = k_b \left(\frac{\rho Y_O}{M_O}\right)^2\left(\frac{\rho Y_O}{M_O} +\frac{\rho Y_{O_2}}{M_{O_2}} +\frac{\rho Y_{N_2}}{M_{N_2}} \right)\,,\\
    \end{aligned}
\end{equation}
where $ k_f = C T^{-2} e^{-E/T}$, $k_b = k_f/\mathrm{exp}(b_1+b_2 \mathrm{log} z + b_3 z + b_4 z^2 + b_5 z^5)$ and $z = 10000/T$. The parameters are $C = 2.9 \times 10^{17}$, $E = 59750$, $b_1 = 2.855$, $b_2 = 0.988$, $b_3 = -6.181$, $b_4 = -0.023$ and $b_5 = -0.001$. The mass per mole of the species are $M_{O} = 0.016$, $M_{O_2} = 0.032$ and $M_{N_2} = 0.028$. $h^0_{O} = 1.558\times10^7$ and $h^0_{O_2} = h^0_{N_2} = 0$. In this case, the specific heat ratios for each specie are set as
\[
    \gamma_{O} = \frac{5}{3}, \gamma_{O_2} = \gamma_{N_2} = \frac{7}{5}.
\] 
Pressure, temperature and internal energy are correlated according to the idea gas law, i.e., $p = \rho R_g T = \frac{\rho e}{\gamma-1}$, with the mixing rules defined as
\[
 \frac{1}{M} = \frac{Y_O}{M_{O}} + \frac{Y_{O_2}}{M_{O_2}} + \frac{Y_{N_2}}{M_{N_2}}\,,
\]
\[
 \frac{1}{\gamma-1} = \frac{1}{\gamma_{O}-1}\frac{M_O}{M} + \frac{1}{\gamma_{O_2}-1}\frac{M_{O_2}}{M} + \frac{1}{\gamma_{N_2}-1}\frac{M_{N_2}}{M},
\]
\[
    R_g = \frac{R}{M},
\]
and $R = 8.314\,\,\mathrm{J/\left(mol\cdot K\right)}$.

Computational domain is $[-1,1]$ consisting of 4000 uniformly distributed cells. Initial values are
\begin{equation}
    \left( u, p, \rho Y_{O}, \rho Y_{O_2}, \rho Y_{N_2} \right) = 
        \left\{
            \begin{aligned}
                &\left( 0, 1000, \rho Y_1^L, \rho Y_2^L, \rho Y_3^L \right)\,,&x\leq 0\,,\\
                &\left( 0, 1, \rho Y_1^R, \rho Y_2^R, \rho Y_3^R \right)\,,&x > 0\,,\\
            \end{aligned}
            \right.
\end{equation}
where 
\[
    \begin{aligned}
\rho Y_1^L = 5.251896311257204\times 10^{-5}, \quad \rho Y_2^L = 3.748071704863518\times10^{-5},\\
\rho Y_3^L = 2.962489471973072\times10^{-4}, \quad  \rho Y_1^R = 8.341661837019181\times 10^{-8},\\
\rho Y_2^R = 9.455418692098664\times 10^{-11},\quad \rho Y_3^R =2.748909430004963\times 10^{-7}\,.
    \end{aligned}
\]
Due to the formation of shock and rarefaction wave, the original equilibrium state is perturbed as shown by the mass fraction of $O$ and $O_2$ in Fig.\ref{fig:1dOxygen}. In Fig. \ref{fig:1dOxygen}, the reference values are solutions at the grid resolution of $\Delta x = 1/20000$, i.e., with grid number of $N = 40000$. All the structures are well-resolved and a converged solution is obtained when $N = 4000$. Particularly, the positivity of the pressure is preserved during the calculation, which is a significant improvement over the original Patankar scheme.
 \begin{figure}
    \centering
    \begin{subfigure}[b]{0.4\textwidth}
        \includegraphics[width=\textwidth]{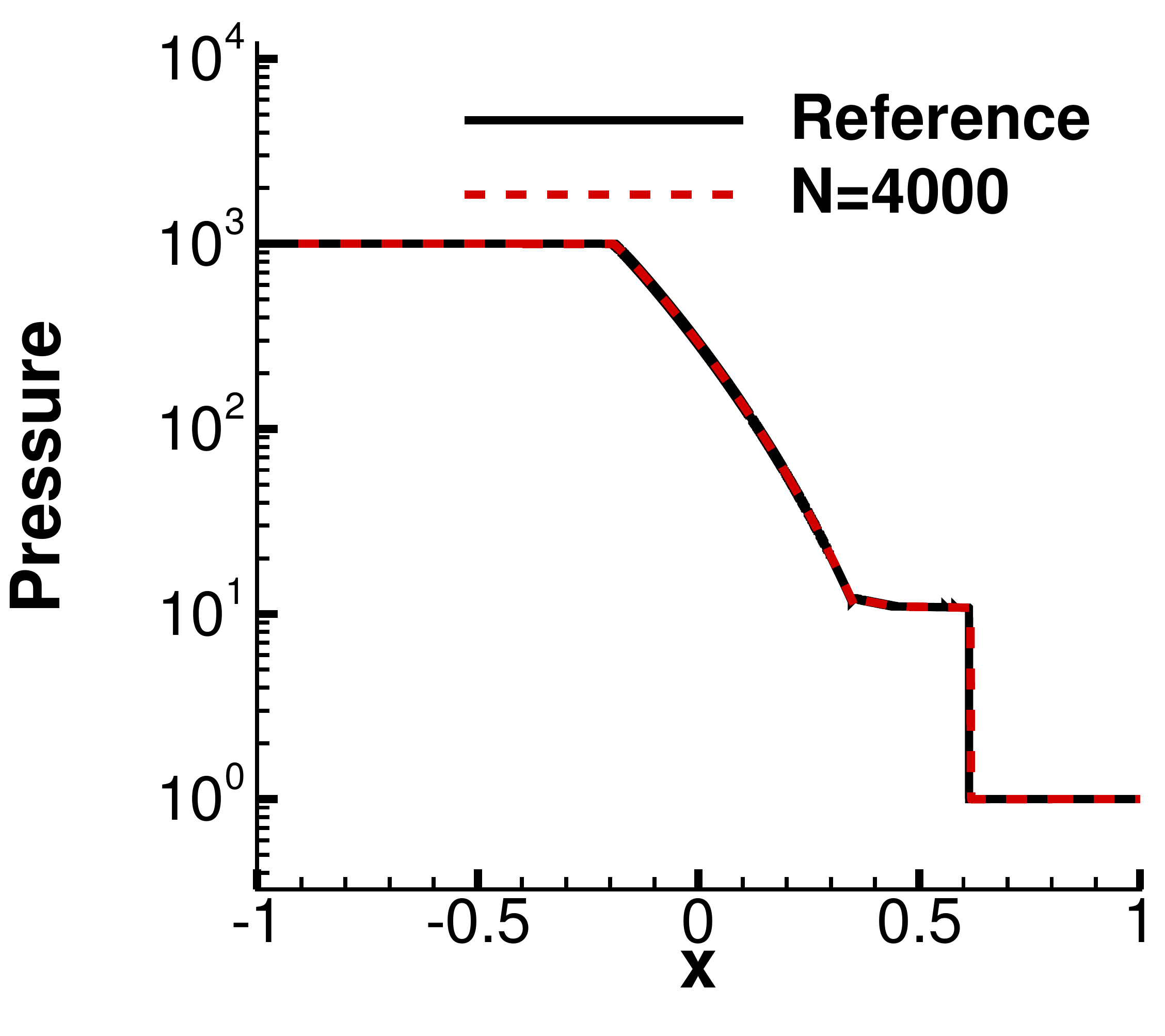}
    \end{subfigure}
    \begin{subfigure}[b]{0.4\textwidth}
        \includegraphics[width=\textwidth]{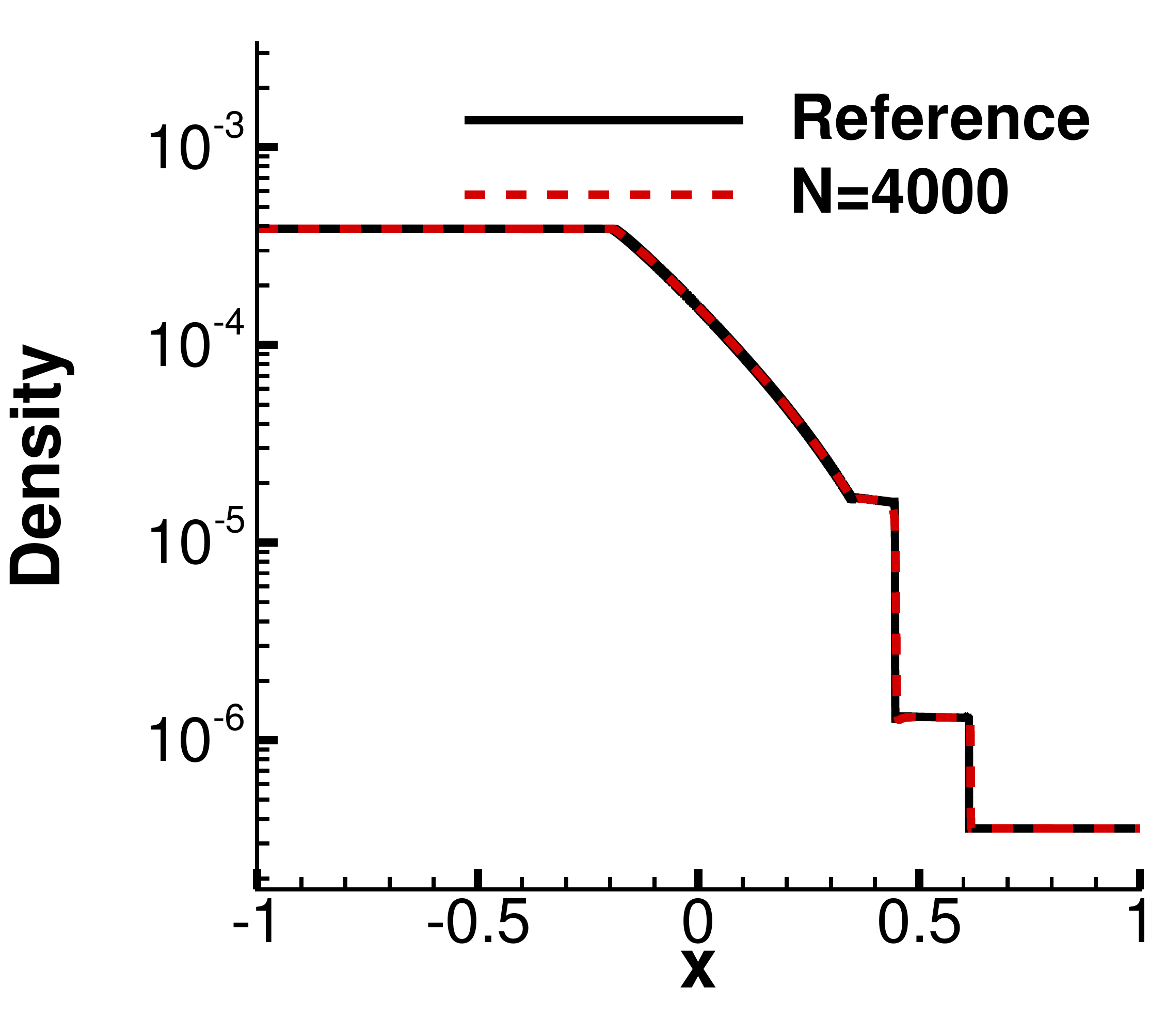}
    \end{subfigure}
    \begin{subfigure}[b]{0.4\textwidth}
        \includegraphics[width=\textwidth]{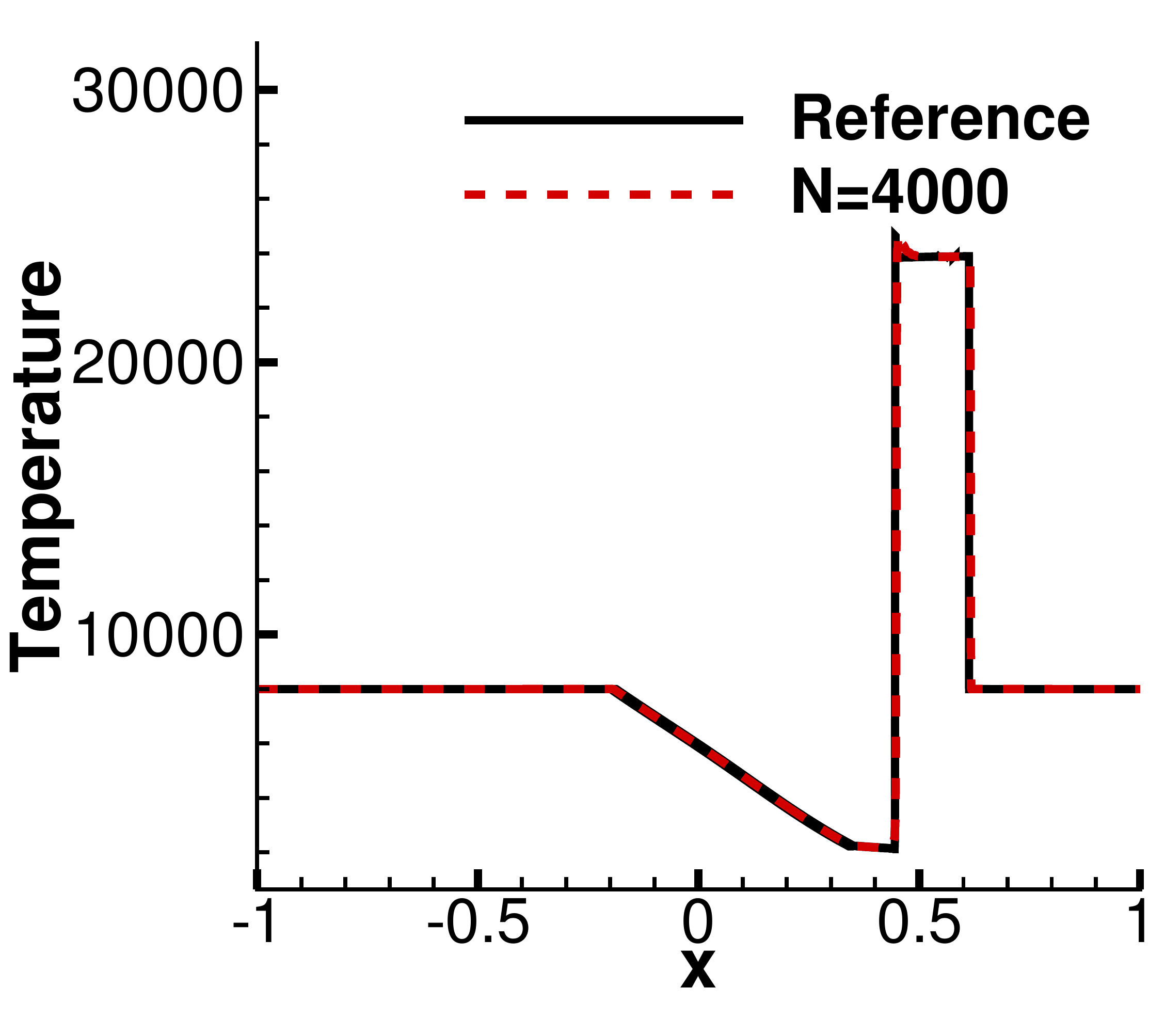}
    \end{subfigure}
    \begin{subfigure}[b]{0.4\textwidth}
        \includegraphics[width=\textwidth]{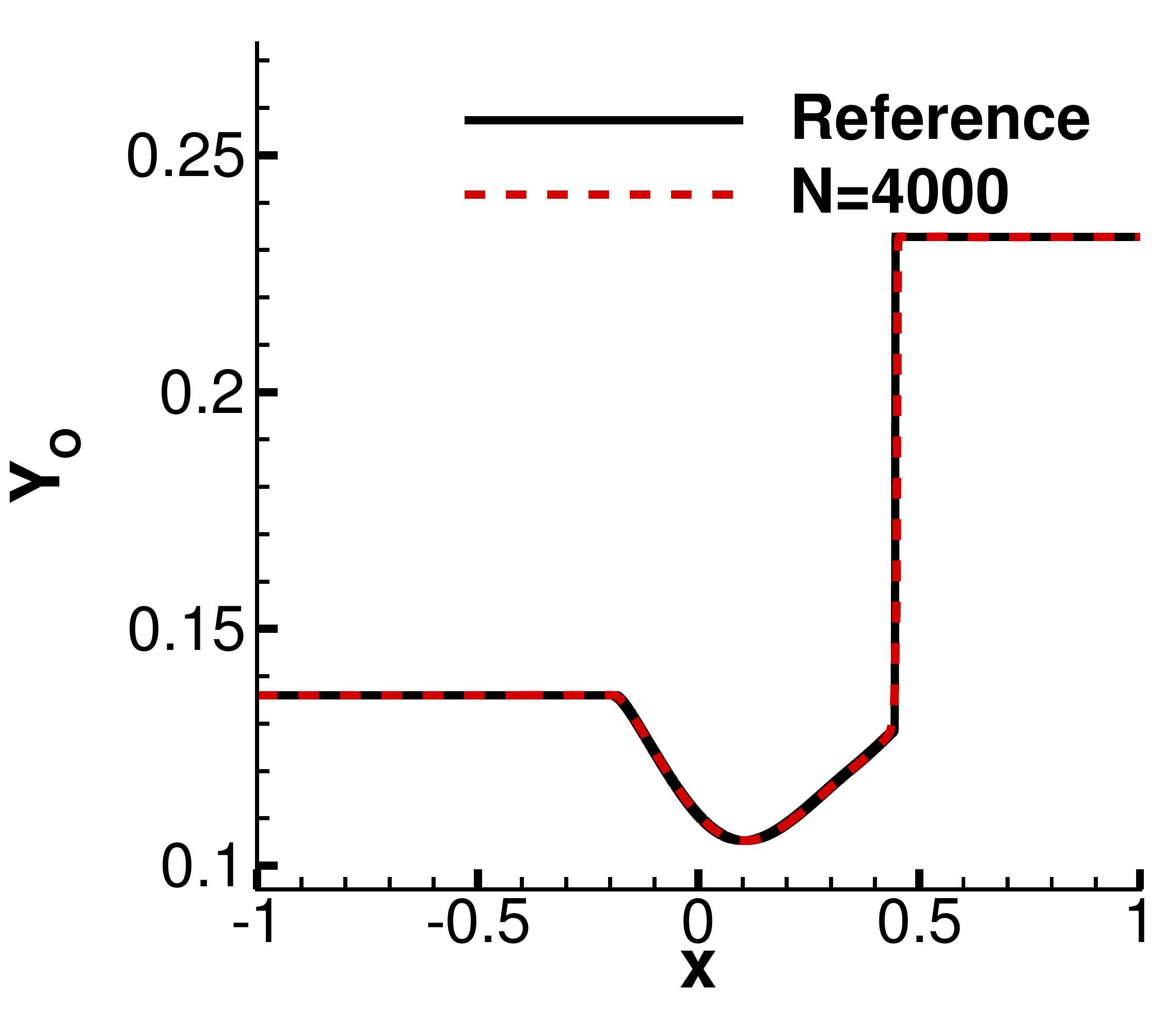}
    \end{subfigure}
    \begin{subfigure}[b]{0.4\textwidth}
        \includegraphics[width=\textwidth]{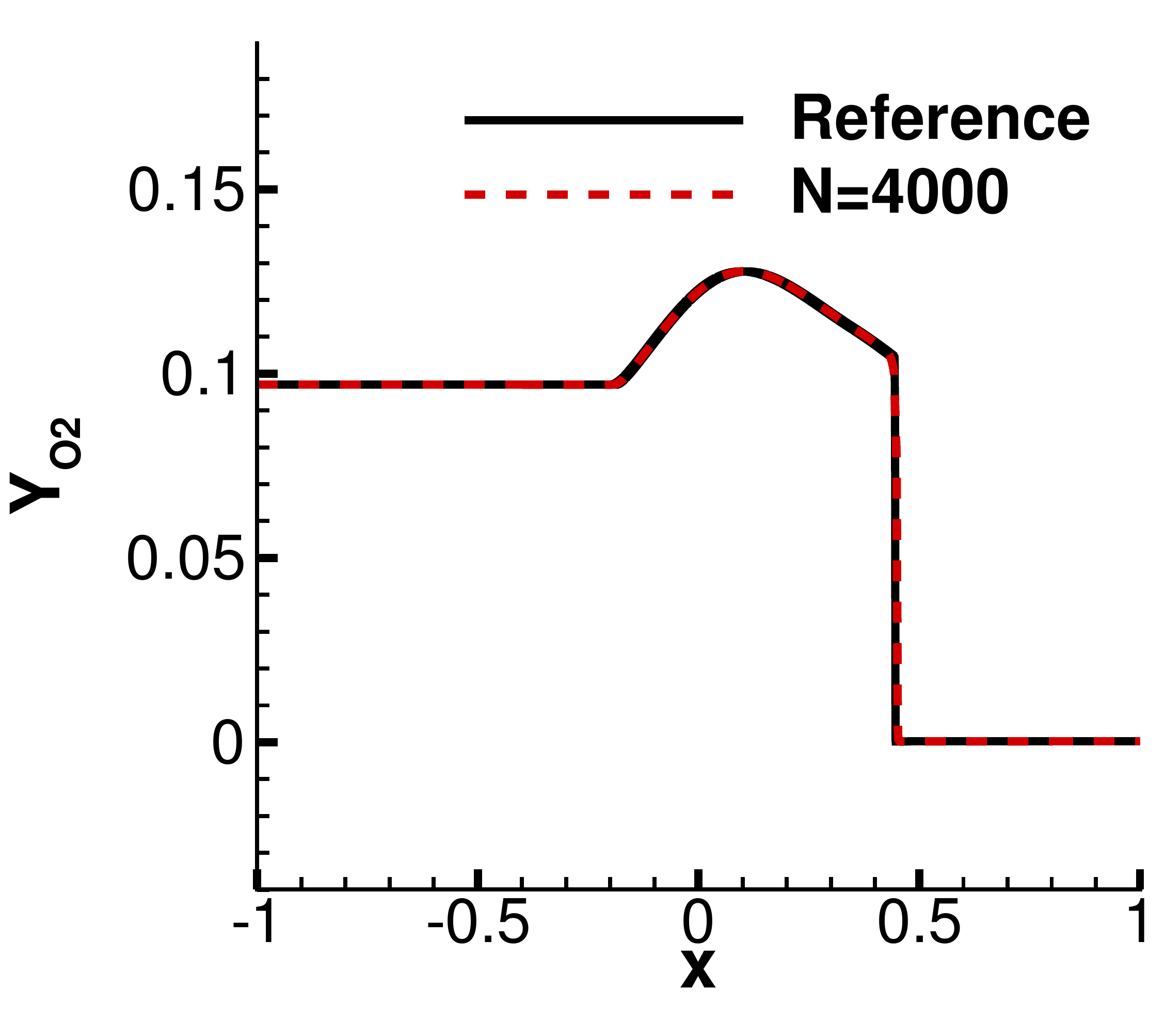}
    \end{subfigure}
    \begin{subfigure}[b]{0.4\textwidth}
        \includegraphics[width=\textwidth]{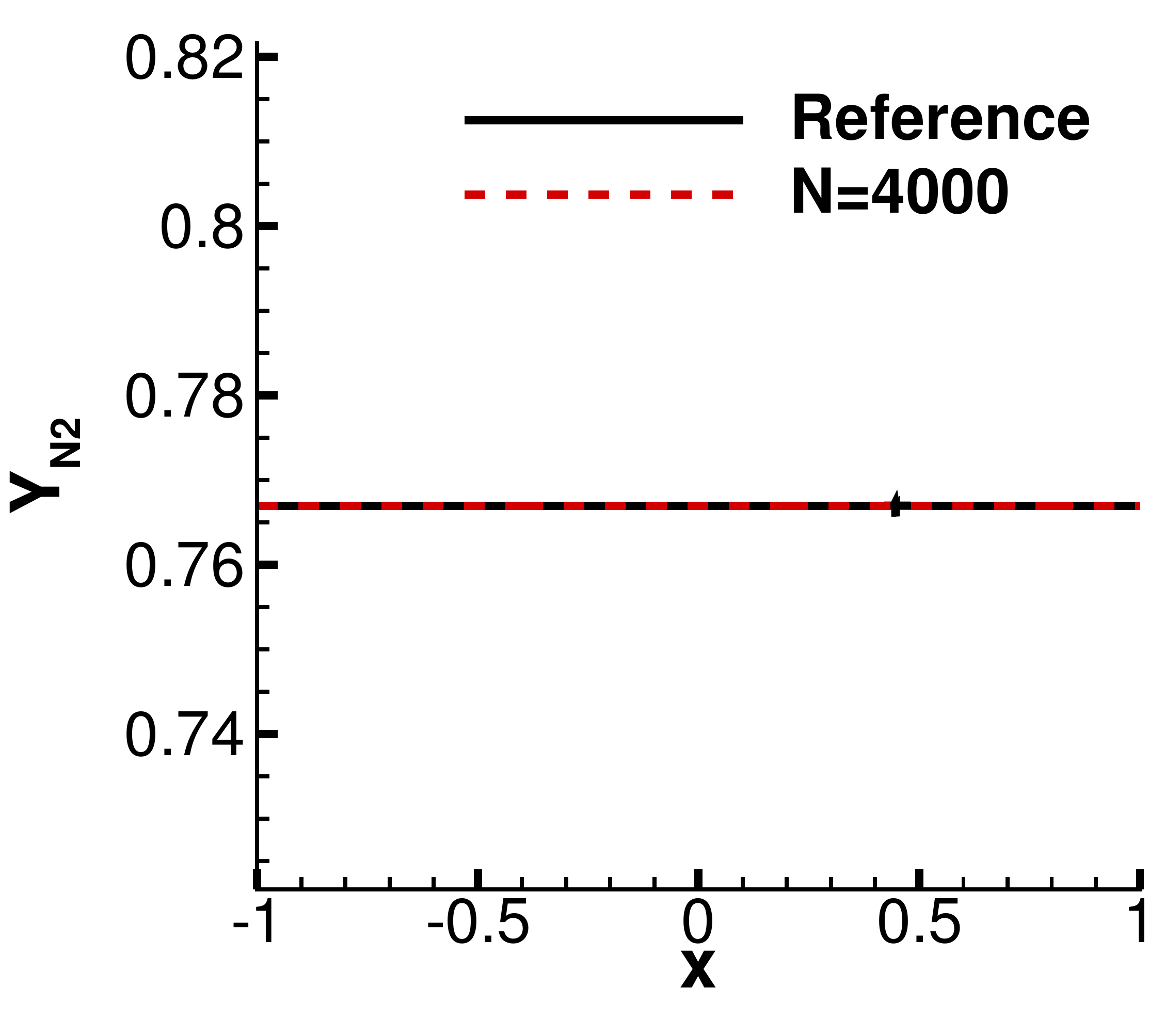}
    \end{subfigure}
\caption{Solutions of 1D reversible reactions of Oxygen at time $t = 0.0001$. CFL = 0.2 and cell number is $N = 4000$. 
The reference values are obtained with grid number of $N=40000$.}
\label{fig:1dOxygen}
\end{figure}
\section{Conclusions\label{sec:conclusion}}
In this study, the second order process based modified Patankar Runge-Kutta (PMPRK) schemes are proposed that maintain both the mass and mole in balance while retaining the positivity of density and pressure unconditionally for Euler equations with non-equilibrium reactive source terms. The sufficient and necessary conditions of the RK and Patankar coefficients are derived to fulfill the prior second order of accuracy. The positivity of the proposed scheme is supported by the rigorous proof under simplified conditions and numerical validations with large randomly generated samples for general conditions. Coupled with the finite volume method, the PMPRK is extended to and tested on the Euler equations with multi-reactions. Work on the schemes with third order of accuracy is ongoing and will be reported in the future.

\section*{Acknowledgment}
This work was supported by the computational resources allocated by the Ohio Supercomputer Center.

\appendix\label{sec:append}

\bibliography{MPRK_second_order}
\bibliographystyle{unsrt}

\end{document}